\def\arxiv{}

\documentclass[mnsc,nonblindrev]{informs3}
\usepackage{amsmath,amsfonts,cancel}
\usepackage{thmtools}
\usepackage{adjustbox}
\usepackage[shortlabels]{enumitem}
\usepackage[table]{xcolor}
\definecolor{edits}{rgb}{0,0,0}
\usepackage{mathtools}
\usepackage{dsfont}
\usepackage{tcolorbox}
\usepackage{booktabs}
\usepackage{bbm}
\usepackage{xspace}
\usepackage{subcaption}
\usepackage{tensor}
\usepackage{bm}
\usepackage{tikz}
\usetikzlibrary{patterns}
\usepackage{natbib}
\usepackage{etoolbox}

\usepackage{placeins}
\usepackage[colorlinks=true,breaklinks=true,bookmarks=true,urlcolor=magenta,citecolor=magenta,linkcolor=magenta,bookmarksopen=false,draft=false]{hyperref}
\usepackage[capitalize]{cleveref}

\usepackage[suppress]{color-edits}
\addauthor{srs}{orange}
\addauthor{vll}{violet}
\addauthor{mq}{red}

\newcommand{\NN}{\mathbb{N}}

\newcommand{\bigO}{\mathcal{O}}

\newcommand{\Ind}[1]{\mathds{1}\left\{ #1 \right\}}

\newcommand{\Exp}[1]{\mathbb E \left[ #1 \right]} 

\renewcommand{\Pr}{\mathbb{P}}

\newcommand{\inv}{{-1}}

\newcommand{\din}{\delta_{\text{in}}}
\newcommand{\dout}{\delta_{\text{out}}}
\newcommand{\GAR}{\mathcal{L}}
\newcommand{\type}{\textsf{type}}
\newcommand{\Xvec}{\mathbf{X}}
\newcommand{\thetavec}{\mathbf{\theta}}
\newcommand{\Dvec}{\mathbf{D}}
\newcommand{\xvec}{\mathbf{x}}

\newcommand{\Yvec}{\mathbf{Y}}
\newcommand{\yvec}{\mathbf{y}}

\newcommand{\rvec}{\mathbf{r}}

\newcommand{\rbar}{\overline{r}}
\newcommand{\V}{\mathcal{V}}
\newcommand{\Ber}{\textsf{Bernoulli}}
\newcommand{\rout}[1]{r_{#1}^{\textsf{out}}}
\newcommand{\A}{\mathbf{A}}
\newcommand{\W}{\mathbf{W}}
\newcommand{\I}{\mathbf{I}}
\newcommand{\M}{\mathbf{M}}
\newcommand{\betavec}{\mathbf{\beta}}
\newcommand{\uvec}{\mathbf{u}}
\newcommand{\mean}{m}
\newcommand{\meanvec}{\mathbf{\mean}}
\newcommand{\pvec}{\mathbf{\pi}}
\newcommand{\oblig}{O}
\newcommand{\tmix}{t_{\textsf{mix}}}
\newcommand{\noblig}{N}
\newcommand{\FJ}{Friedkin-Johnsen\xspace}
\newcommand{\abs}[1]{\left\lvert #1 \right\rvert}
\newcommand{\revenue}{\textsf{Revenue}}
\newcommand{\bond}{\textsf{BondAmt}}
\newcommand{\ctype}{\textsf{Type}}
\newcommand{\out}{\text{out}}
\newcommand{\Deltain}{\Delta_{\text{in}}}

\newcommand{\E}{\mathcal{E}}

\usepackage{changepage}

\usepackage{multirow}

\crefname{assumption}{assumption}{assumptions}
\crefname{algocf}{algorithm}{algorithm}

\mathchardef\mhyphen="2D 
\ifdefined \argmin
\else \DeclareMathOperator*{\argmax}{arg\,min}
\fi
\ifdefined \argmax
\else \DeclareMathOperator*{\argmax}{arg\,max}
\fi
\DeclarePairedDelimiter{\norm}{\lVert}{\rVert}

\let\originalleft\left
\let\originalright\right
\renewcommand{\left}{\mathopen{}\mathclose\bgroup\originalleft}
\renewcommand{\right}{\aftergroup\egroup\originalright}

\usepackage{amsmath,amssymb,amsfonts}

\usepackage[ruled]{algorithm2e} 

\SetAlFnt{\small}
\SetAlCapFnt{\small}
\SetAlCapNameFnt{\small}
\SetAlCapHSkip{0pt}
\IncMargin{-\parindent}

\usepackage{natbib}
 \bibpunct[, ]{(}{)}{,}{a}{}{,}%
 \def\bibfont{\small}%

\usepackage{rotating}
\usepackage{fancyvrb}

\renewenvironment{proof}[1][Proof]{%
  \par\noindent{\itshape #1.} \ignorespaces
}{%
  \hfill\Halmos\par
}

\EquationsNumberedThrough    

\TheoremsNumberedThrough     

\ECRepeatTheorems  %

\MANUSCRIPTNO{MNSC-0001-2024.00}

\begin{document}


\RUNAUTHOR{Broderick et al.}

\RUNTITLE{Network and Risk Analysis of Surety Bonds}

\TITLE{Network and Risk Analysis of Surety Bonds}
%

\ARTICLEAUTHORS{%

\AUTHOR{Tamara Broderick, Ali Jadbabaie, Vanessa Lin, Manuel Quintero, Arnab Sarker}
\AFF{Institute for Data, Systems, and Society,
Massachusetts Institute of Technology,\\ Cambridge, MA 02139, \EMAIL{\{tamarab,jadbabaie,vllin,mquint,arnabs\}@mit.edu}}

\AUTHOR{Sean R. Sinclair\footnote{Contact author}}
\AFF{Department of Industrial Engineering and Management Sciences,
Northwestern University, \\
Evanston, IL 60208, \EMAIL{sean.sinclair@northwestern.edu}}
} 

\ABSTRACT{%
Surety bonds are financial agreements between a contractor (principal) and obligee (project owner) to complete a project.  However, most large-scale projects involve multiple contractors, creating a network and introducing the possibility of incomplete obligations to propagate and result in project failures.  Typical models for risk assessment assume independent failure probabilities within each contractor.  However, we take a network approach, modeling the contractor network as a directed graph where nodes represent contractors and project owners and edges represent contractual obligations with associated financial records.  To understand risk propagation throughout the contractor network, we extend the celebrated Friedkin-Johnsen model and introduce a stochastic process to simulate principal failures across the network.  From a theoretical perspective, we show that under natural monotonicity conditions on the contractor network, incorporating network effects leads to increases in the average risk for the surety organization.  We further use data from a partnering insurance company to validate our findings, estimating an approximately 2\% higher exposure when accounting for network effects.
}%





\maketitle

\section{Introduction}
\label{sec:intro}

Surety bonds are a foundational mechanism in contractual risk management, widely used to guarantee the completion of projects in sectors such as construction, infrastructure, and public works. In a typical surety agreement, the surety company guarantees to the obligee (the project owner) that the principal (a contractor) will fulfill the terms of a bonded contract. If the principal fails to perform, the surety must step in to ensure project completion, often absorbing substantial financial losses in the process~\citep{russell1990surety}. These agreements are not only mandated for public contracts under laws such as the U.S. Miller Act of 1935~\citep{uscode40_3131}, but also play a growing role in private-sector project financing~\citep{wambach2011surety}. Despite their ubiquity, surety bonds remain difficult for contractors to secure, in part due to the unexpected and systemic nature of failures, which can leave insurers liable; for instance, there were over \$21 billion in claims between 1990 and 1997 in the U.S. alone~\citep{wambach2011surety}.

The pricing of surety bonds depends critically on the ability to assess the default risk of individual contractors. A large body of work in finance and insurance focuses on this task, estimating failure probabilities using firm-level covariates such as credit ratings, leverage ratios, and liquidity metrics~\citep{kim2019default}. While such models (often using statistical machine learning) have improved the accuracy of idiosyncratic risk estimation, they share a key limitation: {\em they assume failures occur independently across firms}. This assumption neglects an increasingly salient feature of real-world contracting environments, the presence of network dependencies among contractors, subcontractors, and project owners (see \cref{fig:simple_diagram} for a toy contractor network, later in \cref{sec:experiments} we consider a real-world contractor network with $\sim30,000$ organizations).  
In practice, most large-scale projects involve multi-tiered contractual relationships, where the performance of one firm is contingent on the timely execution of work by others. As highlighted by recent industry reviews~\citep{assuredpartners2024surety, bci2021supplychain}, subcontractor failure is one of the leading causes of bonded losses. For instance, if a plumbing contractor cannot begin work until the electrical subcontractor completes their portion of a build, the default of the latter creates a domino effect. In such environments, {\em risk is not merely a function of a firm’s own characteristics, but also of its position in the broader contractor network}.  Moreover, these cascading effects are pervasive across domains. In supply chains, for example, the bankruptcy of an upstream supplier can cripple downstream production. In collaborative research, a delay in one lab’s work can stall the entire study. In interbank lending networks, financial contagion spreads through credit exposures. Across all these settings, a network-aware perspective is essential to understand and mitigate systemic risk.

\begin{figure}[!t]
\centering
\includegraphics[scale=1]{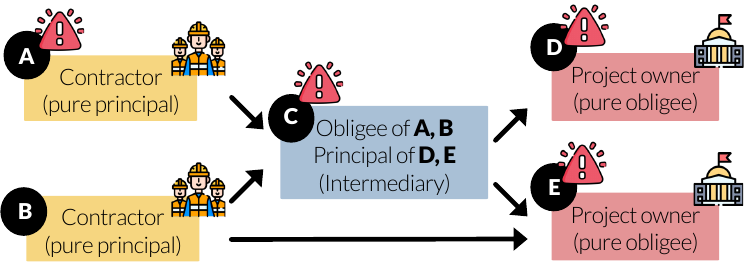}
\caption{Illustrated representation of a series of subcontractor dependencies.  Here we observe that a failure of subcontractor A has the potential to propagate and affect C, and also the obligees D and E.  Note that even though A does not work {\em directly} with D or E, the intermediary C allows them to influence the risk of project incompletion.}
\label{fig:simple_diagram}
\end{figure}

However, despite growing recognition of these interdependencies, there remains a lack of formal models that account for network effects in surety risk assessment. While prior work has examined dynamic models of credit contagion and equilibrium default~\citep{benzoni2015modeling,nickerson2017debt}, these frameworks are often not tailored to the structure of bonded contractor networks, where obligations are directional. Notably, recent theoretical work in stochastic dynamics and mean-field models~\citep{amini2022dynamic,carmona2013mean} explores related ideas, but typically assumes irreversible failures or continuous-time evolution.  Our work is motivated by this important gap. Namely, we seek to answer the following research questions:
\smallskip

\begin{center}
{\em How do contractor relationships influence systemic risk in surety-based contractor networks?  What conditions lead to cascading failures, and how do these propagate over time?  Can we identify key contractors whose failures disproportionately impact financial stability?
}
\end{center}

\subsection{Main Contributions}

\paragraph{Network-Based Model of Risk Propagation.} In an attempt to tackle these questions, one of the main contributions of this work is a network-based modeling framework for risk propagation in surety-based contractor networks. Traditional models like \cite{kim2019default} assume independent failure probabilities, not capturing how failures spread through contractor relationships. Instead, we represent the contracting environment as a directed network $G = (\V, \E)$, where nodes correspond to principals (contractors) and obligees (project owners), and edges capture the potential flow of risk through a contract for bonded work from principal to obligee.  In contrast to network-unaware models of failure dynamics, we introduce a stochastic process $\Xvec^t = (X_i^t)_{i \in \V}$ which represents whether node $i \in \V$ in the network fails at timestep $t \in \mathbb{N}$.  This stochastic process evolves according to the following simple dynamics (see \cref{eq:sp_dynamics}):
\[
    X_i^{t+1} = \Ber\Big((1 - \alpha_i) r_i + \alpha_i \sum_{j \in \din(i)} w_{ij} X_j^{t}\Big).
\]
Here, $r_i$ denotes the contractor's {\em idiosyncratic} or individual risk score, $\alpha_i$ represents the probability that $i$ is affected by one of their neighbors, $\din(i)$ is the set of in-neighbors to $i$, and $w_{ij}$ is the fraction of $i$'s projects that are contracted to principal $j$.  At $t = 0$, we set $X_i^{0} = \Ber(r_i)$ to denote the ``independent'' failure model.  However, as $t$ increases, we see that failures have the ability to propagate and affect their neighbors through the terms $w_{ij}X_j^{t}$.
This framework generalizes standard independent failure models and allows for a more realistic assessment of systemic risk in these interdependent contracting environments, as we demonstrate in \Cref{sec:experiments}.

\paragraph{Mean-Field and Limiting Distribution Analysis.}
To understand how risk propagates in the network, we analyze the stationary distribution of the stochastic failure process $X_i^t$.  We start off by showing that the marginal failure probabilities $m_i^t = \Exp{X_i^t}$ converge to a unique fixed point $m_i$ as their {\em limiting failure probability} (\cref{thm:inverse_fixed_point}).  This result generalizes the \FJ model of opinion dynamics to a setting with heterogeneous bias parameters~\citep{FriedkinJohnsen}.  We further quantify the rate of convergence, establishing an exponential decay in $|m_i^t - m_i|$ governed by the operator norm of a squared weight-adjusted adjacency matrix, and additionally show that the convergence occurs in finite time for acyclic contractor networks (\cref{thm:limiting_probabilities_rate_of_convergence}).

Building on the mean-field analysis, we next study the full joint distribution over failures.  The stochastic process $X_i^t$ defines a Markov chain over $\{0,1\}^n$, where $n$ is the number of nodes (contractors and obligees) in the network.  However, naive analysis will establish its convergence to the stationary distribution in $\bigO(2^n)$ time~\citep{resnick2013adventures}.  In contrast, we show that in acyclic graphs, the convergence occurs in at most $d$ steps, where $d > 0$ is the length of the longest path in the network (\cref{thm:mix_dag}). For general (not necessarily acyclic) networks, we leverage the structure of the stochastic process to show it admits a coupling. This allows us to develop a contraction-based analysis, and show that the rate of convergence scales logarithmically with respect to $n$ (\cref{thm:mixing_time}).
Together, these results characterize the stationary behavior of the failure stochastic process, and show that the stationary distribution can be simulated efficiently.

\paragraph{Structural Insights into Amplification of Systemic Risk.}  
Our modeling framework allows us to quantify how systemic risk is amplified by the network structure beyond what traditional independent risk models predict.  The main insight of our analysis is we show that when obligees hire riskier contractors on average (\Cref{ass:larger_neighbors}), which is often observed in contractor networks \citep{dietz2018mitigating}, the expected failure probabilities (\cref{thm:average_increase}) increase over time.  
 We further outline conditions (\Cref{ass:strictly_larger_neighbors}) such that the total risk and loss in the network is {\em strictly} larger (\cref{thm:gap}).  With this analysis we observe that contractor failures are not uniformly impactful, intermediary nodes, which serve as bridges between multiple principals and obligees, are the primary drivers behind the propagation of project incompletion across the network. In essence, by accounting for network effects, we see that surety organizations may be {\em fundamentally underestimating how risky a contractor network is.}  
Lastly, we define an eigenvalue centrality (\cref{def:page_rank}) that captures the extent to which a contractor’s risk influences the broader network through downstream intermediaries.
To summarize, our results formally characterize when and how network structure exacerbates risk.

\paragraph{Empirical Validation and Risk Estimation.}
We validate our theoretical findings using anonymized real-world surety bond data from a partnering insurance firm. Our empirical analysis shows that accounting for network dependencies leads to a 2\% higher estimated systemic risk compared to traditional models that assume independent contractor failures, and that the distribution of losses exhibits larger right tails, underscoring the potential for more severe extreme events. We further develop a methodology for identifying critical nodes (contractors whose failures have large effects on the network's overall stability) and illustrate this with a detailed case study.
Together, these insights provide methodology and actionable recommendations for insurance providers and policymakers seeking to better anticipate and mitigate systemic risk in contractor networks.

\paragraph{Paper Organization.} We review the related literature in the remainder of this section. We formally present our model in \cref{sec:preliminary}.  In \cref{sec:mean_field} we analyze the mean failure probabilities, establishing the rate of convergence to their limits. Then, in \cref{sec:stochastic_process_analysis} we analyze the mixing time of our stochastic process. Under both sections we provide insights into risk propagation due to network structure. Finally, in \cref{sec:experiments} we complement our theoretical results on real-world data from our partnering surety organization, and conclude in \Cref{sec:conclusion}.  When omitted, all proofs are deferred to the appendix.
\subsection{Related Literature}
\label{sec:related_work}

Our work lies at the intersection of operations research, economics, applied probability, and network analysis, with close connections to models of financial contagion and surety risk. See \citet{caccioli2018network} for a broad survey.

\paragraph{Empirical Studies in Surety Bonds.}
Surety bonds are a widely adopted mechanism to protect against contractor defaults and offer several advantages over traditional insurance \citep{schubert2002point}.  Surety bonds are typically priced using firm-level information such as financial ratios or credit data \citep{schubert2002point, kim2019default}. A large empirical literature applies statistical and machine learning methods, including logit regression~\citep{tserng2014prediction}, SVMs~\citep{tserng2011svm,horta2013company}, ensemble learning~\citep{choi2018predicting}, and Bayesian networks~\citep{cao_2022}, to estimate individual contractors’ default probabilities from accounting data~\citep{barboza2017machine,nguyen2025bankruptcy,shumway2001forecasting,vassalou2004default}. Although some models incorporate macroeconomic covariates~\citep{shumway2001forecasting,vassalou2004default}, all of this literature treats contractors as independent units, with no mechanism by which one contractor’s failure propagates to others.

In practice, however, systemic factors and subcontracting dependencies create correlations in defaults. Historical events such as the 1980s oil embargo led to widespread contractor failures despite strong individual credit profiles~\citep{russell1990surety}, and modern construction projects often hinge on ``lower-tier'' subcontractors, whose disruptions can cascade through the supply chain~\citep{dietz2018mitigating,bci2021supplychain}. Motivated by these limitations, our model augments contractor-level default estimates from the existing literature with network interactions, capturing how contractual ties generate correlated risks and amplify potential losses.

\paragraph{Financial Contagion and Cascade Models.}
Much of the contagion literature studies interbank lending. An early contribution is \citet{allenFinancialContagion2000}, who model contagion as an equilibrium phenomenon in interbank markets, where small liquidity shocks can spread through overlapping claims. Subsequent work often uses threshold models in which a node defaults once losses from its neighbors exceed a threshold \citep{watts2002simple,gai2010contagion,elliott2014financial,acemoglu2015systemic}. These frameworks highlight how localized shocks can spread systemically, but their dynamics are typically deterministic and tied to an initial shock event. More recent analyses consider noisy threshold contagion, where a small probability of below-threshold adoption can accelerate the spread of complex contagion \citep{ecklesLongTiesAccelerate2024}. Unlike threshold contagion models that assume diminishing returns from additional affected neighbors, our model is stochastic and cumulative.  A contractor’s failure risk grows as an aggregate function of weighted neighbor defaults and an idiosyncratic baseline.

Other extensions consider multilayer contagion with mutations, where new strains can emerge as the process spreads and heterogeneity across layers (e.g., schools, workplaces) shapes transmission \citep{SoodSridhar2023}. This underscores how ignoring heterogeneity in either the contagion type or the network structure can miscalculate systemic risk. In our setting, the analogous challenge is to capture heterogeneity in contractor obligations  and recovery, rather than multilayer or mutating contagions.

A further distinction is that most threshold models assume that failed institutions remain insolvent, whereas in surety settings, defaults trigger intervention from the surety organization to ensure project completion. Models with recovery or stochasticity include recovery in reinsurance networks \citep{klages2020cascading}, Gaussian noise in asset values \citep{ramirez2023stochastic}, and dynamic link formation with Cramér–Lundberg premiums/claims \citep{amini2022dynamic}.  However, these approaches rely on explicit thresholds or detailed balance-sheet information uncommon for bonded contractors. Instead, we propose a stochastic cascade model with heterogeneous exposures and explicit recovery, tailored to contractor–surety networks.

\citet{davis2001infectious} introduce an alternative to threshold contagion by modeling ``infectious'' defaults, where bonds fail independently or through Bernoulli contagion within a sector. This framework shares similarities with our approach in that defaults can arise either idiosyncratically or via neighbors, but it assumes uniform exposures and fully connected networks. Moreover, it lacks recovery, which are central in the surety context. Our model builds on this probabilistic contagion idea while incorporating heterogeneity in network structure, firm characteristics, and explicit recovery.

\paragraph{Opinion Dynamics.}

Opinion dynamics provides a natural basis for modeling failure cascades, where contractors correspond to individuals, and default probabilities correspond to opinions. A central feature is neighbor influence, allowing local shocks to generate global effects across the network. For instance, \citet{benzoni2015modeling} show how investors’ beliefs about bond pricing can propagate through financial networks in ways that resemble contagion.  A well-known framework is the Friedkin–Johnsen model, in which agents reconcile their intrinsic beliefs with their neighbors’ views~\citep{FriedkinJohnsen}. Our model has a similar structure where defaults may be triggered by neighbors but are also shaped by inherent contractor-level failure probabilities. In fact, the mean-field of our stochastic process corresponds to a variation of the Friedkin–Johnsen model when opinions are reinterpreted as default probabilities (see \Cref{sec:mean_failure_probabilities}).  Related extensions, such as the interacting Pólya Urn model of \citet{tang2024estimating}, further highlight how noisy observations and social pressure interact with intrinsic beliefs.

\paragraph{Eigenvalue Centralities.}

The Friedkin–Johnsen model is closely connected to eigenvector-based centrality measures such as Bonacich centrality and PageRank. Namely, \citet{bonacich2001eigenvector} highlight that Bonacich centrality can be applied to the same situations as the Friedkin-Johnsen model, with the distinction that the latter is concerned with the limiting state, while Bonacich centrality quantifies the level of influence each node has on the final equilibrium.  This centrality is analogous to the equilibrium output in response to a shock in the input-output model introduced by \citet{leontief1986input} and appears in other settings in the economics literature.  For example, \citet{acemoglu2012network} studies a network of production sectors and shows that the volatility of aggregate output scales with the size of an eigenvector-based centrality vector that is closely related to Bonacich centrality.  A special class of the Friedkin-Johnsen model also coincides with the teleportation model of random surfing introduced by \citet{page1999pagerank,proskurnikov2016pagerank}.

In financial settings, \citet{battiston2012debtrank} propose DebtRank, an eigenvector-based measure of systemic importance in interbank lending. While their measure prevents risk from being transmitted multiple times along cycles, in surety networks repeated impacts of earlier failures are precisely what matter. For this reason, we use an eigenvector-based centrality as a complement to our stochastic model, capturing how individual risks amplify through contractual ties (see \Cref{sec:page_rank}).
\section{Network Model Definition}
\label{sec:preliminary}

\paragraph{Technical notation.} In what follows, for $N \in \mathbb{N}_+$, we let $[N] = \{1, 2, \ldots, N\}$.  For a vector $\xvec$ we use $\norm{\xvec}$ to be its $\ell_\infty$ norm, i.e. $\norm{\xvec} = \max_i |x_i|$, and for a matrix $\A$ we use $\norm{\A}$ to denote the $\ell_\infty$-induced matrix norm, i.e. $\norm{\A} = \sup\{\norm{\A\xvec} \,:\, \norm{\xvec} \leq 1\}$. For two vectors $\mathbf{x}$ and $\mathbf{y}$ we write $\mathbf{x} \geq \mathbf{y}$ to denote the inequality holds entrywise.   See \cref{table:notation} (appendix) for a full table of notation.

\paragraph{Model primitives.} We consider a large-scale network of contracts (edges) between contractors and project owners (nodes). We represent this system as a directed graph $G = (\mathcal{V}, \mathcal{E})$, referred to as the {\bf contractor network}, where each node $i \in \mathcal{V}$ corresponds to an organization (subcontractor), a sub-unit within an organization, a project owner (general contractor), or a collection thereof. Throughout, let \(n = \vert\mathcal{V}\vert\) denote the total number of nodes (equivalently, the number of contractors/project owners) in the graph.  Directed edges $e = (j,i) \in \mathcal{E}$ represent one or more bonded contracts from principal $j$ (the contractor) to obligee $i$ (the project owner). These edges can capture an entire portfolio of contracts issued from $j$ to $i$, or a single agreement. Multiple edges are not permitted, though self-loops and cycles are allowed.\footnote{Self-loops may seem superfluous, but it is often the case that one arm of an organization subcontracts to another arm of the same organization.  Crucially, a self-loop feeds the consequences of a default back into the same contractor one time-step later: node $i$ failing at time $t$ raises the likelihood that $i$ again defaults at $t+1$.} Nodes may act as both principals and obligees, and thus can have both incoming and outgoing edges. For notational convenience, we define the edge direction from $j$ to $i$, indicating the flow of bonded obligations from contractor to project owner; this also aligns with the flow of {\em risk} in the network, which moves in the opposite direction of payment. Each edge $e = (j,i)$ is associated with a weight $w_{ij}$, denoting the total fraction of $i$'s projects that are subcontracted to principal $j$. By construction, the incoming weights for any obligee sum to one. We let $\W$ denote the weighted adjacency matrix, with entries $\W_{ij} = w_{ij}$, and use $\din(i) = \{j \mid (j,i) \in \mathcal{E}\}$ and $\dout(i) = \{k \mid (i,k) \in \mathcal{E}\}$ to denote the incoming and outgoing neighborhoods of node $i$, respectively.

A node $i \in \mathcal{V}$ is said to be a {\bf pure principal} if $\din(i) = 0$. This corresponds to organizations that only act as subcontractors to other obligees, and do not have any bonded work that is deferred to lower tier subcontractors. Similarly, a node $i \in \mathcal{V}$ is said to be a {\bf pure obligee} if $\dout(i) = 0$.  
In practice, pure obligees typically represent project owners such as municipal agencies that contract with a single general contractor or construction manager. Their indegree is usually one, reflecting the primary contractor that organizes the project on their behalf.
Any other nodes $i$ are said to be {\bf intermediaries}. (See \cref{fig:simple_diagram} for a representation of the three classes of contractors.)
If the graph contains only pure principals and pure obligees, it is bipartite and the flow of risk is straightforward to characterize; principals affect only their obligees, and obligees are influenced only by their principals.  However, if the contracting network contains an intermediary, the risk exposure it imposes on its obligees is dependent on its principals, because it relies on principals to complete some of its obligations.  This creates opportunities for risk to flow in unexpected ways, where obligees are affected by principals they do not directly contract with.

As a model for network failure, we assume that each principal $i \in \mathcal{V}$ has an associated {\bf idiosyncratic risk} score $r_i \in (0,1)$.  We interpret $r_i$ as the probability that node $i$ fails independently.  These are assumed to be determined exogenously, based only on individual node level attributes and without any direct knowledge of the network.\footnote{In practice these scores are based on each organization's financial records and hence include some limited network effects. However, we treat these as exogenous inputs into the model.} 
Because pure obligees do not perform bonded work themselves we model their idiosyncratic risk as zero, $r_i=0$.  Any project failure at that level is therefore interpreted as the consequence of downstream contractor defaults rather than an independent failure of the obligee.

We additionally associate with each node a value $\alpha_i \in [0,1]$ corresponding to their {\bf network-associated failure propagation probability}.  We use $\alpha_i$ to denote the probability that a failure of one of node $i$'s neighbors propagates and affects node $i$, essentially a measure of node $i$'s susceptibility to project incompletion by its principals.  Accordingly, we set $\alpha_i=1$ for pure obligees, which by construction means that any failure of their principals directly translates to project incompletion at the obligee node, and $\alpha_i = 0$ for any {pure principals} since none of their work is performed by other principals.  Further, any node which is an {intermediary} has $\alpha_i \in (0,1)$.
 We let $\A$ be the diagonal matrix with entries $\A_{ii} = \alpha_i$ for $i \in \mathcal{V}$.

\begin{remark}
In our model we assume that the entire contractor network is {\em fully observed} by the surety organization. This is unlikely to hold in practice, since each surety organization only observes contracts that their organization bonds. 
In \cref{app:simulations_unobserved_edges} we present methodology to impute these unobserved edges based on observed contracts and organization-level financial records.
\end{remark}

\paragraph{Stochastic Risk Propagation.} We are interested in simulating cascading failures across the contractor network.  In accordance with this goal, we will define a stochastic process $(X_i^t)_{i \in \mathcal{V}, t \in \mathbb{N}}$ to model our {\bf failure dynamics}, where $X_i^t \in \{0,1\}$ will denote the indicator of whether contractor $i$ fails at timestep $t$.  When $i$ is a pure obligee we represent this as the indicator that one of $i$'s project fails. We further denote $\Xvec^t$ to represent the vector $(X_i^t)_{i \in \mathcal{V}}$ of node-level failures at time-step $t$.  This is with slight abuse of notation, since elsewhere we use bold capital letters to denote matrices.  We emphasize that the notion of timestep in this model is primarily used as a vehicle for understanding the stationary failure dynamics.

Initially we assume that each $X_i^0 \sim \Ber(r_i)$, corresponding to each node $i$ failing independently according to their own inherent idiosyncratic risk score.  Since pure obligees have $r_i = 0$, $X_i^0 = 0$ for those nodes.  The dynamics of the stochastic process are:
\begin{equation}
\label{eq:sp_dynamics}
X_i^{t+1} \sim \Ber\Big((1 - \alpha_i) r_i + \alpha_i \sum_{j \in \din(i)} w_{ij} X_j^{t}\Big).
\end{equation}
In $X_i^0$, each node fails independently according to their inherent idiosyncratic failure probability.  Afterwards, conditional on $(X_j^{t})_{j \in \mathcal{V}}$, the failure probability is as follows.  First, $\alpha_i$ denotes the probability that node $i$'s failure is affected by its neighbors.  Hence, $\alpha_i \sum_{j \in \din(i)} w_{ij} X_j^{t}$ is the cumulative risk associated with their in-neighbors (principals) under this event. Otherwise, with probability $(1 - \alpha_i)$ the node fails according to its own idiosyncratic risk score $r_i$.
We further note that this defines a probability distribution over $\{0,1\}^n$ which we denote as $\Pr(\Xvec^t)$.  

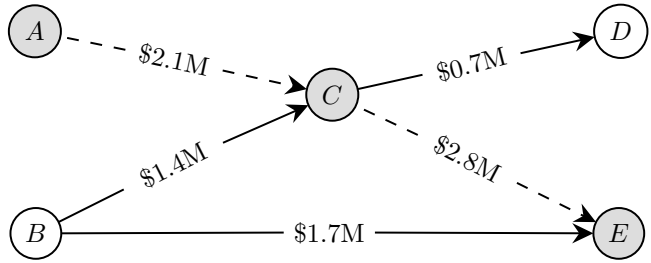
\begin{figure}
    \centering
    \tikzset{every picture/.style={line width=0.75pt}} 
        \begin{tikzpicture}[x=0.75pt,y=0.75pt,yscale=-1,xscale=1]
    \draw (150.69,63.35) node  [font=\small,rotate=-12.26] [align=left] {\$2.1M};
    \draw  [fill={rgb, 255:red, 155; green, 155; blue, 155 }  ,fill opacity=0.34 ]  (80.5, 49) circle [x radius= 13.12, y radius= 13.12]   ;
    \draw (80.5,49) node  [font=\small]  {$A$};
    \draw  [fill={rgb, 255:red, 155; green, 155; blue, 155 }  ,fill opacity=0.34 ]  (375, 151) circle [x radius= 12.81, y radius= 12.81]   ;
    \draw (375,151) node  [font=\small]  {$E$};
    \draw    (81, 151) circle [x radius= 12.81, y radius= 12.81]   ;
    \draw (81,151) node  [font=\small]  {$B$};
    \draw  [fill={rgb, 255:red, 155; green, 155; blue, 155 }  ,fill opacity=0.34 ]  (230.5, 81) circle [x radius= 13.12, y radius= 13.12]   ;
    \draw (230.5,81) node  [font=\small]  {$C$};
    \draw    (376, 49) circle [x radius= 13.45, y radius= 13.45]   ;
    \draw (376,49) node  [font=\small]  {$D$};
    \draw (150.35,116.29) node  [font=\small,rotate=-333.15] [align=left] {\$1.4M};
    \draw (300.72,65.48) node  [font=\small,rotate=-345.28] [align=left] {\$0.7M};
    \draw (298.62,113.1) node  [font=\small,rotate=-27.81] [align=left] {\$2.8M};
    \draw (229,150.5) node  [font=\small] [align=left] {\$1.7M};
    \draw [line width=0.75]  [dash pattern={on 4.5pt off 4.5pt}]  (175.22,68.78) -- (214.75,77.52) ;
    \draw [shift={(217.68,78.17)}, rotate = 192.47] [fill={rgb, 255:red, 0; green, 0; blue, 0 }  ][line width=0.08]  [draw opacity=0] (10.72,-5.15) -- (0,0) -- (10.72,5.15) -- (7.12,0) -- cycle    ;
    \draw [line width=0.75]  [dash pattern={on 4.5pt off 4.5pt}]  (126.15,58.33) -- (93.36,51.63) ;
    \draw [line width=0.75]    (125.33,128.81) -- (92.45,145.27) ;
    \draw [line width=0.75]    (175.38,105.27) -- (215.74,87.5) ;
    \draw [shift={(218.49,86.29)}, rotate = 156.24] [fill={rgb, 255:red, 0; green, 0; blue, 0 }  ][line width=0.08]  [draw opacity=0] (10.72,-5.15) -- (0,0) -- (10.72,5.15) -- (7.12,0) -- cycle    ;
    \draw [line width=0.75]    (206,150.58) -- (93.81,150.96) ;
    \draw [line width=0.75]    (252,150.58) -- (359.19,150.95) ;
    \draw [shift={(362.19,150.96)}, rotate = 180.2] [fill={rgb, 255:red, 0; green, 0; blue, 0 }  ][line width=0.08]  [draw opacity=0] (10.72,-5.15) -- (0,0) -- (10.72,5.15) -- (7.12,0) -- cycle    ;
    \draw [line width=0.75]  [dash pattern={on 4.5pt off 4.5pt}]  (273.61,101.31) -- (242.38,86.6) ;
    \draw [line width=0.75]  [dash pattern={on 4.5pt off 4.5pt}]  (323.63,125.51) -- (360.84,143.97) ;
    \draw [shift={(363.53,145.31)}, rotate = 206.39] [fill={rgb, 255:red, 0; green, 0; blue, 0 }  ][line width=0.08]  [draw opacity=0] (10.72,-5.15) -- (0,0) -- (10.72,5.15) -- (7.12,0) -- cycle    ;
    \draw [line width=0.75]    (325.45,60.06) -- (359.92,52.52) ;
    \draw [shift={(362.85,51.88)}, rotate = 167.65] [fill={rgb, 255:red, 0; green, 0; blue, 0 }  ][line width=0.08]  [draw opacity=0] (10.72,-5.15) -- (0,0) -- (10.72,5.15) -- (7.12,0) -- cycle    ;
    \draw [line width=0.75]    (276,70.94) -- (243.32,78.17) ;
    \end{tikzpicture}
    \caption{Sample contractor network (see \cref{fig:simple_diagram}).  Here we see that contractor $C$ is an obligee for both $A$ and $B$ (with contract value \$2.1M and \$1.4M respectively).  Solid (dashed) edges denote active (failed) obligations; dark-filled nodes are in default; light-filled nodes are solvent.  Hence, our model captures the effect of contractor $A$'s failure on both the intermediary $C$ but also the pure obligee $E$.}
    \label{fig:example}
\end{figure}

\begin{example}
\label{ex:risk_propagation}
Consider a simplified contracting network composed of five organizations (see \Cref{fig:example}).  Each edge in the network represents a bonded contractual obligation, with annotated edges indicating the associated financial exposure.  Edge weights are then calculated as the relative financial exposure to the obligee from each of its principals (e.g. $C$'s exposure to $A$ is $0.6$ since $A$ is responsible for $60\%$ of the work conducted to $C$).
Companies $A$ and $B$ are pure principals, since they do not subcontract work to any other organizations.  Company $C$ is an intermediary, since they are both subcontractors to Companies $D$ and $E$, but obligees to companies $A$ and $B$. Finally, Companies $D$ and $E$ are pure obligees.

This example illustrates how the failure of a pure principal (e.g., $A$ defaults) can propagate through the network in the stochastic process. If $C$ cannot complete its contractual obligations to $E$ because $A$ fails, then $E$ may incur losses, even though it never directly contracted with $A$. Such indirect dependencies are not captured in standard models assuming independent risk, but they are central in our network-aware framework.
This further underscores the role of intermediaries like $C$ in amplifying risk.  Even with moderate idiosyncratic risk levels for $A$ and $B$, the dependency structure means that failures can cascade through the network, elevating systemic risk beyond what node-level scores would suggest.
\end{example}

For each node $i \in \mathcal{V}$ we use $\beta_i$ to denote their {\bf financial loss}, i.e. the amount that the insurer needs to pay out in case of node $i$'s failure.  Inherent to this definition is that $\beta_i = 0$ for all nodes $i$ that are {pure obligees} (since they only receive bonded work).  Lastly, we denote the {\bf global financial loss} as:
\begin{equation}
\label{eq:gar}
\GAR(\Xvec^t) = \sum_{i \in \mathcal{V}} \beta_i X_i^{t} = \betavec^\top \Xvec^{t},
\end{equation}
which denotes the cumulative financial loss associated with all of the nodes in the network if their failures are dictated by $\Xvec^t$ and the financial loss per node is $\beta_i$.  This quantity captures how individual contractor defaults aggregate into broader network-wide loss for the surety organization.

Our main goal in the rest of this work centers on understanding the stochastic process $\Xvec^t$, its asymptotic behavior, and providing insights into how network structure influences systemic risk.  We then leverage this analysis to compare $\GAR(\Xvec^0)$ (the independent failure model) to the stationary behavior of $\GAR(\Xvec^t)$ to quantify the financial effect of cascading failures in surety networks.  We focus on the stationary behavior of the stochastic process primarily to serve as a measure of risk propagation in the network, and leave further studies on the transient behavior of the stochastic process for future work.

\subsection{Discussion of Modeling Assumptions}

We conclude the section with a discussion of our modeling assumptions.

\paragraph{Static Network Structure.} 
In reality, contractor networks are time-varying since edges are dictated by contracts with fixed terms. However, our model assumes a fixed contractor network over time, meaning we do not allow for the entry or exit of contracting organizations, nor do we model the formation or dissolution of contractual ties. This assumption enables a clean equilibrium analysis of systemic risk and allows us to characterize how risk distributes over the network in steady state. That said, incorporating a dynamic network formation remains an important practical direction for future work. In \cref{sec:mixing_time} and \cref{app:time_varying_mixing} we show that our result on the mixing time for the stochastic process applies under time-varying contractor networks.

\paragraph{Exogenous Risk Scores.} We treat each contractor $i$'s idiosyncratic risk score $r_i$ and network sensitivity parameter $\alpha_i$ as exogenously specified inputs to the model. In practice, these parameters are inferred from financial health indicators or historical default data, and likely take into account mild network risk indicators~\citep{kim2019default}. However, regardless of how $r_i$ are estimated, our model allows for the direct incorporation of network effects on risk in contractor networks.

\paragraph{Linear Risk Amplification.} Our stochastic process assumes that a contractor’s risk of failure increases linearly based on the impact of their neighbors via $w_{ij}$ (see \cref{eq:sp_dynamics}). This additive structure simplifies both analysis and simulation, but it may fail to capture important nonlinearities in real-world contagion effects. For instance, a contractor may be robust to isolated failures but vulnerable to risks beyond a certain threshold, such as in the threshold contagion model like \citet{watts2002simple}.

\paragraph{Risk Amplification Proportional to Financial Obligations.} Our work assumes that network-induced risk depends on the proportion of a contractor’s total subcontracted value attributed to each subcontractor, as encoded by the normalized edge weights $w_{ij}$. This formulation reflects a reasonable first-order approximation where risk exposure grows with financial dependence on risky neighbors. Although real-world contracts vary in risk beyond their dollar value, introducing such heterogeneity would significantly complicate the model without materially changing our core theoretical insights.

\paragraph{Inclusion of Pure Obligees.}
We emphasize that pure obligees $i$ in the network have $r_i = 0$, $\beta_i = 0$, and no outgoing edges. Consequently, they do not contribute directly to either the global financial loss $\GAR(\Xvec^t)$ or to risk propagation through the network. In practice, pure obligees typically correspond to project owners such as city municipalities or agencies that contract with a single primary contractor and do not perform bonded work themselves. While such organizations could, in principle, experience project disruptions for idiosyncratic reasons (e.g., funding or scheduling issues), these are exogenous to the surety relationship and thus outside the scope of our model. We nevertheless include pure obligees to quantify the likelihood that project owners receive incomplete work from their principals, and to evaluate how the position of these owners within the network affects their exposure. Their inclusion also enables the computation of our centrality measure, which captures differences in downstream vulnerability across obligees. We return to these points in our numerical simulations (\cref{sec:experiments}).

\section{Mean Field Analysis of Expected Risk}
\label{sec:mean_field}

We start off our analysis by considering the marginal expected failure probabilities of the stochastic process $\Xvec^t$ for each node $i$.  We will later see that this corresponds to a modified \FJ model in the opinion dynamics literature~\citep{FriedkinJohnsen}, and calculate a closed-form expression for the mean failure probabilities.  We also exploit this representation to describe an eigenvector-based centrality measure, assigning scores to each node in the graph corresponding to their risk-based centrality within the contractor network.  We close this section by providing a simple {\em monotonicity} condition under which the mean failure probabilities for each node increase due to network effects.

\subsection{Expected Failure Probabilities}
\label{sec:mean_failure_probabilities}

We start off by analyzing the {\em mean field dynamics} of our stochastic process $\Xvec^t$.  We introduce notation and set $\mean_i^t = \Exp{X_i^t}$ for all $i \in \mathcal{V}$ and $t \in \mathbb{N}$.  All proofs are deferred to \cref{sec:limit_mean_proofs}. Note that $\mean_i^t$ corresponds to the marginal failure probability of node $i$ in step $t$ of the stochastic process.  By definition in \cref{eq:sp_dynamics}, it is easy to see that $\mean_i^t$ satisfies the following recursive equation:

\begin{restatable}{lemma}{MeanFieldRecurrence}
\label{lem:mean_field_recurrence}
For all $i \in \mathcal{V}$ and $t \in \mathbb{N}$ we have that $\mean_i^0 = r_i$ and
\begin{equation}
\label{eq:mean_field_recurrence}
    \mean_i^{t  + 1} = (1 - \alpha_i) r_i + \alpha_i \sum_{j \in \din(i)} w_{ij} \mean_j^{t}.
\end{equation}
Equivalently in matrix notation,  $\meanvec^0 = \rvec$ and \( \meanvec^{t + 1} = (\I - \A) \rvec + \A \W \meanvec^{t}\).
\end{restatable}

\cref{lem:mean_field_recurrence} shows how the mean failure probabilities satisfy a similar dynamics equation to the original stochastic process in \cref{eq:sp_dynamics}.  The obvious next question is whether $\mean_i^t$ converges as $t \rightarrow \infty$, and whether we can characterize the rate of convergence by the underlying contractor network.  In the case that $\mean_i^t$ converges, we let:
\begin{equation}
\label{eq:mean_field_limit}
\mean_i {\triangleq} \lim_{t \rightarrow \infty} \mean_i^t = \lim_{t \rightarrow \infty} \Exp{X_i^t}
\end{equation}
denote the {\em limiting failure probability} of node $i  \in \mathcal{V}$. 

Our first main result of this section shows that the limiting failure probabilities indeed exist and satisfy a fixed point equation.
\begin{restatable}{proposition}{InverseFixedPoint}
\label{thm:inverse_fixed_point}
For any contractor network and any node $i \in \mathcal{V}$, the limiting failure probabilities $\mean_i$ exist and satisfy the following fixed point equation:
\begin{equation}
    \label{eq:mean_field_fixed_point}
    \mean_i = (1 - \alpha_i) r_i + \alpha_i \sum_{j \in \din(i)} w_{ij} \mean_j, \quad\quad \meanvec = (\I - \A) \rvec + \A \W \meanvec.
\end{equation}
Moreover, $(\I - \A\W)$ is invertible and so $\meanvec$ is unique and satisfies:
\begin{equation}
\label{eq:mean_field_solution}
    \meanvec = (\I - \A\W)^\inv(\I - \A)\rvec.    
\end{equation}
\end{restatable}

\begin{figure}
    \centering
    \tikzset{every picture/.style={line width=0.75pt}} 
    \begin{tikzpicture}[x=0.75pt,y=0.75pt,yscale=-1,xscale=1]
\draw (150.69,63.35) node  [font=\small,rotate=-12.26] {0.6};
\draw (150.35,116.29) node  [font=\small,rotate=-333.15] {0.4};
\draw (300.72,65.48) node  [font=\small,rotate=-345.28] {1};
\draw (298.62,113.1) node  [font=\small,rotate=-27.81] {0.62};
\draw (229,150.5) node  [font=\small] {0.38};

\draw (80.5,49) circle [x radius=13.12, y radius=13.12]; 
\draw (80.5,49) node  [font=\small]  {$A$};

\draw (375,151) circle [x radius=12.81, y radius=12.81]; 
\draw (375,151) node  [font=\small]  {$E$};

\draw (81,151) circle [x radius=12.81, y radius=12.81]; 
\draw (81,151) node  [font=\small]  {$B$};

\draw (230.5,81) circle [x radius=13.12, y radius=13.12]; 
\draw (230.5,81) node  [font=\small]  {$C$};

\draw (376,49) circle [x radius=13.45, y radius=13.45]; 
\draw (376,49) node  [font=\small]  {$D$};

\draw [line width=0.75] (175.22,68.78) -- (214.75,77.52);
\draw [shift={(217.68,78.17)}, rotate=192.47] [fill=black][draw opacity=0] (10.72,-5.15) -- (0,0) -- (10.72,5.15) -- (7.12,0) -- cycle;

\draw [line width=0.75] (126.15,58.33) -- (93.36,51.63);
\draw [line width=0.75] (125.33,128.81) -- (92.45,145.27);
\draw [line width=0.75] (175.38,105.27) -- (215.74,87.5);
\draw [shift={(218.49,86.29)}, rotate=156.24] [fill=black][draw opacity=0] (10.72,-5.15) -- (0,0) -- (10.72,5.15) -- (7.12,0) -- cycle;

\draw [line width=0.75] (206,150.58) -- (93.81,150.96);
\draw [line width=0.75] (252,150.58) -- (359.19,150.95);
\draw [shift={(362.19,150.96)}, rotate=180.2] [fill=black][draw opacity=0] (10.72,-5.15) -- (0,0) -- (10.72,5.15) -- (7.12,0) -- cycle;

\draw [line width=0.75] (273.61,101.31) -- (242.38,86.6);
\draw [line width=0.75] (323.63,125.51) -- (360.84,143.97);
\draw [shift={(363.53,145.31)}, rotate=206.39] [fill=black][draw opacity=0] (10.72,-5.15) -- (0,0) -- (10.72,5.15) -- (7.12,0) -- cycle;

\draw [line width=0.75] (325.45,60.06) -- (359.92,52.52);
\draw [shift={(362.85,51.88)}, rotate=167.65] [fill=black][draw opacity=0] (10.72,-5.15) -- (0,0) -- (10.72,5.15) -- (7.12,0) -- cycle;

\draw [line width=0.75] (276,70.94) -- (243.32,78.17);
\end{tikzpicture}
    \caption{\cref{fig:example} but with normalized edge weights. If we set $\rvec = [.2,.1,.05,0,0]$ and $\mathbf{\alpha} = [0,0,0.25,1,1]$ then $\meanvec = [0.2,0.1,0.0775,0.0775,0.08605]$.  Thus, we see that contractor $C$'s risk score increases from $0.05$ to $0.0775$ due to their position within the network.  The pure obligee $D$ gets a risk score equal to its sole subcontractor $C$, while obligee $E$'s risk score is a weighted average of both its principals.}
    \label{fig:example_mean_failure}
\end{figure}
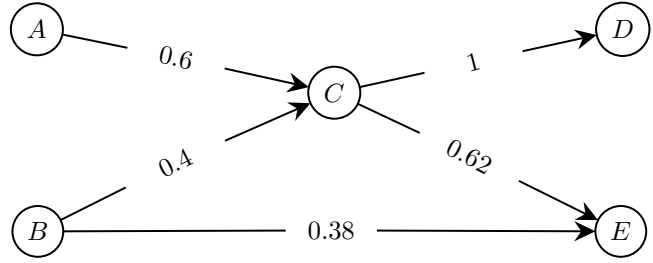

\cref{thm:inverse_fixed_point} establishes a closed-form expression for the limiting failure probabilities $\meanvec$ in terms of the adjacency matrix $\W$, idiosyncratic risk scores $\rvec$, and failure propagation probabilities $\A$.  In \cref{fig:example_mean_failure} we illustrate the average failure probabilities computed over the contractor network from \cref{ex:risk_propagation}.  First note that $\mean_i = r_i$ for any pure principals $i$, since they do not experience any network effects.  However, the intermediary $C$'s mean failure probability {\em increases} when taking into account network structure due to their position in the network (we provide conditions under which this occurs in \cref{sec:marginal_risk_increase}).

The operator $(\I - \A\W)^{-1}$ admits multiple interpretations across related literatures. In Markov chain theory, it parallels the fundamental matrix of absorbing chains, where each entry gives the expected number of visits to a state prior to absorption~\citep{kemeny1969finite}. From a graph-theoretic perspective, it plays a role similar to the pseudoinverse of the graph Laplacian, which encodes mean hitting times, commute times, and effective resistances~\citep{lovasz1993random,doyle1984random,chung1997spectral}. In our context, this highlights that $(\I - \A\W)^{-1}$ aggregates contributions from all higher-order paths in the contractor network, providing an interpretable link between local idiosyncratic risks and their amplified global effects.

We defer the proof of \cref{thm:inverse_fixed_point} to~\cref{sec:limit_mean_proofs}, and here we outline the key technical steps. Our approach begins by establishing in \cref{lem:second_norm_bounded} that the $\ell_\infty$-induced norm satisfies $\norm{(\A\W)^2} < 1$. This implies that the power sequence $(\A\W)^t$ vanishes as $t \to \infty$, ensuring the existence and uniqueness of the limiting mean field. Importantly, while the process may not be contractive in a single iteration (since pure obligees with $\alpha_i = 1$ can only be influenced by their principals), it is always contractive after two iterations. Intuitively, every obligee is connected to at least one intermediary or principal with $\alpha_i < 1$, so risk propagation cannot sustain itself indefinitely. This ``two-step contraction'' property is a structural feature of the surety process.


\begin{restatable}{lemma}{SecondNormBounded}
\label{lem:second_norm_bounded}
Let $\A$ be a diagonal matrix with $\A_{ii}=\alpha_i \text{ for all } i \in \mathcal{V}$ and let $\W$ be the \srsedit{row-substochastic} adjacency matrix of the contractor graph $G$. Then we have $\norm{(\A\W)^2} < 1$ and $(\A\W)^t \to 0$ as $t \to \infty$.
\end{restatable}

{This bound is useful for showing that the Neumann series of $\A\W$ converges, which we formalize in the following corollary:}
\begin{restatable}{corollary}{Neumann}
\label{cor:neumann}
Let $\A$ be a diagonal matrix with $\A_{ii}=\alpha_i \text{ for all } i \in \mathcal{V}$ and let $\W$ be the \srsedit{row-substochastic} adjacency matrix of the contractor graph $G$. Then $(\I-\A\W)$ is invertible, and the Neumann series $\sum_{t=0}^{\infty} (\A\W)^t$ converges to $(\I-\A\W)^{\inv}$.
\end{restatable}

{Consequently, we conclude that the matrix $(\I - \A\W)$ is invertible, and its inverse is given by the Neumann series.  This guarantees both existence and uniqueness of the limiting vector $\meanvec$, which satisfies the fixed-point formula in \cref{eq:mean_field_solution}.  We complete the proof by expanding the recurrence relation in \cref{eq:mean_field_recurrence} and taking limits as $t$ goes to infinity of the mean failure probability vector.}

\cref{thm:inverse_fixed_point} highlights that, for all $i \in \mathcal{V}$, the failure probabilities $\mean_i^t$ converge to their limiting failure probabilities $\mean_i$, and offers a closed form expression for $\mean_i$.  This allows us to calculate the network-adjusted limiting failure probabilities for each node in the network for an arbitrary contractor graph.  In our next result we quantify the rate of convergence.
We establish that in directed acyclic graphs, the convergence occurs in finite time.  In particular, the process converges in just $d$ steps, where $d  > 0$ is the longest path length in the graph.

\begin{restatable}{proposition}{LimitingConvergenceRateAcyclic}
\label{thm:limiting_probabilities_rate_of_convergence}
For any contractor network and $t \geq 2$ we have that:
    \[
    \norm{\meanvec^t - \meanvec} \leq \left( 1 + \frac{2}{1 - \norm{(\A\W)^2}} \right) \norm{\rvec} \norm{(\A\W)^2}^{\lfloor t/2 \rfloor}.
    \]
Moreover, if $G$ is acyclic and we denote by $d > 0$
its maximum path length, then for all $t \geq d$ and all $i \in \mathcal{V}$ we have $\mean_i^t = \mean_i$.
\end{restatable}

\paragraph{Relation to \citet{FriedkinJohnsen}.} 
Our model is closely related to the social influence model under static conditions introduced by \citet{FriedkinJohnsen}, which describes a process in which an individual's opinion at time $t + 1$ is represented as a real-valued linear function of both their innate opinion (determined by exogenous variables) and the opinions of their in-neighbors in the previous time step $t$.
Let $\meanvec^t$ denote the vector of opinions at time $t\geq 0$, where $\rvec=\meanvec^0$ denotes the vector of inherent opinions that each individual holds at time $t=0$.
Then opinions evolve according to the following dynamics for all $t \in \mathbb{N}$:
\begin{equation}\label{eq:fj_dynamic}
    \meanvec^{t+1} = \gamma \rvec + \alpha \W \meanvec^{t},
\end{equation}
where $\W$ is an influence matrix in which $w_{ij}$ corresponds to the extent to which $i$ takes $j$'s opinion into account, and $\alpha,\gamma$ are scalar bias parameters that represent the importance nodes place on the weighted sum of their neighbors' opinions and on their inherent opinion, respectively.
Then if $\alpha<1$ and $\alpha^{-1}$ is not an eigenvalue of $\W$, the process converges to
\[
    \meanvec \triangleq \lim_{t\rightarrow\infty} \meanvec^t = \lim_{t\rightarrow\infty} \left( \I + \alpha\W + \cdots + \alpha^t\W^t \right)\gamma\rvec = (\I-\alpha\W)^{-1}\gamma\rvec.
\]
The stationary distribution then reflects each individual's final opinion as a scalar on some spectrum.

Adapting their model to our context of risk, we interpret a node's {\em opinion} as its probability of failure, which is dependent on failure probabilities of subcontractors according to the transition matrix $\W$ in the same manner that individuals in the social influence model take their neighbors' opinions into account.
Thus in our context we require that ``opinions" take values in $[0,1]$.
Moreover, to ensure that $\meanvec^t$ maintains its interpretation as a vector of probabilities for all $t\in\mathbb{N}$, we also require $\W$ to be \srsedit{row-substochastic} and that $\alpha,\gamma>0$ satisfy $\alpha+\gamma=1$.
Then we can rewrite \Cref{eq:fj_dynamic} as
\(
    \meanvec^{t+1} = (1-\alpha)\rvec + \alpha\W\meanvec^{t}.
\)

Comparing with our dynamics in \Cref{eq:mean_field_recurrence}, the main difference from \citet{FriedkinJohnsen} is that they assume a common bias parameter $\alpha$, whereas our model allows node-specific bias parameters by replacing $\alpha$ and $\gamma$ with diagonal matrices $\A$ and $(\I-\A)$. Note that in contrast to the Friedkin-Johnsen model, we consider a directed network, so to handle nodes $i \in \V$ that lack in- or out-edges we fix their $\alpha_i$ as described in \Cref{sec:preliminary}. More broadly, their framework models the deterministic evolution of average opinions, while ours defines a discrete stochastic process over binary outcomes, yielding the full distribution of failures across the network rather than only mean behavior.

\subsection{Identifying the Impact of Network Structure}
\label{sec:page_rank}

The previous discussion highlights that the mean failure probabilities $\meanvec$ satisfy an equilibrium-type dynamic. This suggests that the underlying risk propagation in the network stabilizes to a fixed distribution of failure probabilities across nodes. Here, we exploit the closed-form definition of the failure probabilities to systematically identify ``risky’’ nodes within the network --- those whose individual or structural properties contribute significantly to overall systemic risk.

From \Cref{thm:inverse_fixed_point} we established that, for all $i \in \mathcal{V}$ the failure probabilities \(\mean_i^t\) converge to their steady-state values \(\mean_i\), which satisfy the fixed-point equation:
\[
\meanvec = (\I-\A\W)^\inv(\I-\A)\rvec.
\]
This expression captures how the individual risk factors \(\rvec\) interact with the network structure, where \(\A\) and \(\W\) determine the interplay between direct risks and dependencies between nodes.

To better understand the network-wide risk contribution, we consider the average limiting failure probability across all nodes, given by (where we denote $\mathbf{1}$ as the vector of all $1$s):
\begin{equation}\label{eqn:gar_x}
\frac{\mathbf{1}^\top}{n} \meanvec = \underbrace{ \frac{\mathbf{1}^\top}{n}(\I-\A\W)^{-1}(\I-\A) }_{\uvec^\top} \rvec,
\end{equation}
where $n$ is the total number of nodes in the network. This expression shows that the overall failure risk can be rewritten as a {\em network-adjusted re-weighting} of the individual idiosyncratic risk scores \(\rvec\). The term $(\I - \A\W)^{-1}$ highlights that risk exposure is not limited to direct neighbors since it aggregates contributions from all paths in the network, with longer paths down-weighted by successive products of exposure probabilities.

\begin{definition}
\label{def:page_rank}
    We set
    \begin{equation}
    \label{eq:network_page_rank}
    \uvec^\top = \frac{\mathbf{1}^\top}{n}(\I-\A\W)^{-1}(\I-\A)
    \end{equation}
    to denote the {\bf risk-based centrality vector}.
\end{definition}
This definition is analogous to traditional PageRank in web search algorithms~\citep{BRIN1998107,page1999pagerank,kleinberg1999authoritative}, where importance is assigned based on structural connectivity.  Here, however, the risk-based centrality captures how individual nodes influence the system-wide failure probability. It accounts for both direct contributions from individual risk levels (through $\A$) and the indirect propagation of risk through network interactions (through $\W$).
Note that by construction, $u_i = 0$ for all pure obligees $i$, i.e. nodes with $\alpha_i = 1$.  This follows from the structure of $\I - \A$, where nodes corresponding to obligees contain only zero entries.  Intuitively, pure obligees absorb risk but do not emit risk, and so their risk-based centrality measure is zero.

More generally, the value of $u_i$ for $i \in \V$ provides an interpretation of how risk is \emph{structurally} amplified within the network.\footnote{Here we are primarily interested in how the \emph{structure} amplifies individual risks, so we focus on the average failure probability over all nodes.  If we instead consider the expected aggregate loss $\GAR(\meanvec)$, we get an alternate centrality measure $\tilde{\uvec}^\top = \betavec^\top(\I-\A\W)^{-1}(\I-\A)$ that weights failure probabilities to describe the expected contribution of each node to aggregate loss.}
Nodes with higher $u_i$ contribute more significantly to the overall risk amplification.  These nodes not only possess inherent risk but also transmit risk to others, affecting the network-wide failure probability.

\subsection{Impact of Risk Propagation on Mean Field}
\label{sec:marginal_risk_increase}
We close out this section with a central concern for contractor networks, the potential for risk amplification due to a contractor's location within the network.  While our model allows for flexible dynamics, it remains unclear how individual risk exposures respond to changes in the network structure.  In particular, we seek to understand whether, for each $i \in \mathcal{V}$, a firm's limiting failure probability $\mean_i$ is amplified (larger than $r_i$) when it interacts with riskier principals.  To facilitate this analysis, we introduce a structural assumption on contracting organizations, which posits that contractors engage with organizations with higher risk. While this may not hold universally, we show that it is a sufficient condition for $\meanvec \geq \rvec$ across the entire network.

\begin{assumption}
\label{ass:larger_neighbors}
    For all intermediaries $i$ \( \in \mathcal{V}\), the average risk of their principals is greater than or equal to their own inherent idiosyncratic risk score, i.e.
    \[
    \sum_{j \in \din(i)} w_{ij} r_j \geq r_i.
    \]
\end{assumption}

While \Cref{ass:larger_neighbors} may appear restrictive, it aligns with patterns observed in real-world contractor networks~\citep{bci2021supplychain}. In practice, principals that subcontract work are often larger and hence less risky, whereas the subcontractors they engage with are typically smaller and more exposed to risk, which is in part due to a less rigorous vetting process \citep{dietz2018mitigating}.  Although this assumption does not hold universally in our empirical setting in \cref{sec:experiments}, we demonstrate that similar results emerge in our application.

We now present our main result for this section (see \cref{app:mean_increase_proofs} for the proof).

\begin{restatable}{theorem}{AssumptionMeanIncrease}
\label{thm:average_increase}
Under \Cref{ass:larger_neighbors}, the mean failure probability vector evolves monotonically: for all \(t\) \( \in \mathbb{N}\),
\[
    \meanvec^t \leq \meanvec^{t+1}, 
    \qquad 
    \Exp{\GAR(\Xvec^t)} \leq \Exp{\GAR(\Xvec^{t+1})}.
\]
Hence,
\[
    \meanvec \geq \rvec, 
    \qquad 
    \Exp{\GAR(\Xvec^\infty)} \geq \Exp{\GAR(\Xvec^0)}.
\]
Moreover, increasing the exposure parameter \(\alpha_i\) of any intermediary \(i\) can only raise the entries of \(\meanvec\):
\[
    \frac{\partial \meanvec}{\partial \alpha_i}
    = (\I-\A\W)^{-1} \frac{\partial \A}{\partial \alpha_i} \left[ \W\meanvec - \rvec \right]
    \;\geq\; \mathbf{0}.
\]
Under the opposite of \Cref{ass:larger_neighbors}, all inequalities hold in the reverse direction \srsedit{for principals and intermediaries}.
\end{restatable}
\begin{remark}
The monotonicity in \cref{thm:average_increase} can be viewed through the lens of potential theory.  Under \Cref{ass:larger_neighbors} the mean failure probabilities evolve as a subharmonic function on the contractor graph, so the limiting vector $\meanvec$ is the smallest subharmonic majorant of the initial risks $\rvec$~\citep{chung1997spectral}. Reversing the assumption yields the superharmonic analogue.
\end{remark}

The first part of \Cref{thm:average_increase} shows that, under \Cref{ass:larger_neighbors}, the mean failure probability for each node $\mean_i$ evolves monotonically over time. Consequently, the expected global average loss $\Exp{\GAR(\Xvec^t)}$ is also monotone with respect to $t$ and converges to the global average risk associated with the stationary process.
The intuition is straightforward: risk propagates through the network because intermediaries rely on principals who, on average, are at least as risky as themselves. Over time, this accumulation amplifies failure probabilities throughout the network.  A central insight here is that under \Cref{ass:larger_neighbors}, {\em incorporating network structure consistently yields a higher expected global loss than models that assume independent failures.}
The second part of \Cref{thm:average_increase} examines monotonicity with respect to the network-associated failure propagation rates $\alpha_i$ for all $i \in \V$. This highlights the robustness of \Cref{thm:average_increase}, since its conclusions hold across a wide range of $\alpha$ values and underscore how systemic risk is shaped jointly by idiosyncratic contractor risk and network position.

Our next assumption strengthens \Cref{ass:larger_neighbors} by requiring that a contractor’s neighbors, on average, have not just a higher inherent risk but a risk level that exceeds the contractor’s by a fixed margin.  We will then show the propagation of the margin in terms of its downstream cumulative risk measures.
\begin{assumption}
\label{ass:strictly_larger_neighbors}
For all intermediaries $i$, the average risk of their neighbors is strictly greater than their own inherent idiosyncratic risk score, i.e. there exists a $\delta > 0$ such that
\[\sum_{j \in \din(i)} w_{ij} r_j - r_i \geq \delta r_i.\]
\end{assumption}
Using this, we highlight the impact of {risk propagation in the network} by presenting lower bounds on the increase in failure probabilities.  Indeed, under \Cref{ass:strictly_larger_neighbors} we can establish the following result (proof deferred to \cref{app:mean_increase_proofs}).

\begin{restatable}{proposition}{MeanFieldGap}\label{thm:gap}
Under \Cref{ass:strictly_larger_neighbors} we have that:
\begin{align*}
    \meanvec - \rvec &\geq \delta (\I-\A\W)^\inv \A \rvec, \\
    \Exp{\GAR(\Xvec^\infty)} - \Exp{\GAR(\Xvec^0)} &\geq \delta \beta^\top (\I-\A\W)^\inv \A \rvec.
\end{align*}
\end{restatable}

This result highlights how the limiting failure probabilities increase in proportion by a factor of $\delta$, highlighting how seemingly modest shifts in network composition can escalate systemic exposure.
\section{Asymptotic Behavior of Stochastic Risk Process}
\label{sec:stochastic_process_analysis}
     \begin{figure}[!t]
        \centering
        \includegraphics[width=\linewidth]{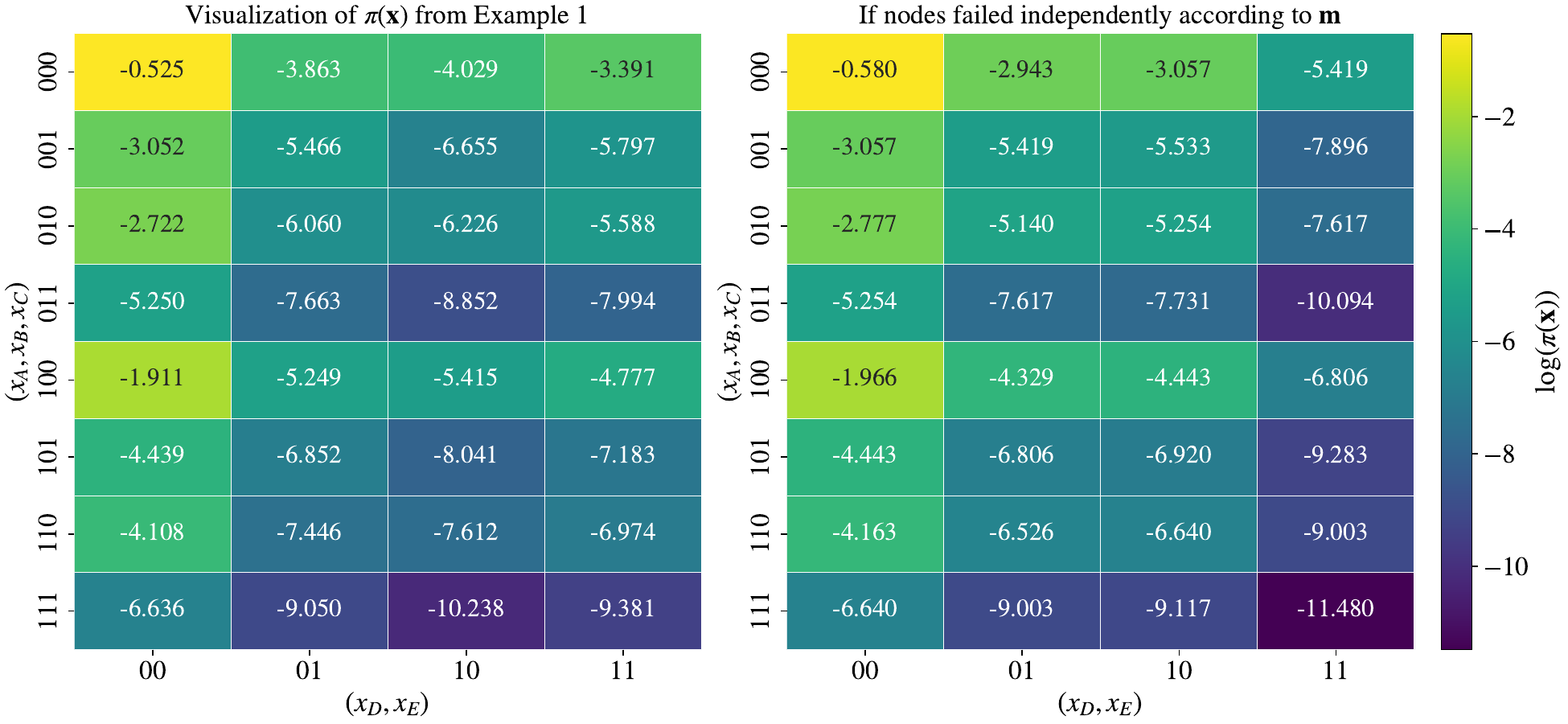}
        \caption{\emph{(Left)} Visualization of the stationary distribution over all possible states $\xvec \in\{0,1\}^5$ in \Cref{ex:risk_propagation} computed from \Cref{thm:mix_dag}.  Rows enumerate the possible states of principals $A,B,C$, while columns enumerate possible states of pure obligees $D,E$.  The heatmap entry in row $(x_A,x_B,x_C)$ and column $(x_D,x_E)$ then gives the log probability of $\xvec=(x_A,x_B,x_C,x_D,x_E)$ in the stationary distribution.
        \emph{(Right)} Visualization of the joint distribution if node failures were instead sampled independently from the mean-field marginals.  We point out that the probabilities of all joint default events in which $D$ and $E$ fail increase when we account for network effects, reflecting how the defaults of $D$ and $E$ become correlated through their shared principal $C$.}
        \label{fig:toy_stationary_dist}
    \end{figure}
    
The previous section analyzed the asymptotic behavior of the marginal average failure probabilities $\mean_i$ for all $i \in \V$.  However, this approach ignores details about the network structure and the main feature that distinguishes this model from independent failure models, the potential for the limiting behavior of nodes to be correlated.  In this section, we turn our attention to the full stochastic process $\Xvec^t$.  We first establish that $\Xvec^t$ is a Markov chain that converges to a unique stationary distribution (proof deferred to \cref{app:stochastic_process_proofs}), and then we quantify the rate of convergence to its stationary distribution in terms of the mixing time.  

\begin{restatable}{lemma}{MarkovChainStationary}
\label{lem:markov_chain_stationary}
$(\Xvec^{t})_{t \in \NN}$ is a Markov chain over the state space $\{0,1\}^n$ that converges to a unique stationary distribution with probability mass function $\pvec(\xvec)=\lim_{t\rightarrow\infty} \Pr(\Xvec^t=\xvec)$ for all $\xvec\in\{0,1\}^n$.
\end{restatable}
\noindent Before continuing, we note that by the continuous mapping theorem and the linearity of $\GAR(\cdot)$, it follows from \Cref{lem:markov_chain_stationary} that the sequence $\GAR(\Xvec^t)$ converges almost surely as $t \to \infty$ to $\GAR(\Xvec^\infty)$.

    The previous lemma establishes that the stochastic process $\Xvec^t$ has a unique stationary distribution, which we will denote as $\pi$.  By analyzing $\pi$, we can understand the limiting behavior of risk propagation in our contractor network. In \cref{fig:toy_stationary_dist} we provide a visualization of the stationary distribution under \cref{ex:risk_propagation}.  We point out that the probabilities of all joint default events in which $D$ and $E$ fail increase when we account for network effects, reflecting how the defaults of $D$ and $E$ become correlated through their shared principal $C$.

    We will abuse notation slightly and let $\Xvec^\infty$ denote a random variable whose distribution is sampled according to $\pi$.  We begin by focusing on quantifying the rate of convergence of the stochastic process $\Xvec^{t}$ for $t \in \NN$ to its stationary distribution $\pvec$, measured in terms of the mixing time:
    \[
    \tmix(\epsilon) = \inf\{t \geq 0: d_{TV}(\pi, \Pr(\Xvec^t)) \leq \epsilon\},
    \]
    where $\varepsilon > 0$ and we treat the distributions as vectors over $\mathbb{R}^{2^n}$.  This directly impacts our ability to simulate network risk efficiently, as a slow-mixing process would require extensive time steps to approximate the long-run behavior accurately. By quantifying the mixing time, we establish bounds on how many iterations are needed before simulations yield reliable estimates of systemic risk (leveraged in our numerical simulations in \cref{sec:experiments}). Additionally, knowing the convergence rate allows us to assess how network structure influences the speed of risk propagation. 

\subsection{Mixing Time of the Stochastic Process}
\label{sec:mixing_time}

Recall in \cref{sec:mean_failure_probabilities} we showed that the convergence rate of the mean failure probabilities in our stochastic model depends on the rate of decay of $\A\W$.  Here, we extend this insight to the full stochastic process, showing that the distribution over network failures also mixes depending on the rate of decay of $\A\W$.  These results are, in a sense, quite surprising.
Classical Markov chain arguments that operate directly on our state space $\{0,1\}^n$ would give bounds on the mixing time that scale with the size of the state space, i.e. exponentially in $n$~\citep{resnick2013adventures}. Such bounds are useless for the networks of interest here. Our analysis instead shows that the mixing time actually increases at a \emph{logarithmic} rate in $n$ (for fixed accuracy $\epsilon$). Moreover, in the special but important case in which the contractor graph $G$ is a directed acyclic graph (DAG), the chain mixes in a \emph{finite} number of steps that equals the depth of the DAG. Moreover, for DAGs the stationary distribution can be written in closed form by propagating probabilities along a topological order; see Fig.~\ref{fig:dag_level_sets}.

 The proof relies on a synchronous coupling of two copies of the process, $(\Xvec^{t},\Yvec^{t})_{t \in \NN}$, driven by the same randomness. A single step of the dynamics need not be contractive in general because of the pure obligees, but we show that the evolution is contractive every two steps (similar to \cref{lem:second_norm_bounded}). Consequently, the decay of the discrepancy can be controlled by $\|(\A\W)^2\|$, which was shown to be strictly less than one in \cref{lem:second_norm_bounded}.  A version of this result also holds under time-varying graphs under the assumption that principals and obligees remain principals and obligees across all time steps.  See \cref{app:time_varying_mixing} for more details.

\begin{restatable}{theorem}{MixingGeneral}
\label{thm:mixing_time}
Let $G$ be an arbitrary contractor network.  Then for all $t \in \mathbb{N}$ we have that,
\begin{equation}
d_{TV}(\pi, \Pr(\Xvec^t)) \leq n \norm{(\A\W)^t}.
\end{equation}
As a result, if $G$ is a directed acyclic graph and \(d > 0\) denotes its maximum path length, then for all $\epsilon > 0$ we have \srsedit{$\tmix(\epsilon) \leq d+1$}.  Similarly, if $G$ is a general graph then
\begin{equation}
\label{eq:tmix_bound}
\tmix(\epsilon) \leq 2 + \frac{2}{1-\|(\A\W)^2\|}\,\log\!\left(\frac{n}{\epsilon}\right).
\end{equation}
\end{restatable}

\begin{proof}
At a high level, we will show for arbitrary starting states $\xvec$ and $\yvec$ that there exists a coupling of the stochastic process $\Xvec^t$ and $\Yvec^t$ (which are initialized at $\xvec$ and $\yvec$ respectively) such that for all $t \in \NN$:
\[
\Exp{\norm{\Xvec^t - \Yvec^t}_1} \leq n \norm{(\A\W)^t}.
\]
As a result, this then implies by letting $\yvec$ follow the law of $\pi$ that:
\begin{equation}
\label{eq:tv_bound_l1_norm}
d_{TV}(\pi, \Pr(\Xvec^t)) \leq \Pr(\Xvec^t \neq \Yvec^t) = \Pr(\norm{\Xvec^t - \Yvec^t}_1 \geq 1) \leq \Exp{\norm{\Xvec^t - \Yvec^t}_1} \leq n \norm{(\A\W)^t},
\end{equation}
as claimed.

Hence, we start by rewriting our stochastic process as $\Xvec^{t+1} = f_{\thetavec^{t+1}}(\Xvec^t)$ where $\thetavec^{t} = (\thetavec_i^{t})_{i \in [N]}$ and each $\theta_i^{t}$ are independent random variables.  We will use $(\thetavec^t)_{t \in \NN}$ to determine our coupling.  Note that we can rewrite our stochastic process as
\[
    X_i^0 = \mathbbm{1} \left[ \theta_i^0 \leq r_i \right], \quad
    X_i^{t+1} = \mathbbm{1} \left[ \theta_i^{t+1} \leq (1 - \alpha_i) r_i + \alpha_i \sum_{j \in \din(i)} w_{ij} X_j^{t} \right],
\]
where \( \theta_i^t \overset{\text{iid}}{\sim} \text{Uniform}[0,1] \). We set $h_i(\xvec) = (1 - \alpha_i)r_i + \alpha_i \sum_{j \in \din(i)} w_{ij} \xvec_j$. 

Now consider the \emph{synchronous coupling} of two copies $(\Xvec^{t},\Yvec^{t})$ of the stochastic process driven by the same randomness (dictated by $(\thetavec^t)_{t \in \NN}$), where $\Yvec^{0} = \yvec \sim\pi$ so that $\Pr(\Yvec^{t} = \cdot)=\pi(\cdot)$ for all $t \in \NN$.
Define for $i=1,\ldots,n$,
\[
D_i^{t} \;\triangleq\; \Exp{\abs{X_i^t - Y_i^t}},
\qquad
\Dvec^{t} \; \triangleq \; \bigl(D_1^{t},\ldots,D_n^{t}\bigr)^{\top},
\]
to be the expected difference between $X_i^t$ and $Y_i^t$.  
We start by showing the following lemma:
\begin{lemma}
\label{lem:one_step_expected_difference}
For all $t \geq 1$ we have that
\(
\Dvec^{t+1} \leq \A\W\Dvec^t.
\)
\end{lemma}
\begin{proof}
Consider an arbitrary index $i \in \mathcal{V}$.  Then by definition of the stochastic process and the coupling, we have that:
\[
X_i^{t+1} = \Ind{\theta_i^{t+1} \leq h_i(\Xvec^t)} \quad\quad Y_i^{t+1} = \Ind{\theta_i^{t+1} \leq h_i(\Yvec^t)}.
\]
Hence, $X_i^{t+1} \neq Y_i^{t+1}$ if and only if $\min\{h_i(\Xvec^t), h_i(\Yvec^t)\} < \theta_i^{t+1} \leq \max\{h_i(\Xvec^t), h_i(\Yvec^t)\}$.  Thus
\begin{align*}
    \Exp{\abs{X_i^{t+1} - Y_i^{t+1}} \mid \Xvec^t, \Yvec^t} & = \Exp{\Ind{\min\{h_i(\Xvec^t), h_i(\Yvec^t)\} < \theta_i^{t+1} \leq \max\{h_i(\Xvec^t), h_i(\Yvec^t)\}} \mid \Xvec^t, \Yvec^t} \\
    & = \Pr(\theta_i^{t+1} \in (\min\{h_i(\Xvec^t), h_i(\Yvec^t)\}, \max\{h_i(\Xvec^t), h_i(\Yvec^t)\}) \mid \Xvec^t, \Yvec^t) \\
    & = \abs{h_i(\Xvec^t) - h_i(\Yvec^t)} \\
    & = \abs{\alpha_i w_i^\top(\Xvec^t - \Yvec^t)} \leq \alpha_i \sum_{j \in \din(i)} w_{ij} |X_j^t - Y_j^t|,
\end{align*}
where the third equality follows from the fact that $\theta_i^{t+1}$ is distributed uniformly. Taking the expectation with respect to $\Xvec^t$ and $\Yvec^t$ over the coupling yields $\Dvec^{t+1} \leq \A\W\Dvec^t$.
\end{proof}

With the previous lemma in hand, we are now ready to show the result.  Indeed, by \cref{eq:tv_bound_l1_norm} we have that:
\begin{align*}
    d_{TV}(\pi, \Pr(\Xvec^t)) & \leq \Exp{\norm{\Xvec^t - \Yvec^t}} = \norm{\Dvec^t}_1 \underset{(\text{C.S.})}{\leq}n \norm{\Dvec^t}_\infty\\
    & \leq n \norm{(\A\W)^t \Dvec^0}_\infty \\
    & \leq n \norm{(\A\W)^t} \norm{\Dvec^0}_{\infty} \leq n \norm{(\A\W)^t}.
\end{align*}

\srsedit{Now we note that if $G$ is a directed acyclic graph, then for any $t > d$ we have $(\A\W)^t = 0$. Hence, we get that $\tmix(\epsilon) \leq d+1$ for any $\epsilon > 0$.}  For the general case, we use the fact that $\norm{(\A\W)^2} < 1$ via \cref{lem:second_norm_bounded}.  Plugging this into the above bound yields:
\begin{equation}
\label{eq:tv_bound_intermed}
    d_{TV}(\pi, \Pr(\Xvec^t)) \leq n (\norm{(\A\W)^2})^{\lfloor t/2 \rfloor}.
\end{equation}

Setting the right hand side of \cref{eq:tv_bound_intermed} $\leq \epsilon$, using the fact that $1 - \gamma \leq - \log(\gamma)$ for $\gamma \in (0,1)$ and solving for $t$ gives the desired bound on $\tmix(\epsilon)$. Indeed, let $\tilde{t} \triangleq \lfloor t/2\rfloor$.  We want to solve $n \norm{(\A\W)^2}^{\tilde{t}} \leq \epsilon$.  Taking the logarithms (and using that $\log(\norm{(\A\W)^2}) < 0$) gives:
\[
    \tilde{t} \ge \frac{\log(\varepsilon/n)}{\log(\norm{(\A\W)^2})}
    \;=\;
    \frac{\log(n/\varepsilon)}{-\log(\norm{(\A\W)^2})}.
\]
Hence it suffices to choose
\[
    \tilde{t} \;:=\; \Bigl\lceil \frac{\log(n/\epsilon)}{-\log(\norm{(\A\W)^2})} \Bigr\rceil,
    \qquad
    t = 2\tilde{t}.
\]
Therefore,
\[
    \;t_{\mathrm{mix}}(\epsilon)
    \;\le\;
    2 \left\lceil \frac{\log(n/\epsilon)}{-\log (\|(\A\W)^2\|)} \right\rceil 
    \;\le\;
    \frac{2}{1-\|(\A\W)^2\|}\,\log\!\left(\frac{n}{\epsilon}\right) + 2,
\]
which trades the logarithm in the denominator for the spectral gap–like term $1-\|(\A\W)^2\|$.
\end{proof}

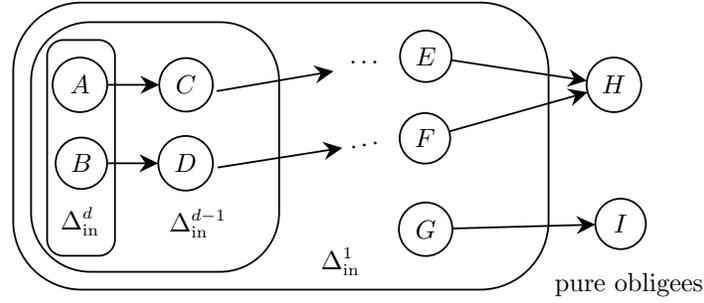
\begin{figure}[!t]
\centering
\scalebox{0.9}{
\tikzset{every picture/.style={line width=0.75pt}} 
\tikzset{every picture/.style={line width=0.75pt}} 

\begin{tikzpicture}[x=0.75pt,y=0.75pt,yscale=-1,xscale=1]

\draw   (10,47.43) .. controls (10,26.21) and (27.21,9) .. (48.43,9) -- (271.57,9) .. controls (292.79,9) and (310,26.21) .. (310,47.43) -- (310,131.57) .. controls (310,152.79) and (292.79,170) .. (271.57,170) -- (48.43,170) .. controls (27.21,170) and (10,152.79) .. (10,131.57) -- cycle ;
\draw   (20,53.18) .. controls (20,34.86) and (34.86,20) .. (53.18,20) -- (125.82,20) .. controls (144.14,20) and (159,34.86) .. (159,53.18) -- (159,126.82) .. controls (159,145.14) and (144.14,160) .. (125.82,160) -- (53.18,160) .. controls (34.86,160) and (20,145.14) .. (20,126.82) -- cycle ;
\draw   (29,39.07) .. controls (29,34.06) and (33.06,30) .. (38.07,30) -- (57.93,30) .. controls (62.94,30) and (67,34.06) .. (67,39.07) -- (67,141.93) .. controls (67,146.94) and (62.94,151) .. (57.93,151) -- (38.07,151) .. controls (33.06,151) and (29,146.94) .. (29,141.93) -- cycle ;
\draw    (124.5,101.5) -- (191.54,90.97) ;
\draw [shift={(194.5,90.5)}, rotate = 171.07] [fill={rgb, 255:red, 0; green, 0; blue, 0 }  ][line width=0.08]  [draw opacity=0] (10.72,-5.15) -- (0,0) -- (10.72,5.15) -- (7.12,0) -- cycle    ;
\draw    (124,58.68) -- (186.54,48) ;
\draw [shift={(189.5,47.5)}, rotate = 170.31] [fill={rgb, 255:red, 0; green, 0; blue, 0 }  ][line width=0.08]  [draw opacity=0] (10.72,-5.15) -- (0,0) -- (10.72,5.15) -- (7.12,0) -- cycle    ;

\draw    (47.38, 55.13) circle [x radius= 15.31, y radius= 15.31]   ;
\draw (47.38,55.13) node   [align=left] {$\displaystyle A$};
\draw    (48.19, 99.13) circle [x radius= 14.42, y radius= 14.42]   ;
\draw (48.19,99.13) node   [align=left] {$\displaystyle B$};
\draw    (106.84, 55.13) circle [x radius= 15, y radius= 15]   ;
\draw (106.84,55.13) node   [align=left] {$\displaystyle C$};
\draw    (241.19, 39.13) circle [x radius= 14.42, y radius= 14.42]   ;
\draw (241.19,39.13) node   [align=left] {$\displaystyle E$};
\draw    (346.61, 55.13) circle [x radius= 15.31, y radius= 15.31]   ;
\draw (346.61,55.13) node   [align=left] {$\displaystyle H$};
\draw    (240.72, 85.13) circle [x radius= 14.15, y radius= 14.15]   ;
\draw (240.72,85.13) node   [align=left] {$\displaystyle F$};
\draw    (241.13, 136.13) circle [x radius= 15, y radius= 15]   ;
\draw (241.13,136.13) node   [align=left] {$\displaystyle G$};
\draw    (350.38, 133.13) circle [x radius= 14.15, y radius= 14.15]   ;
\draw (350.38,133.13) node   [align=left] {$\displaystyle I$};
\draw    (106.61, 99.13) circle [x radius= 15.31, y radius= 15.31]   ;
\draw (106.61,99.13) node   [align=left] {$\displaystyle D$};
\draw (35,122) node [anchor=north west][inner sep=0.75pt]   [align=left] {$\displaystyle \Deltain^{d}$};
\draw (96,123) node [anchor=north west][inner sep=0.75pt]   [align=left] {$\displaystyle \Deltain^{d-1}$};
\draw (181,145) node [anchor=north west][inner sep=0.75pt]   [align=left] {$\displaystyle \Deltain^{1}$};
\draw (312,159) node [anchor=north west][inner sep=0.75pt]   [align=left] {pure obligees};
\draw (197,40) node [anchor=north west][inner sep=0.75pt]   [align=left] {$\displaystyle \dotsc $};
\draw (197,88.4) node [anchor=north west][inner sep=0.75pt]  [rotate=-349.96] [align=left] {$\displaystyle \dotsc $};
\draw    (62.69,55.13) -- (88.84,55.13) ;
\draw [shift={(91.84,55.13)}, rotate = 180] [fill={rgb, 255:red, 0; green, 0; blue, 0 }  ][line width=0.08]  [draw opacity=0] (10.72,-5.15) -- (0,0) -- (10.72,5.15) -- (7.12,0) -- cycle    ;
\draw    (255.45,41.29) -- (328.51,52.38) ;
\draw [shift={(331.47,52.83)}, rotate = 188.63] [fill={rgb, 255:red, 0; green, 0; blue, 0 }  ][line width=0.08]  [draw opacity=0] (10.72,-5.15) -- (0,0) -- (10.72,5.15) -- (7.12,0) -- cycle    ;
\draw    (254.34,81.27) -- (328.99,60.12) ;
\draw [shift={(331.88,59.3)}, rotate = 164.18] [fill={rgb, 255:red, 0; green, 0; blue, 0 }  ][line width=0.08]  [draw opacity=0] (10.72,-5.15) -- (0,0) -- (10.72,5.15) -- (7.12,0) -- cycle    ;
\draw    (256.12,135.71) -- (333.23,133.6) ;
\draw [shift={(336.23,133.51)}, rotate = 178.43] [fill={rgb, 255:red, 0; green, 0; blue, 0 }  ][line width=0.08]  [draw opacity=0] (10.72,-5.15) -- (0,0) -- (10.72,5.15) -- (7.12,0) -- cycle    ;
\draw    (62.61,99.13) -- (88.3,99.13) ;
\draw [shift={(91.3,99.13)}, rotate = 180] [fill={rgb, 255:red, 0; green, 0; blue, 0 }  ][line width=0.08]  [draw opacity=0] (10.72,-5.15) -- (0,0) -- (10.72,5.15) -- (7.12,0) -- cycle    ;

\end{tikzpicture}}
\caption{
Representation of the ``levels" $\Delta^k_\text{in}$ in an acyclic graph such that $\Delta^1_\text{in}\supset\cdots\supset\Delta^d_\text{in}$.
Nodes $A$ and $B$ each have one $d$-length path to $H$, so they belong to $\Delta^d_\text{in}$.  They also have an edge to the next node in paths $(A,C,\dots,H)$ and $(B,D,\dots,H)$, so they are also in $\Delta^1_\text{in}$. $E$ and $F$ only have paths of length 1, so they are only contained in $\Delta^1_\text{in}$.
Note that pure principals can belong to any level, not just $\Delta^d_\text{in}$ (e.g. $G$).
This represents that the states of obligees in $\Delta^k_\text{in}$ at any time $t$ depend only on the previous states of principals in $\Delta^{k+1}_\text{in}$ at $t-1$.
}
\label{fig:dag_level_sets}
\end{figure}

\srsedit{We close out our discussion here with the case when the graph is a DAG.  Because there are no directed cycles, the linear operator $(\A\W)$ is nilpotent with index at most the depth $d+1$, and therefore the coupled chains coalesce after at most $d+1$ steps. This immediately implies that $\tmix(\epsilon) \leq d+1$ for all $\epsilon>0$, as described in \cref{thm:mixing_time}.  In the same spirit, the stationary distribution is computable in closed form by propagating probabilities layer by layer along a (reverse) topological ordering of the contractor graph $G$.}

In particular, we leverage the Markov property and observe the transition to any time $t \geq 1$ depends only on the state of in-neighbors in the previous timestep.  Thus, we can ignore the states of pure obligees at time $t-1$ because they are not in-neighbors to any other nodes, and we only need to consider the set of nodes that act as principals. To formalize this, we introduce the following definition.

\begin{definition}[In-neighbor layers]
\label{def:in-layers}
For $k=1$, define
\[
\Deltain^1 \triangleq \bigcup_{i \in \mathcal V} \din(i),
\]
to be the set of nodes that act as principals of another node.  
For $k>1$, define recursively
\[
\Deltain^k \triangleq \bigcup_{i \in \Deltain^{k-1}} \din(i),
\]
which equivalently consists of all nodes $i$ such that there exists a directed path of length $k$ from $i$ to some node in $\mathcal V$.
\end{definition}

With this notation, the states of nodes in $\Deltain^1$ at time $t$ depend only on the states of nodes in $\Deltain^2$ at time $t-1$, and so on. This reasoning can be applied recursively until we reach time $t-d$, where $d$ is the maximum tree depth of the graph. The set $\Deltain^d$ contains only pure principals $i$, whose failures occur independently and according to fixed $r_i$ for all time steps. Thus, we do not need to consider earlier times past $t-d$. In other words, we only need to consider a finite number of earlier time steps to obtain an exact form of the full distribution. \Cref{fig:dag_level_sets} gives a visualization of these topological layers.  Formalizing this we have the following representation for the distribution of the stochastic process.
\begin{restatable}{theorem}{MixingDAG}
\label{thm:mix_dag}
    If $G$ is a directed acyclic graph, then for any $t \geq d$ where $d > 0$ is the maximum tree depth of the graph, and using the in-neighbor layers $\Deltain^k$ from \cref{def:in-layers}, the stationary distribution admits the following closed form: 
    \begin{equation*}
        \Pr(\Xvec^t=\xvec)
        = \sum_{\xvec^{t-1}_{\Deltain^1}} \cdots \sum_{\xvec^{t-d}_{\Deltain^d}}
\Pr(\Xvec^t=\xvec \mid \Xvec^{t-1}_{\Deltain^1}=\xvec^{t-1}_{\Deltain^1})
\prod_{k=1}^{d-1} \Pr(\Xvec^{t-k}_{\Deltain^k}=\xvec^{t-k}_{\Deltain^k} \mid \Xvec^{t-k-1}_{\Deltain^{k+1}}=\xvec^{t-k-1}_{\Deltain^{k+1}})
\Pr(\Xvec^{t-d}_{\Deltain^d}=\xvec_{\Deltain^d}^{t-d}).
    \end{equation*}
\end{restatable}

    While the full derivation is technical and omitted from the discussion here (see \cref{app:mixing_time}), we give an explicit representation of the stationary distribution for the Markov chain when the underlying contractor graph is acyclic.
    Then for size $n$ acyclic graphs with directed diameter $d$, computing the exact stationary distribution can be done in $\bigO(d\cdot2^n)$ time.
    In most cases, this computation can actually be done in much less time, because the sets $\Deltain^k$ exclude pure obligees and are strictly decreasing in size.

\section{Numerical Results}
\label{sec:experiments}

To complement our theoretical guarantees, we conduct extensive computational experiments to evaluate risk propagation in contractor networks. Our experiments are based on a contractor network constructed from empirical data provided by a partnering surety organization. Across all simulations, we evaluate the impact of network effects on global loss and the joint distribution of possible failure events.\ifdefined\arxiv\footnote{See \url{https://github.com/seanrsinclair/Network-Risk-Analysis-Surety-Bounds} for the code base.}\else\fi  The main questions we seek to answer through our experiments are:
\begin{itemize}
    \item {\em Structure of real-world surety networks} (\cref{sec:sims_network_description}): We begin with a descriptive overview of an empirical surety network, highlighting its scale, connectivity, and summary statistics of the loss values $\beta_i$ and idiosyncratic risk scores $r_i$ across nodes $i \in \V$ in the network.
    \item {\em Case study} (\cref{sec:experiment_case_study}): We then present a detailed case study on a representative intermediary, illustrating how local network position affects systemic risk contributions.
    \item {\em Conditions for higher expected loss and tail behavior} (\cref{sec:sims_global_loss}): We analyze when expected losses are amplified and derive empirical tail bounds, emphasizing the increase in global average risk even in networks where \Cref{ass:larger_neighbors} fails.
    \item {\em Robustness to network exposure $\alpha_i$} (\cref{sec:simulations_robustness_alpha}): Finally, we test the robustness of our findings under alternative specifications of the exposure parameters $\alpha_i$ for intermediaries $i$ in the network.
\end{itemize}

\paragraph{Network Construction} We build an anonymized contractor network from empirical surety bond data of a partnering surety organization, preserving key structural and statistical properties while protecting sensitive information. Original node identities are replaced with generic indices; contract values, risk scores, and loss amounts are rescaled and perturbed with Laplace noise. The network topology is reconstructed via an edge–rewiring procedure that retains degree distributions, node roles (pure principals, intermediaries, pure obligees), and depth in the hierarchy. Edge weights are recalibrated from the anonymized bond amounts to ensure each obligee’s in-degree sums to one. This produces a synthetic but structurally faithful replica of the real network for simulation and analysis.  See~\cref{app:simulations_network_construction} for further details.

\paragraph{Accounting for Unobserved Edges} Since the surety organization only observes bonded contracts with known principals, some obligees may have additional, unobserved contracting activity. We detect such cases by comparing an obligee’s reported revenue to the total value of observed bonded work; any excess implies unobserved principals. To incorporate the potential risk from these missing relationships, we introduce a synthetic “dummy” principal connected to the obligee with weight equal to the fraction of revenue not explained by observed principals. The dummy’s baseline risk is estimated from the observed mix of contractor types the obligee engages with, assuming the same type distribution holds for unobserved contractors. This approach preserves network structure while accounting for external risk exposure not directly visible in the data.  Further details on this methodology are included in \cref{app:simulations_unobserved_edges}.

\paragraph{Experimental Setup.} We approximate the stationary distribution $\pi$ of the stochastic process $\Xvec^t$ via Monte-Carlo simulation.  We set $\alpha_i = 0.25$ for all intermediaries $i$ in the network, as decided during discussion from our partnering organization. However, later in \cref{sec:simulations_robustness_alpha} we test the robustness of our results to this selection.  Since our contractor graph is a directed acyclic graph with maximum depth $d = 7$, the chain mixes in finite time (see \cref{thm:mixing_time}).  Thus, we report the empirical distribution of $\Xvec^7$ after $t = 7$ time steps over $100,000$ times to form empirical estimates.  All metrics are reported as the average over these replications, and confidence intervals are computed with a significance level of $\delta = 0.05$ when presented.

\subsection{Description of Surety Network}
\label{sec:sims_network_description}

\begin{figure}
    \centering
    \includegraphics[width = 0.7\textwidth]{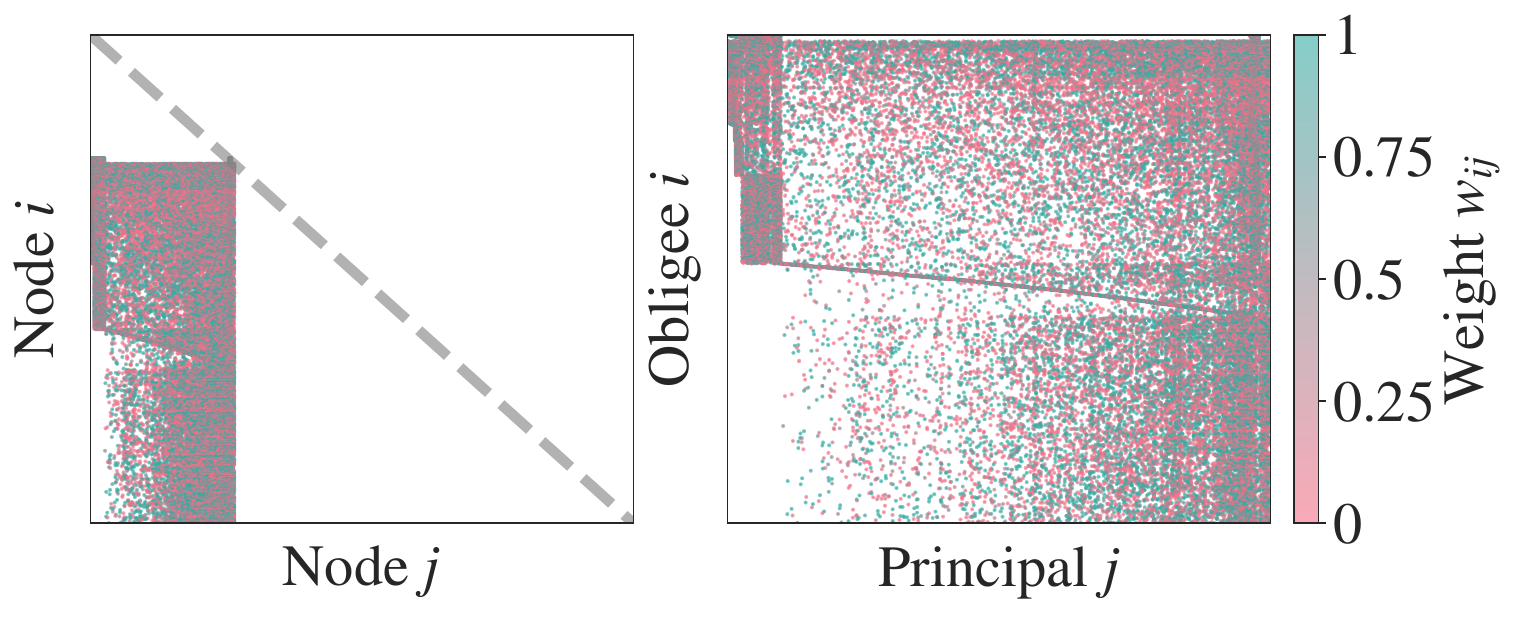}
    \caption{Visualization of the giant component in the anonymized empirical network.  The left figure corresponds to the adjacency matrix of the weakly connected component of the network, where the $x$ and $y$ axes corresponds to different node indices. In the right figure we show the sub adjacency matrix between principals and obligees.  Colors correspond to the edge weights $w_{ij}$.}
    \label{fig:graph_visualization}
\end{figure}

We start off our discussion by examining the contractor network obtained from our partnering surety organization.  This network represents the contractual obligations between contractors and project owners insured by the surety, with each edge corresponding to a surety bond over a one-year period in 2018.

The full contractor graph contains 40,457 nodes (contracting organizations).  The majority of these nodes belong to a single (weakly-connected) component, which accounts for roughly 87.7\% of the graph and contains 35,483 nodes.  We focus on this weakly connected component for the remainder of our simulations.  Within this component, there are 8,984 pure principals (contractors who never act as obligees), making up roughly 25\% of the nodes; 26,137 are pure obligees (project owners who never act as principals), accounting for 74\%; and the remaining 362 nodes (about 1\%) are intermediaries that appear as both principals and obligees in different contracts.  This composition reflects the predominantly bipartite nature of the network, with only a small fraction of nodes serving the dual roles of both principal and obligee.  However, we will see that these intermediaries are the key vehicle for risk to propagate through the network.

\begin{figure}
    \centering
    \includegraphics[scale=0.35]{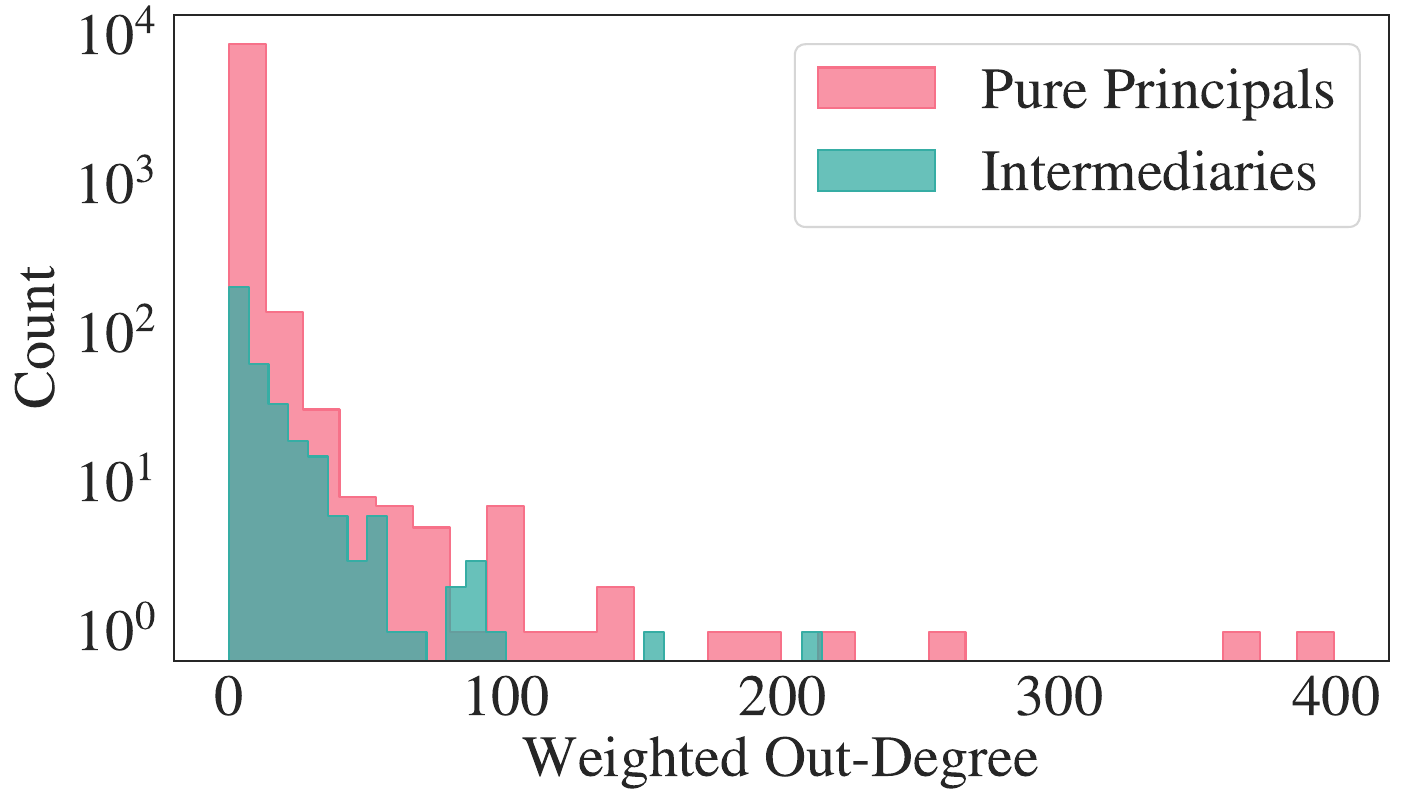}
    \caption{Histogram of the weighted out degree distribution for pure principals and intermediaries.  Counts are shown on a logarithmic scale.  Pure obligees are not included since their out-degree is zero by design.}
    \label{fig:degree_distribution}
\end{figure}

See~\cref{fig:graph_visualization} for a visualization of the adjacency matrix of our surety network.  We emphasize that the graph is mostly bipartite, with the exception of a small number of intermediaries. We further note that the connectivity pattern in the graph is sparse: among the 35,483 nodes in the giant connected component the edge density is around 0.009\%, representing 56,707 contracts.  The nature of the graph also reveals interesting structural properties, it is acyclic with a directed diameter of seven.  This acyclic structure has important implications for our modeling since our stochastic failure-propagation process converges to its stationary distribution in at most seven time steps (see \Cref{thm:mix_dag}).  Lastly, we note some degree of heterogeneity in degree distribution across nodes, witnessed in \cref{fig:degree_distribution}.

In addition to network topology, we measure the idiosyncratic risk scores $r_i$ and loss-given-default values $\beta_i$ for each principal $i$.
These features, shown in \Cref{fig:node_feature_dists}, are heavy-tailed and right-skewed, with notable heterogeneity across node types.  In particular, we observe that the $\beta_i$ values for intermediaries are, on average, larger than that of pure principals.  However, for the risk scores $r_i$, pure principals have larger idiosyncratic risk.  We emphasize that pure obligees are excluded from these plots because their risk and loss values are set to zero by definition.
The loss-given-default distribution also contains a single extreme outlier of approximately \$31 billion; while rare, such high values can occur in large infrastructure projects underwritten by the surety organization.
Later in our analysis, we will explore how these attributes vary between pure principals and intermediaries, and how their heterogeneity influences network-wide risk propagation.

\begin{figure}
    \centering
    \includegraphics[width=0.49\linewidth]{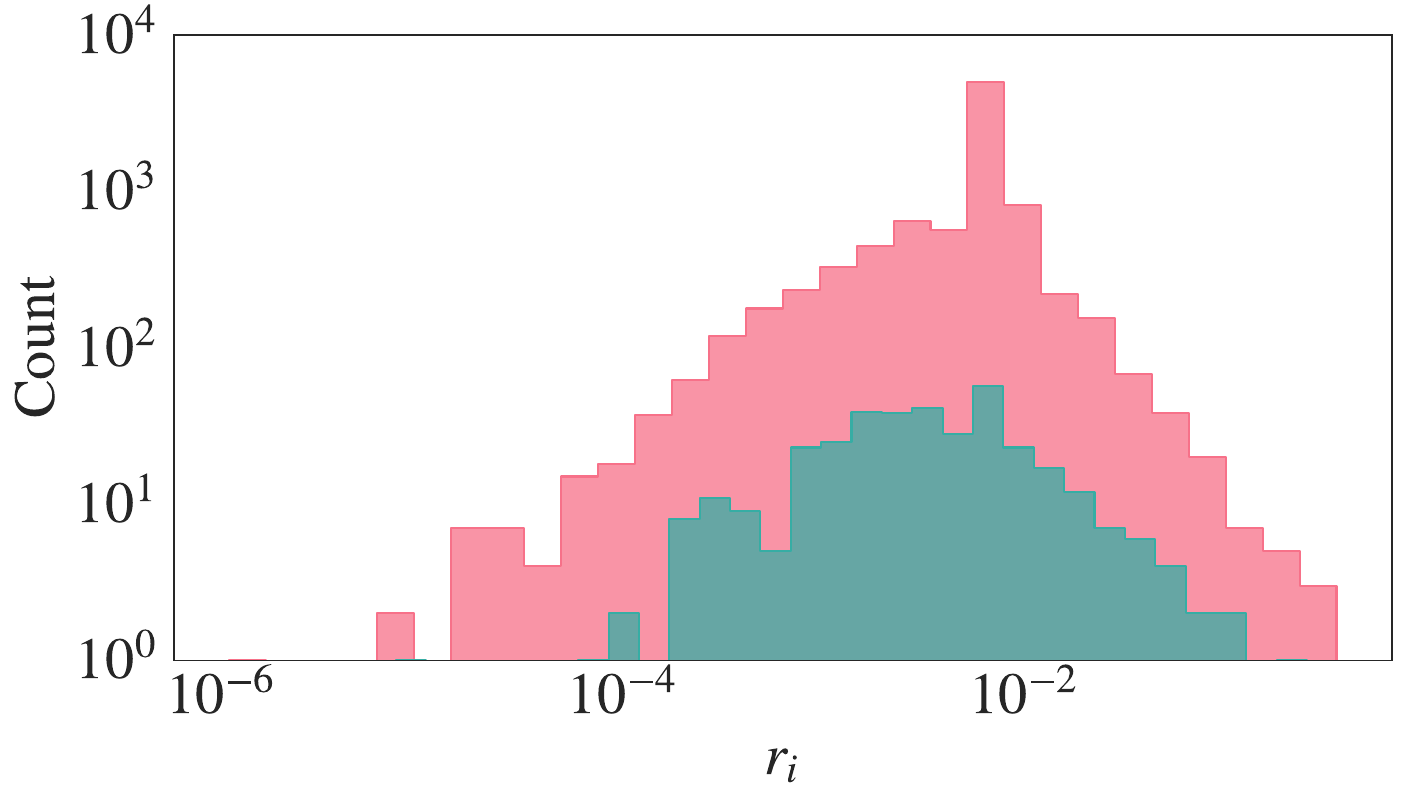}
    \includegraphics[width=0.49\linewidth]{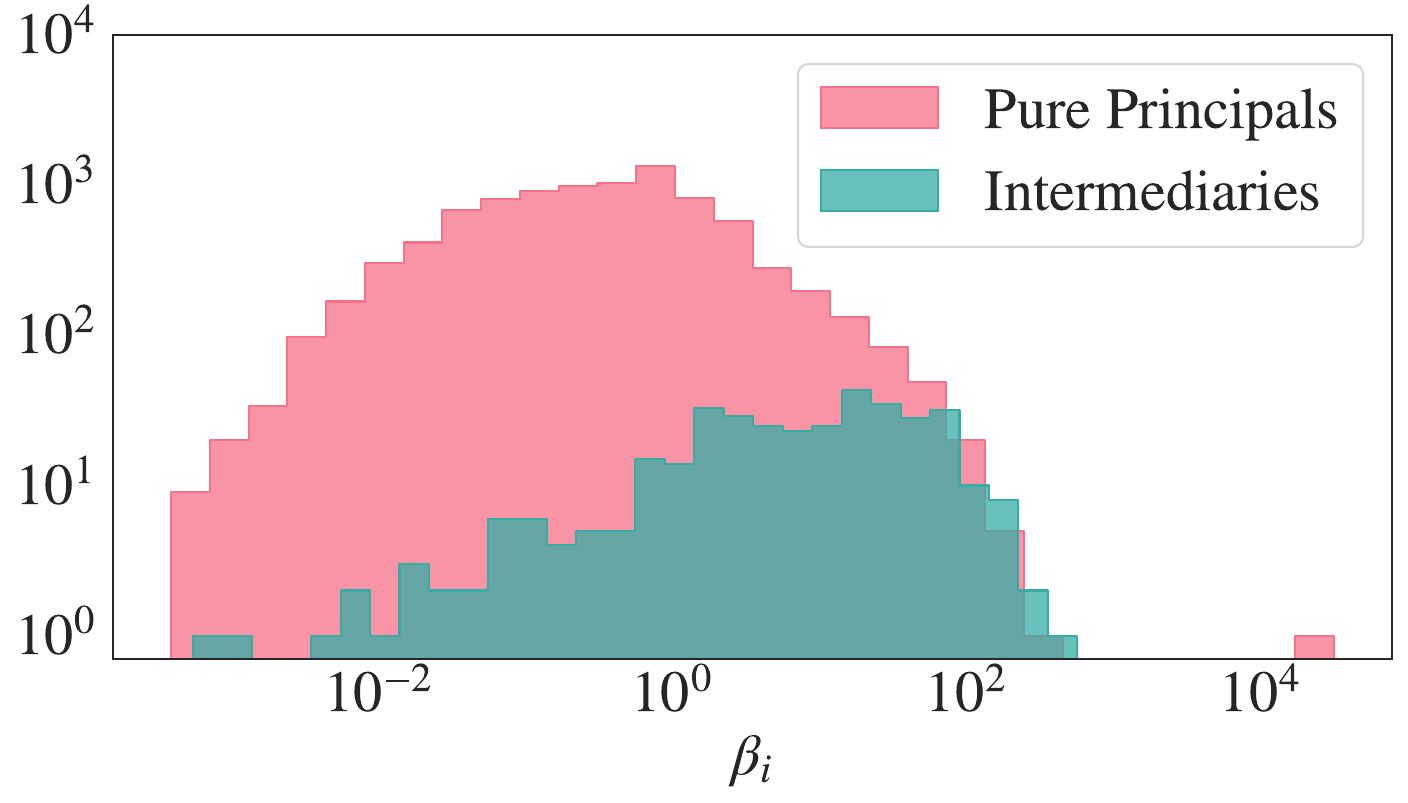}
    \caption{Histograms of idiosyncratic default probabilities $r_i$ (left) and loss given default values $\beta_i$ in millions of USD (right).  Both counts and values are shown on a logarithmic scale.  Pure obligees are excluded here because their values for both attributes are set to zero by definition.}
    \label{fig:node_feature_dists}
\end{figure}

\subsection{Case Study}
\label{sec:experiment_case_study}

\begin{figure}
    \centering
    \begin{minipage}[b]{0.49\linewidth}
        \centering
        \includegraphics[width=\linewidth]{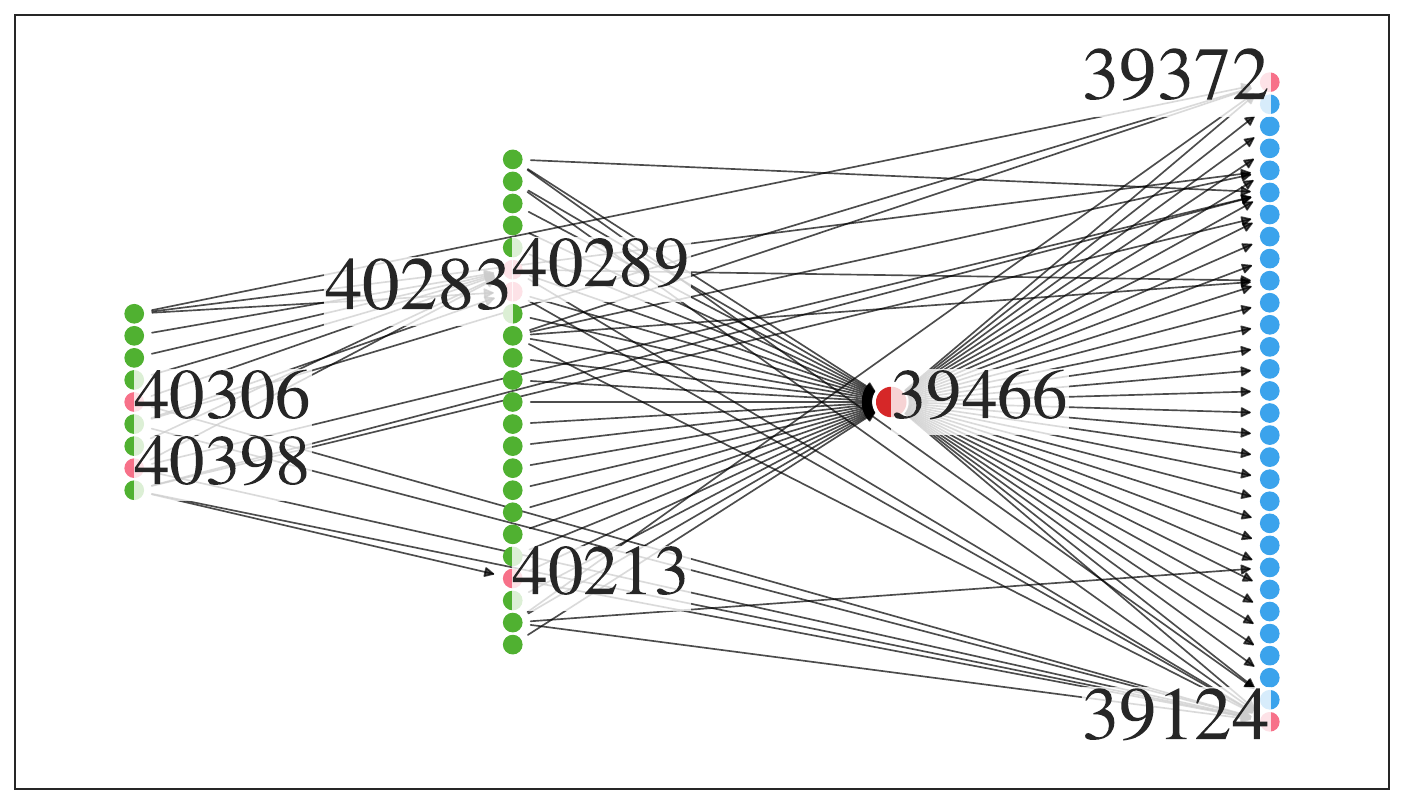}
    \end{minipage}%
    \hfill
    \begin{minipage}[b]{0.49\linewidth}
        \centering
        \includegraphics[width=\linewidth]{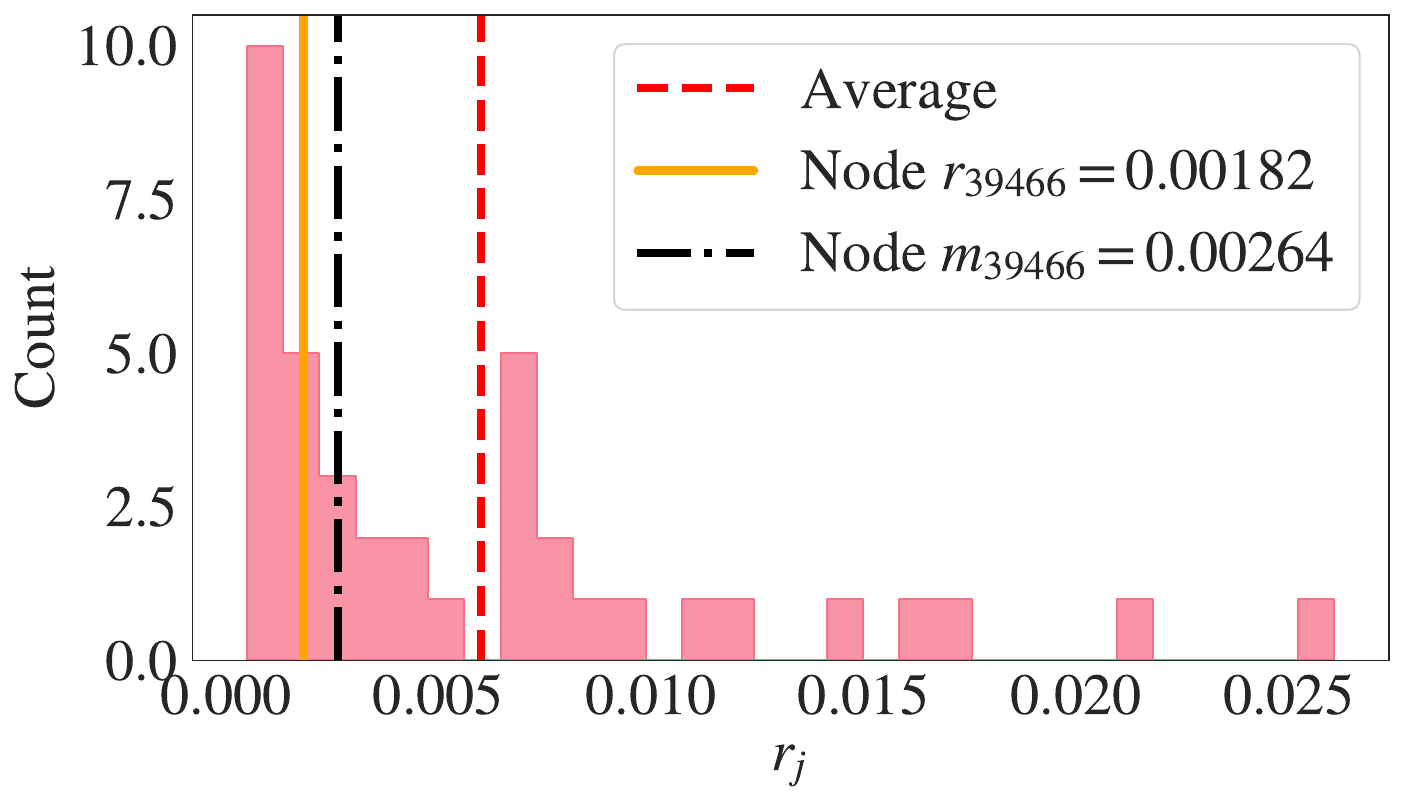}
    \end{minipage}
    \caption{Case study for node 39466. On the left we show its induced sub-graph of the contractor network, where green nodes correspond to principals, pink to other intermediaries, and blue to pure obligees. On the right we plot a histogram of the risk scores $r_i$ (x-axis) against counts for the upstream neighbors. The red dotted line corresponds to the average risk score of the upstream neighbors, the orange to the risk score of node 39466, and the black line to the induced limiting failure probability $\mean_{39466}$ of node 39466.}
    \label{fig:case_study}
\end{figure}

Before diving into our simulations on the weakly connected component of the network, we begin with an illustrative case study. While we emphasize that the nodes here do not refer to specific contractors, in our collaboration with the partnering insurance company we performed a similar methodological analysis that they leveraged for identifying and monitoring risk in key contractors.

We focus on node 39466, which contracts as an \emph{intermediary} (has both incoming and outgoing edges). According to our risk-based centrality measure $u_i$ (cf.\ \cref{def:page_rank}), this node lies in the 75th percentile of the distribution among intermediaries, marking it as structurally ``risky'' due to its prominent position in the network. Its induced subgraph (\cref{fig:case_study}, left) shows both upstream principals and downstream obligees, highlighting its central role in bridging multiple tiers of the contracting system. This centrality measure thus provides a systematic way of flagging such intermediaries as candidates for closer attention.

Turning to the upstream neighbors’ risk distribution (\cref{fig:case_study}, right), we observe substantial heterogeneity.  Some neighbors have risk scores lower than node 39466’s idiosyncratic value $r_{39466} = 0.00182$, while others are considerably higher. While some neighbors are less risky, when incorporating these exposures, the induced limiting failure probability increases to $\mean_{39466} = 0.00264$, corresponding to a relative increase of roughly \(\,45\%\). This amplification illustrates how even a moderately risky intermediary, once identified through the $u_{39466}$ centrality metric, can see its effective default probability substantially elevated due to network position and contracting relationships.

\subsection{Impact of network effects on global financial loss}
\label{sec:sims_global_loss}

Next we begin to illustrate the impact of network risk on the global average financial loss in the surety network.  We start by recalling in \cref{sec:marginal_risk_increase} that under \Cref{ass:larger_neighbors} the average risk increases.
However, this assumption is not always satisfied in our surety network (see more discussion on this in \cref{app:sims_not_satisfying_assumption}).

\begin{figure}[!t]
    \centering
    \begin{subfigure}[t]{0.4\linewidth}
        \centering
        \includegraphics[width=\linewidth]{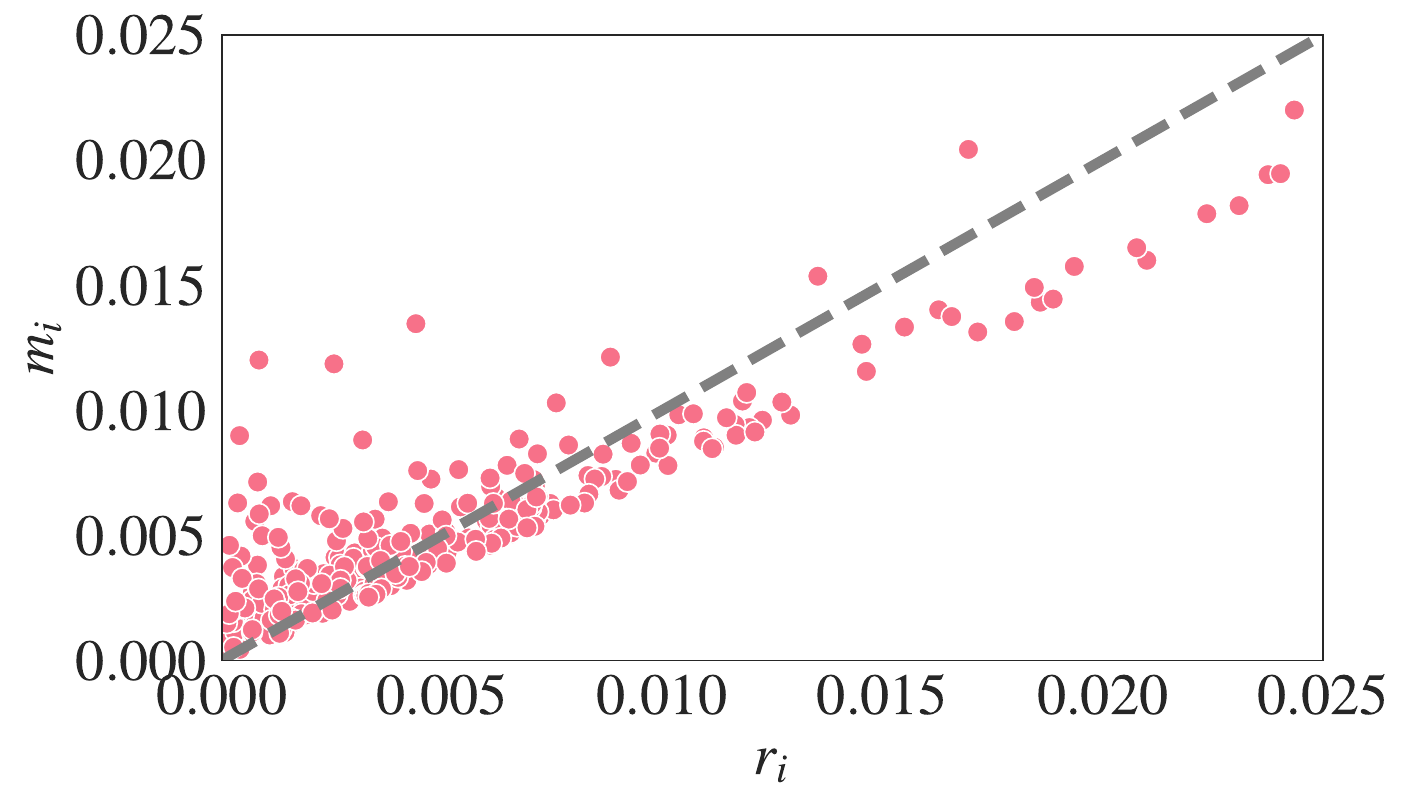}
        \label{fig:r_vs_m}
    \end{subfigure}
    \hfill
    \begin{subfigure}[t]{0.58\linewidth}
        \centering
        \includegraphics[width=\linewidth]{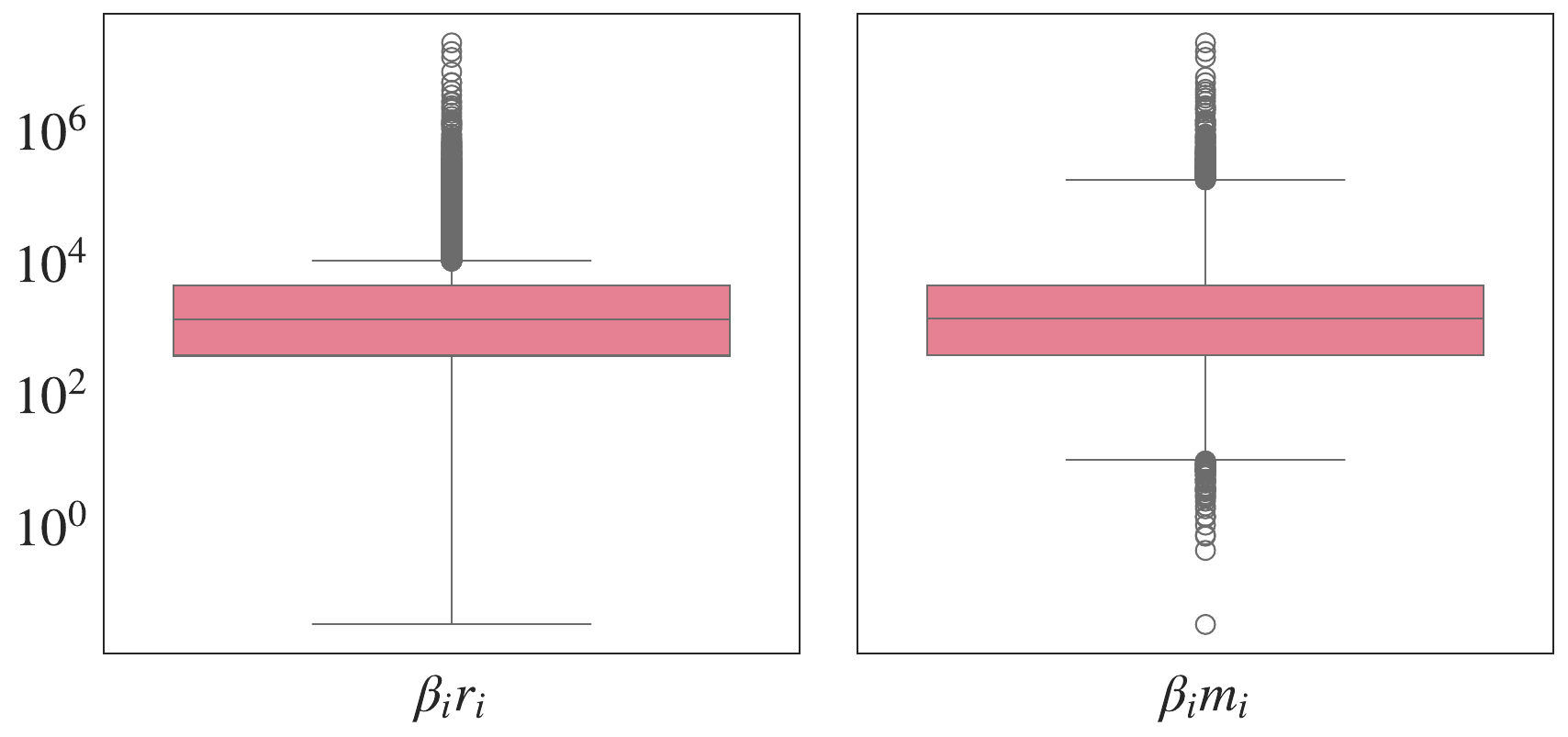}
        \label{fig:loss_box_plot}
    \end{subfigure}
    \caption{Comparison of node-level risk metrics. (Left) Idiosyncratic default probabilities $r_i$ compared to limiting failure probabilities $m_i$ for intermediaries $i$. (Right) Box plots of $\beta_i r_i$ (left) vs.\ $\beta_i \mean_i$ (right), where $\beta_i$ is in USD (log scale).}
    \label{fig:combined_risk_figures}
\end{figure}

Despite this, we observe that for a nontrivial subset of nodes $i \in \V$, the expected marginal loss probability $\mean_i$ exceeds their idiosyncratic risk scores $r_i$.  This effect is visible in \cref{fig:combined_risk_figures} (left), which compares $r_i$ (x-axis) and $\mean_i$ (y-axis), with numerous points above the $y = x$ line.  Moreover, as a result of this, the expected aggregate loss $\Exp{\GAR(\Xvec^{\infty})} = \beta^\top \meanvec$ increases from the expected aggregate loss under the independent failure model $\Exp{\GAR(\Xvec^0)} = \beta^\top \rvec$ by 1.89\%.  This emphasizes that {\em failing to account for network interference causes downstream risk to be under-estimated.}  While 1.89\% might seem like a mild value, we emphasize that the units for these are on the order of hundreds of millions of dollars, so this increase is roughly on the order of 2 million dollars of underestimated risk.  See~\cref{fig:combined_risk_figures} (right) where we also include a box plot of $\beta_i \mean_i$ vs $\beta_i r_i$ across nodes $i \in \V$. Here we observe not only does the average increase, but $\beta_i \mean_i$ is more right-skewed.  This again emphasizes the ability for network risk to destabilize the network.

Lastly, we look at the distribution of $\GAR(\Xvec^\infty)$ versus $\GAR(\Xvec^0)$. 
Table~\ref{tab:loss_quantiles} reports several quantiles of the two distributions, including the 50th, 90th, 95th, 99th, and 99.5th percentiles. In every case, the quantile under \(\GAR(\Xvec^\infty)\) exceeds the corresponding quantile under \(\GAR(\Xvec^0)\). 
These plots reveal a clear shift in the distribution.  Not only is the right tail substantially heavier, indicating a higher probability of extreme loss realizations, but the central behavior is also affected; the median aggregate loss under $\GAR(\Xvec^\infty)$ is noticeably larger than that under $\GAR(\Xvec^0)$, in addition to the mean being higher.  This combination of a heavier tail and an upward shift in the bulk of the distribution highlights that the amplification of network interactions are not just confined to rare catastrophic events, but also manifest across the entire distribution as well.

\begin{table}
    \centering
    \begin{tabular}{lcc}
\toprule
 & $\GAR(\Xvec^0)$ & $\GAR(\Xvec^\infty)$ \\
\midrule
0.5 & $188.71 \pm 0.33$ & ${\bf 193.24^\star \pm 0.33}$ \\
0.9 & $344.92 \pm 0.80$ & ${\bf 349.16^\star \pm 0.74}$ \\
0.95 & $404.75 \pm 0.99$ & ${\bf 410.03^\star \pm 1.02}$ \\
0.99 & $536.35 \pm 2.40$ & ${\bf 559.99 \pm 3.08^\star}$ \\
0.995 & $597.16 \pm 4.33$ & ${\bf 652.92 \pm 5.68^\star}$ \\
\bottomrule
\end{tabular}
\caption{Comparison of the quantiles (in millions of dollars) of $\GAR(\Xvec^\infty)$ to $\GAR(\Xvec^0)$.  Confidence intervals computed with a two-sided binomial quantile test. Bold indicates larger for the same quantile, and $\star$ that the increase is significant from the two-sided binomial quantile test.}
\label{tab:loss_quantiles}
\end{table}

To summarize, these results suggest that only considering the independent failure model without accounting for how the interconnectedness of contractors creates correlations in defaults leads to underestimating the potential losses incurred by failures.  This can have significant repercussions in practice, such as from failing to set aside enough capital to cover large tail probability losses due to this underestimation.

\subsection{Robustness to choice of $\alpha_i$}
\label{sec:simulations_robustness_alpha}

\begin{figure}
    \centering
    \includegraphics[width=0.5\linewidth]{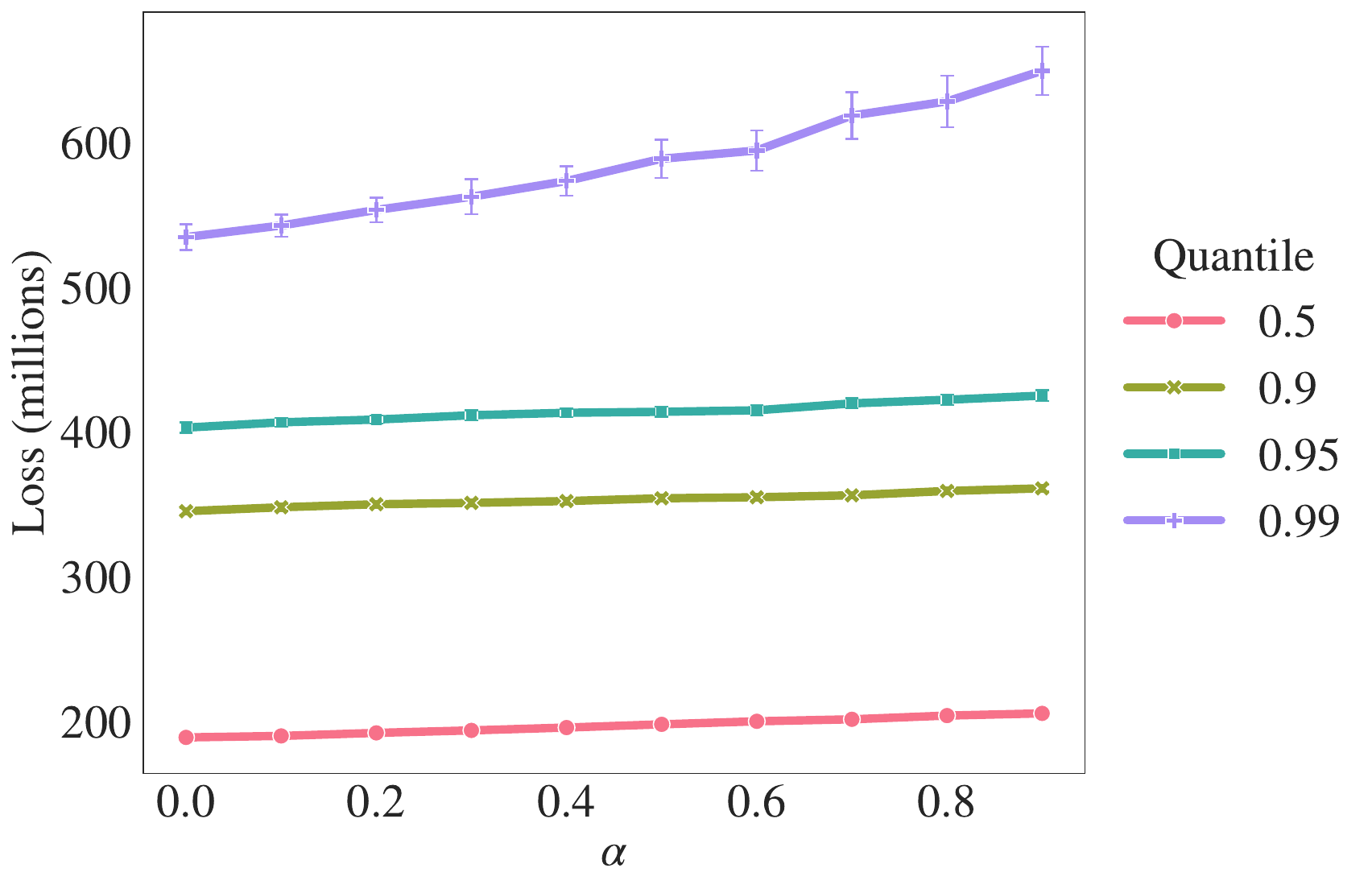}
    \caption{Quantiles of $\GAR(\Xvec^\infty)$ under different values of $\alpha$ in ten equal spaces between $[0,1]$.  Confidence intervals computed with a two-sided binomial quantile test.
    }
    \label{fig:alpha_robustness}
\end{figure}

We close out by testing the robustness of our earlier empirical insights with respect to the choice of $\alpha$, where the baseline specification set $\alpha_i = 0.25$ for all intermediaries $i$ in the network. In \cref{fig:alpha_robustness}, we report the quantiles of \(\GAR(\Xvec^\infty)\) as a function of $\alpha \in [0,1]$, applied uniformly across all intermediaries. Across all specifications, the outcomes remain monotone in $\alpha$ as greater exposure consistently increases systemic risk. This finding is noteworthy because the formal monotonicity condition (\Cref{ass:larger_neighbors}) does not hold in our empirical network. 

More broadly, the extent of amplification is not uniform across the loss distribution. While moderate quantiles exhibit only gradual upward shifts as $\alpha$ increases, the tail behavior is dramatically more sensitive. For example, at the $0.995$ quantile we observe a sharp and disproportionate increase, indicating that the risk of catastrophic losses escalates much faster than median losses. This widening gap across quantiles highlights how network contagion disproportionately impacts the extreme right tail. Taken together, \cref{fig:alpha_robustness} conveys two critical insights: (i) our monotonicity results are robust in practice, even when their sufficient conditions are partially violated, and (ii) the most severe consequences of increasing $\alpha$ manifest in the extremes of the distribution, where insurers and policymakers are most vulnerable.

\section{Conclusion}
\label{sec:conclusion}

In this work, we introduced a network-based approach to analyzing risk propagation in surety-backed contractor networks. By modeling contractual dependencies as a directed stochastic process, we demonstrated that network effects systematically amplify failure probabilities and increase expected loss beyond what traditional independent risk models predict. Our theoretical results establish conditions under which systemic risk accumulates over time, and our empirical analysis using real-world surety data validates these findings, showing that accounting for network dependencies leads to a higher estimated risk exposure than independent models. Additionally, we identified key intermediary nodes that disproportionately influence network-wide stability, highlighting their role in amplifying or mitigating failures.

Several future directions emerge from our work. First, while our analysis focused on risk propagation in a static network, real-world contractor networks evolve over time as firms form new contracts or exit the market. Extending our model to a dynamic setting, where network structure evolves alongside risk accumulation, is a promising avenue for further research. Second, our framework assumes full network observability by the surety provider, yet in practice, some contractual relationships may be hidden due to the presence of multiple insurers or informal agreements. Developing robust risk estimation techniques that account for missing or latent network information would enhance the applicability of our approach. Finally, while we focused on financial surety networks, similar risk propagation dynamics arise in other interdependent systems, such as supply chains, infrastructure networks, and research collaborations. Extending our methods to these domains could provide new insights into systemic vulnerabilities and optimal risk mitigation strategies.

\medskip

\noindent\textbf{Acknowledgments.} Part of this work was done while Sean R. Sinclair was a Postdoctoral Associate at the Massachusetts Institute of Technology, advised by Ali Jadbabaie and Devavrat Shah. The authors would also like to thank Alireza Tahbaz-Salehi for preliminary feedback on the network surety model, and Kelton Busby for valuable assistance with the model and simulation studies. \srsedit{We are also grateful to Runhuan Wang for  identifying several errors in the proofs.} Icons for \cref{fig:simple_diagram} are provided by \citet{flaticon2025}. Finally, we gratefully acknowledge the partnership of an anonymous surety organization for providing data and insights that made this research possible.

\newpage

\bibliographystyle{informs2014} 
\bibliography{references} 


%
%
%

\newpage

\AtBeginEnvironment{APPENDICES}{%
  \renewcommand{\theHsection}{appendix.\Alph{section}}%
  \renewcommand{\theHsubsection}{\theHsection.\arabic{subsection}}%
  \renewcommand{\theHsubsubsection}{\theHsubsection.\arabic{subsubsection}}%
  \renewcommand{\theHfigure}{\theHsection.\arabic{figure}}%
  \renewcommand{\theHtable}{\theHsection.\arabic{table}}%
  \renewcommand{\theHequation}{\theHsection.\arabic{equation}}%
  \renewcommand{\theHtheorem}{\theHsection.\arabic{theorem}}%

  \crefalias{section}{appendix}%
  \crefalias{subsection}{appendix}%
  \crefalias{subsubsection}{appendix}%
}

\crefalias{section}{appendix}
\begin{APPENDICES}

\section{Table of Notation}
\label{app:notation}

\renewcommand{\arraystretch}{1.1}

\begin{table*}[h!]
\begin{tabular}{>{\color{edits}}l | >{\color{edits}}l}
\textbf{Symbol} & \textbf{Definition} \\ \hline
\multicolumn{2}{c}{Problem setting specifications}\\
\hline
$\mathcal{G} = (\mathcal{V},\mathcal{E})$ & Network surety graph \\
$i,j, k \in\mathcal{V}$ & Index over nodes in the graph \\
$(j,i) \in \mathcal{E}$ & Directed edge from principal (subcontractor) $j$ to obligee (project owner) $i$ \\
$\din(i)$ & \srsedit{Principals (in-neighbors) of obligee $i$, $\din(i) = \{ j : (j,i) \in \mathcal{E} \}$} \\
$\dout(j)$ & Obligees (out-neighbors) of principal $j$, $\dout(j) = \{i : (i,j) \in \mathcal{E}\}$ \\
$w_{ij}\in[0,1]$ & Fraction of obligee $i$'s projects contracted to principal $j$, where $\sum_{j\in\din(i)} w_{ij} \leq 1$ \\ 
$\W\in\mathbb{R}^{|\mathcal{V}|\times|\mathcal{V}|}$ & Weighted adjacency matrix with entries $w_{ij}$ \\
pure obligee & Node $i$ such that $\dout(i) = 0$ \\
pure principal & Node $j$ such that $\din(j) = 0$ \\
$r_i$ & Risk score for node $i$ ($r_i=0$ if $i$ is a pure obligee) \\
$\rvec \in\mathbb{R}^{|\mathcal{V}|}$ & Column vector of idiosyncratic risk scores with entries $r_i$ \\
$\alpha_i\in[0,1]$ & Probability of network effects ($\alpha_j=0, \alpha_i=1$ for pure principals $j$ and obligees $i$) \\
$\beta_i$ & Financial loss for each node $i \in \mathcal{V}$ \\
$\A \in\mathbb{R}^{|\mathcal{V}|\times|\mathcal{V}|}$ & Diagonal matrix of $\alpha_i$ \\
$\Xvec^t$ & Stochastic process for node failures at step $t$ \\
$\pi$ & Stationary distribution for $\Xvec^{t}$ \\
$\meanvec^t$ & Vector of failure probabilities with entries $\mean_i^t = \Exp{X_i^t}$ \\ 
$\meanvec$ & $\lim_{t \rightarrow \infty} \meanvec^t$, the limiting failure probabilities for each node \\
${\mathbf{1}}$ & Column vector of all ones, dimensions may vary based on context \\
$\tau_1,\dots,\tau_T$ & Product segment IDs indicating the type of work performed \\
$\type(i)$ & Product segment ID of node $i$, $\type(i)\in\{\tau_1,\dots,\tau_T\}$ \\
$\rbar_i$ & Median risk score in $i$'s product segment, $\rbar_i=\text{median}(\{r_j\,:\,\type(j)=\type(i)\})$ \\
$\uvec \in\mathbb{R}^{|\mathcal{V}|}$ & risk-based centrality in equilibrium model, $\uvec = (\I-\A)(\I - \W^\top \A)^\inv \frac{\mathbf{1}}{n}$ \\
$\GAR(\mathbf{x})$ & Weighted average risk: $\sum_{v} \beta_v x_v$ \\
\hline
\end{tabular}
\caption{Common notation}
\label{table:notation}
\end{table*}

\section{\cref{sec:mean_field} Omitted Proofs}
\phantomsection
\label{app:mean_field_proofs}

\subsection{\cref{sec:mean_failure_probabilities} Omitted Proofs}
\label{sec:limit_mean_proofs}

We start off by analyzing the mean field dynamics of our stochastic process and showing the mean failure probabilities satisfy a similar dynamics to the original process.

\MeanFieldRecurrence*
\begin{proof}
The case when $t = 0$ immediately follows since $X_i^0 \sim \Ber(r_i)$.  For the step case we use the law of total probability and the Markov property of the Markov chain to have:
\begin{align*}
    m_i^{t+1} = \Exp{X_i^{t+1}} & = \Exp{\Exp{X_i^{t+1} \mid \Xvec^{t}}} = \Exp{(1 - \alpha_i) r_i + \alpha_i \sum_{j \in \din(i)} w_{ij} X_j^{t}} \\
    & = (1 - \alpha_i) r_i + \alpha_i \sum_{j \in \din(i)} w_{ij} \Exp{X_j^{t}} = (1 - \alpha_i) r_i + \alpha_i \sum_{j \in \din(i)} w_{ij} \mean_j^{t},
\end{align*}
where the second equality follows by \cref{eq:sp_dynamics}, and the final equality by an inductive argument.
\end{proof}
Our first main result for the mean field shows that the limiting failure probabilities indeed exist and satisfy a fixed point equation.

\InverseFixedPoint*
Before presenting the proof of \cref{thm:inverse_fixed_point}, we start off with the following technical lemma.

\SecondNormBounded*
\begin{proof}
    Since $\A\W$ is row sub-stochastic, we have that $\norm{\A\W} \leq 1$.  Additionally, for all pure obligees $i$ we have set $\alpha_i=1$, so their corresponding row sums satisfy $\sum_j (\A\W)_{ij} = \alpha_i \sum_j w_{ij} = 1$.  Thus by construction, $\|\A\W\|=1$, which does not necessarily imply that powers of $\A\W$ are decreasing in size.
    We can, however, show that $\norm{(\A\W)^2}<1$.  In particular, the row sums of $(\A\W)^2$ are given by:
    \begin{equation*}
        \sum_{j\in\mathcal{V}} \srsedit{((\A\W)^2)_{ij}} = \sum_{j\in\mathcal{V}} \sum_{k\in\mathcal{V}} \alpha_i \alpha_k w_{ik} w_{kj} = \alpha_i \sum_{k\in\mathcal{V}} \alpha_k w_{ik} \sum_{j\in\mathcal{V}} w_{kj} \leq \sum_{k\in\mathcal{V}} \alpha_k w_{ik}.
    \end{equation*}
    If $w_{ik}>0$, node $k$ cannot be a pure obligee and $\alpha_k$ must be strictly less than 1.
    Therefore, the row sums are strictly less than 1:
    \begin{equation*}
        \sum_{j\in\mathcal{V}} (\A\W)^2_{ij} \leq \sum_{k\in\mathcal{V}} \alpha_k w_{ik} < \sum_{k\in\mathcal{V}} w_{ik} \leq 1,
    \end{equation*}
    and $\|(\A\W)^2\|<1$.
    Then sub-multiplicativity of the norm implies that higher powers of $\A\W$ are decaying in size.
    Even powers are bounded $\|(\A\W)^{2k}\| \leq \|(\A\W)^2\|^k < 1$, and odd powers can be bounded by even powers: $\|(\A\W)^{2k+1}\| \leq \|(\A\W)^{2k}\|\|\A\W\| \leq \|(\A\W)^{2k}\|.$  Using that $\norm{(\A\W)^2} < 1$ we get that $(\A\W)^t \rightarrow 0$ as required.
\end{proof}
Using \cref{lem:second_norm_bounded} we can show the following corollary.

\Neumann*
\begin{proof}
\srsedit{Note that $\rho((\A\W)^2) \leq \norm{(\A\W)^2} < 1$ implies that $\rho(\A\W)^2=\rho((\A\W)^2)<1$, hence $\rho(\A\W)<1$.
Therefore $\I-\A\W$ is invertible and the Neumann series converges, with
\[
(\I-\A\W)^{-1}=\sum_{t=0}^\infty (\A\W)^t,
\]
by the standard result that $\rho(\mathbf{M})<1$ implies $\sum_{t\ge 0} \mathbf{M}^t$ converges (in any induced matrix norm) and equals $(\I-\mathbf{M})^{-1}$.
}
\end{proof}

With the previous results in hand, we are finally ready to show \cref{thm:inverse_fixed_point}.
\begin{proof}[Proof of \cref{thm:inverse_fixed_point}]
First via \cref{eq:mean_field_recurrence} we have:
\begin{equation}
\label{eq:mean_field_matrix_recurrency}
\meanvec^{t+1} = (\I - \A)\rvec + \A\W\meanvec^{t}.
\end{equation}
By expanding out the previous equation we have that:
\[
\meanvec^{t+1} = \sum_{k < t+1} (\A\W)^k (\I - \A)\rvec + (\A\W)^{t+1} \rvec.
\]
However, in \cref{cor:neumann} we showed that the Neumann series $\sum_{t=0}^\infty (\A\W)^t$ converges to $(\I - \A\W)^\inv$.  Hence, it follows that
    \begin{equation*}
        \lim_{t\rightarrow\infty} \meanvec^{t+1} = \lim_{t\rightarrow\infty}\sum_{k<t+1}(\A\W)^k(\I - \A)\rvec + (\A\W)^{t+1}\rvec = (\I - \A\W)^{-1}(\I - \A)\rvec = \meanvec.
    \end{equation*}
\end{proof}

\LimitingConvergenceRateAcyclic*
\begin{proof}
    Using the fact that $\norm{(\A\W)^2}<1$ by \cref{lem:second_norm_bounded} we look at the vector $\meanvec^t - \meanvec$ to have:
        \[
    (\meanvec^t-\meanvec) 
    = -\sum_{k=t}^\infty (\A\W)^k (\I - \A)\rvec + (\A\W)^t\rvec.
    \]
    By the triangle inequality
    \begin{align*}
        \norm{\meanvec^t-\meanvec} &\leq \norm{\textstyle\sum_{k\geq t} (\A\W)^k (\I - \A)\rvec}_\infty + \norm{(\A\W)^t\rvec}_\infty, \\
        &\leq \norm{(\A\W)^t} \norm{\rvec}_\infty \Big(\norm{\textstyle\sum_{k\geq0} (\A\W)^k} \norm{\I - \A} + 1\Big), \\
        &\leq \Big(1+\textstyle\sum_{k\geq0} \norm{(\A\W)^k}\Big) \norm{\rvec}_\infty \norm{(\A\W)^t}.
    \end{align*}
    Using that
    \[
        \sum_{k\geq0} \norm{(\A\W)^k}
        \leq \sum_{k\geq0} \norm{(\A\W)^2}^{\lfloor k/2 \rfloor}
        = 2 \sum_{k\geq0} \norm{(\A\W)^2}^k
        = \frac{2}{1-\norm{(\A\W)^2}},
    \]
    we then get:
    \begin{equation*}
        \norm{\meanvec^t - \meanvec}_\infty \leq \left( 1 + \frac{2}{1-\norm{(\A\W)^2}} \right) \norm{\rvec}_\infty \norm{(\A\W)^2}^{\lfloor t/2 \rfloor}.
    \end{equation*}

    Next suppose that $G$ is a directed acyclic graph.  Then the last property can be shown from the fact that $(\A\W)^t=0$ for all $t>d$.
    \srsedit{To see why this is, we observe that each entry is $(\A\W)^t_{ij} = \sum_{(v_t=j,\dots,v_t=i)\in\mathcal{V}^t} \prod_{k=0}^{t-1} \alpha_{v_k}w_{v_k,v_{k+1}} > 0$ if and only if $(v_0,\dots,v_t)$ is a length $t$ path from $v_0$ to $v_t$.
    Since the longest path contains $d$ edges, it follows that all entries of $(\A\W)^t$ are zero and $\meanvec^t = \sum_{k \leq d} (\A\W)^k (\I - \A)\rvec = \meanvec$ for $t>d$.
    In fact, we also have that $\meanvec^d=\meanvec$ because $(\A\W)^d$ is only nonzero in columns corresponding to pure principals with $\alpha_i=0$, so $(\A\W)^{d} = (\A\W)^{d}(\I-\A)$.
    Then $\meanvec^d = \sum_{k < d} (\A\W)^k (\I - \A)\rvec + (\A\W)^{d}(\I-\A) \rvec =\meanvec$.}
\end{proof}


\subsection{\cref{sec:marginal_risk_increase} Omitted Proofs}
\label{app:mean_increase_proofs}

\AssumptionMeanIncrease*

We show the result through the combination of the following lemmas.

\begin{lemma}
\label{lem:monotonicity}
Under Assumption~\ref{ass:larger_neighbors}, for all \(i \in \mathcal{V}\) and \(t \in \mathbb{N}\) the mean failure probabilities $\mean_i^t$ are monotone with respect to $t$, i.e.
\[
    \mean_i^{t+1} \geq \mean_i^t.
\]
Hence, for all \(t\) and \(i\), $\mean_i \geq \mean_i^t$.
\end{lemma}
\begin{proof}
We will show the property by induction for each node $i \in \mathcal{V}$.

\noindent \textit{Base Case $(t = 0)$}: 
\begin{align*}
    m_i^1 &= (1 - \alpha_i) r_i + \alpha_i \sum_{j \in \din(i)} w_{ij} \Exp{X_j^0} = (1 - \alpha_i) r_i + \alpha_i \sum_{j \in \din(i)} w_{ij} r_j \\
          &\underset{A\ref{ass:larger_neighbors}}{\geq} (1 - \alpha_i) r_i + \alpha_i r_i = r_i = m_i^0.
\end{align*}

\noindent \textit{Step Case $(t \rightarrow t+1)$}:  By definition of the mean failure probabilities,
\[
    \begin{cases}
        m_i^{t+1} = (1 - \alpha_i) r_i + \alpha_i \sum_{j \in \din(i)} w_{ij} m_j^t, \\
        m_i^t = (1 - \alpha_i) r_i + \alpha_i \sum_{j \in \din(i)} w_{ij} m_j^{t-1}.
    \end{cases}
\]
Thus,
\[    
    m_i^{t+1} - m_i^t = \alpha_i \sum_{j=1}^{n} w_{ij} (m_j^t - m_j^{t-1}) \underset{(\text{I.H.})}{\geq} 0 \implies m_i^{t+1} \geq m_i^t, \quad \text{for all } i \in \mathcal{V}.
\]
Therefore, $\mean_i^{t+1} \geq \mean_i^t$ as required.
\end{proof}

\begin{lemma}
\label{lem:alpha_monotonicity}
    Under \Cref{ass:larger_neighbors}, increasing the exposure to risk from the network $\alpha_i$ of any intermediary $i \in \mathcal{V}$ can only increase default probabilities in $\meanvec$.
    That is, all entries of $\meanvec$ are monotonically non-decreasing with respect to $\alpha_i$:
    \[
        \frac{\partial \meanvec}{\partial \alpha_i}
        = (\I-\A\W)^{-1} \frac{\partial \A}{\partial \alpha_i} \left[ \W\meanvec - \rvec \right]
        \;\geq\; \mathbf{0}.
    \]
\end{lemma}

\begin{proof} 
    The vector derivative of $\meanvec = (\I-\A\W)^{-1}(\I-\A)\rvec$ with respect to $\alpha_i$ is
    \begin{align*}
        \frac{\partial \meanvec}{\partial \alpha_i}
        &= \frac{\partial}{\partial \alpha_i} \left[(\I-\A\W)^{-1}\times(\I-\A)\rvec\right] \\
        &= \frac{\partial (\I-\A\W)^{-1}}{\partial \alpha_i}(\I-\A)\rvec + (\I-\A\W)^{-1} \frac{\partial (\I-\A)\rvec}{\partial \alpha_i} \\
        &= -(\I-\A\W)^{-1}\frac{\partial (\I-\A\W)}{\partial \alpha_i}(\I-\A\W)^{-1}(\I-\A)\rvec + (\I-\A\W)^{-1} \frac{\partial (\I-\A)}{\partial \alpha_i}\rvec \\
        &= (\I-\A\W)^{-1}\frac{\partial \A}{\partial \alpha_i} \W \underbrace{(\I-\A\W)^{-1}(\I-\A)\rvec}_{\meanvec} - (\I-\A\W)^{-1} \frac{\partial \A}{\partial \alpha_i}\rvec \\
        &= (\I-\A\W)^{-1}\frac{\partial \A}{\partial \alpha_i}(\W\meanvec-\rvec).
    \end{align*}
    Because $\A\W$ is row sub-stochastic and non-negative, $(\I-\A\W)^{-1} = \sum_{t=0}^\infty (\A\W)^t \succeq \mathbf{0}$ is also non-negative.
    Additionally, $\frac{\partial \A}{\partial \alpha_i}\succeq\mathbf{0}$ is a diagonal matrix with zero entries for all rows corresponding to pure principals and pure obligees, since their entries are fixed as 0 and 1 in $\A$.
    We will multiply $(\W\meanvec-\rvec)$ by $\frac{\partial \A}{\partial \alpha_i}$ from the left, which will give a vector $\frac{\partial \A}{\partial \alpha_i} (\W\meanvec-\rvec)$ whose only nonzero entries are those corresponding to intermediaries.
    \srsedit{From \Cref{lem:monotonicity} and \Cref{ass:larger_neighbors}, we have that for all intermediaries $i$,}
    \[
        (\W\meanvec-\rvec)_i
        = \Big( \sum_{j\in\din(i)} w_{ij} \mean_j \Big) - r_i
        \geq \Big( \sum_{j\in\din(i)} w_{ij}r_j \Big) - r_i
        \geq 0.
    \]
    Then $\frac{\partial \A}{\partial \alpha_i} (\W\meanvec-\rvec)\geq$, and the vector derivative $\frac{\partial \meanvec}{\partial \alpha_i}\geq\mathbf{0}$ is non-negative.
\end{proof}

\begin{corollary}
\label{cor:monotonicity_loss}
Under \Cref{ass:larger_neighbors}, we have that for all $t  \in \NN$,
\begin{equation}
    \label{eq:montone_loss}
    \Exp{\mathcal{L}(\Xvec^{t+1})} \geq \Exp{\mathcal{L}(\Xvec^{t})}.
\end{equation}
Hence, we have that:
\begin{equation}
\label{eq:montone_loss_stat}
    \Exp{\mathcal{L}(\Xvec^\infty)} \geq \Exp{\mathcal{L}(\Xvec^{0})}.
\end{equation}
\end{corollary}
\begin{proof}
Since \(\mathcal{L}(\Xvec^t) = \sum_{i=1}^{n} \beta_i X_i^t\),
with $\beta_i \geq 0$ for all $i$, \cref{eq:montone_loss} follows by linearity of expectation and~\Cref{lem:monotonicity}, and \cref{eq:montone_loss_stat} follows by taking the limit as \(t \to \infty\) and the continuous mapping theorem and the linearity of $\GAR(\cdot)$.
\end{proof}

An obvious question remains as to whether the reverse is true under \Cref{ass:larger_neighbors}, i.e. if intermediaries have lower risk principals on average then the mean failure probabilities are monotone decreasing.  Indeed, a straightforward extension to the previous discussion establishes the following:

\srsedit{
\begin{corollary}
\label{cor:reverse_mean_monotonicity}
Let $\mathcal N := V \setminus O$ denote the set of principals and intermediaries.  If all intermediaries $i \in \mathcal N$ have lower-risk principals on average, i.e. $\sum_{j \in \din(i)} w_{ij} r_j \leq r_i$, then the mean failure probabilities are componentwise non-increasing on $\mathcal N$, i.e. $\mean^{t+1}_i \leq \mean^t_i$ for all $i \in \mathcal N$.
Consequently, $\mean_i \leq r_i$ for all $i \in \mathcal N$ and $\Exp{\GAR(\Xvec^{t+1})} \leq \Exp{\GAR(\Xvec^t)}$.
Furthermore, for any intermediary $i$,
\[
\frac{\partial \meanvec}{\partial \alpha_i}
=
(\I-\A\W)^{-1}
\frac{\partial \A}{\partial \alpha_i}
\left[ \W\meanvec - \rvec \right]
\le \mathbf{0}.
\]
\end{corollary}
}

\begin{proof}
    We can follow arguments similar to those in \Cref{lem:monotonicity,cor:monotonicity_loss,lem:monotonicity} by replacing $\geq$ with $\leq$ where appropriate to show that similar statements hold for the opposite direction.
\end{proof}


\MeanFieldGap*
\begin{proof}
To show this claim we start off with the following lemma:

\begin{lemma}
\label{lem:mean_risk_increase}
For any contractor graph $G$ and all $t \in \mathbb{N}$, under \Cref{ass:strictly_larger_neighbors} we have
\[
    \meanvec^t \geq f_t(\delta)\,\rvec,
\]
where
\[
    f_t(\delta) :=
    \begin{cases}
        \I, & t=0, \\[6pt]
        \I + \delta \sum_{k=0}^{t-1} (\A\W)^k \A, & t \geq 1.
    \end{cases}
\]
\end{lemma}
\begin{proof}
First we argue that under \Cref{ass:strictly_larger_neighbors}, we have that for all nodes $i$:
\[
\alpha_i w_i^\top \rvec \geq (1 + \delta) \alpha_i r_i,
\]
where $w_i^\top = [w_{i1}\quad w_{i2} \quad \dots \quad w_{in}]$ denotes the $i$th row in $\W$.
For any pure principal $i$ we have $\alpha_i = 0$, so both sides are zero and the inequality is trivially satisfied. For pure obligees $i$ $r_i = 0$, and so again the inequality is trivially satisfied. Then for intermediaries with $\alpha_i\in(0,1)$, the inequality is precisely \Cref{ass:strictly_larger_neighbors}.
Thus, we have that $\A \W \rvec \geq (1 + \delta) \A \rvec$.

We now show the result by considering the change in $\meanvec^t$ at each step.
In the first step, we have that
\begin{equation*}
    \meanvec^1 = (\I-\A)\rvec + \A\W\meanvec^0 
    = (\I-\A)\rvec + \underbrace{\A\W\rvec}_{\geq(1+\delta)\A\rvec} 
    \geq r + \delta\A\rvec,
\end{equation*}
i.e. $\meanvec^1-\meanvec^0\geq\delta\A\rvec$.
Then \Cref{lem:mean_field_recurrence} implies that $\meanvec^{t+1}-\meanvec^t = (\A\W)^t(\meanvec^1-\meanvec^0)$, so we have that $\meanvec^{t+1}-\meanvec^t \geq \delta(\A\W)^t\A\rvec$.
Then
\begin{equation*}
    \meanvec^t = (\meanvec^t-\meanvec^{t-1})+\cdots+(\meanvec^1-\meanvec^0) + \meanvec^0
    \geq \rvec + \delta \sum_{k=0}^{t-1} (\A\W)^k\A\rvec,
\end{equation*}
which shows that $\meanvec^t \geq \Big[ \I + \delta \sum_{k=0}^{t-1} (\A\W)^k\A \Big] \rvec = f_t(\delta)\rvec$.
\end{proof}

Using \cref{lem:mean_risk_increase} and taking the limit as $t \rightarrow \infty$ we have that:
\begin{align*}
    \meanvec & \geq (\lim_{t \rightarrow \infty} f_t(\delta)) \rvec,
\end{align*}
where we can use the interchange, since all of the limits exist.  However,
\srsedit{\begin{align*}
    \lim_{t \rightarrow \infty} f_t(\delta) & = \lim_{t \rightarrow \infty} \I + \delta \sum_{k=0}^{t-1}(\A\W)^k \A \\
    & = \I + \delta \lim_{t \rightarrow \infty} \sum_{k=0}^{t-1} (\A\W)^k \A \\
    & = \I + \delta (\I - \A\W)^\inv\A
\end{align*}
where in the last line we used the Von-Neumann expansion for a matrix (see \cref{cor:neumann}). Rearranging this expression gives that $\meanvec - \rvec \geq \delta(\I - \A\W)^{-1}\A\rvec$.}

Moreover, we also have:
\begin{align*}
    \Exp{\GAR(\Xvec^\infty)} & = \sum_i \beta_i \meanvec_i = \beta^\top \meanvec \\
    & \geq \beta^\top (\I + \delta (\I - \A\W)^\inv \A) \rvec = \beta^\top \rvec + \delta \beta^\top (\I - \A\W)^\inv \A \rvec = \Exp{\GAR(\Xvec^0)} + \delta \beta^\top (\I - \A\W)^\inv \A \rvec. 
\end{align*}
\end{proof}

\section{\cref{sec:stochastic_process_analysis} Omitted Proofs}
\label{app:stochastic_process_proofs}

\MarkovChainStationary*
\begin{proof}
    The fact that $\Xvec^t$ is a Markov chain over the finite state space $\{0,1\}^n$ follows immediately from our construction; that is, $(\Xvec^t)_{t \in \NN}$ is a sequence of random variables that satisfy the Markov property
    \(
        \Pr(\Xvec^{t+1}=\xvec^{t+1} \mid \Xvec^{t},\dots,\Xvec^0)
        = \Pr(\Xvec^{t+1}=\xvec^{t+1} \mid \Xvec^{t})
    \)
    for all $t$.
    We will restrict our attention to states $\xvec^t\in\{0,1\}^n$ that are \emph{feasible}, in the sense that there exists some time $t$ such that $\Pr(\Xvec^t=\xvec^t)>0$.
    The Markov chain $\Xvec^t$ has a unique stationary distribution if it is \emph{ergodic} (irreducible, positive recurrent) and \emph{aperiodic} over the feasible state space~\citep{resnick2013adventures}.  Thus, we wish to show that all feasible states can reach each other in a finite number of steps, and that at least one feasible state has a self-loop, so $\Xvec^t$ is also aperiodic whenever it is ergodic.

    Before we show that self-loops exist and that the Markov chain is ergodic, we characterize the states that can be reached in one step from a given current state.
    Conditioning on a feasible current state $\Xvec^t=\xvec^t$, nodes fail or do not fail independently in the following time step, so we can consider the conditional default probability of each node individually.
    We first consider principals (pure principals and intermediaries).  Because we require $r_i\in(0,1)$ and $\alpha_i<1$ for all principals $i$, all of their default probabilities lie strictly between zero and one.  More formally, the probability of any principal $i$ failing in the next step satisfies
    \begin{equation}\label{eq:stochastic_principal_states}
        0 < (1-\alpha_i)r_i \leq \Pr(X_i^{t+1}=1 \mid \Xvec^t = \xvec^t) \leq (1-\alpha_i)r_i + \alpha_i < 1,
    \end{equation}
    for any time $t$ and feasible state $\xvec^t$.  Intuitively, this property holds for principals because their default probability contains an idiosyncratic component.
    Pure obligees, on the other hand, can only fail (or not fail) if one of their principals failed (or did not fail) in the previous time step.
    For any pure obligee $i \in \oblig$, feasible $\xvec^t$, and $x^{t+1}_i\in\{0,1\}$,
    \[
        \Pr(X^{t+1}_i=x^{t+1}_i \mid \Xvec^t=\xvec^t)
        = \sum_{j\in\din(i)} w_{ij} \Ind{x^t_j=x^{t+1}_i} > 0
    \]
    if and only if at least one principal is currently in state $x^{t+1}_i$.
    Then conditional independence implies
    \[
        \Pr(\Xvec^{t+1}=\xvec^{t+1} \mid \Xvec^t=\xvec^t)
        = \prod_{j \in \noblig} \underbrace{\Pr(X^{t+1}_j=x^{t+1}_j \mid \Xvec^t=\xvec^t)}_{>0\;\forall j}
        \prod_{i \in \oblig} \underbrace{\Pr(X^{t+1}_i=x^{t+1}_i \mid \Xvec^t=\xvec^t)}_{>0\text{ iff. }\{j \in\din(i) \mid x^t_j=x^{t+1}_i\}\neq\varnothing}.
    \]
    Thus we have the following necessary and sufficient condition:
    \begin{equation}\label{eqn:reachable_prop}
        \Pr(\Xvec^{t+1}=\xvec^{t+1} \mid \Xvec^t=\xvec^t) > 0
        \iff
        \{ j \in\din(i) \mid x^t_j=x^{t+1}_i \}\neq\varnothing,\text{ for all pure obligees }i.
    \end{equation}
    This shows that we can usually go between any two states in one step, with the only exception being when \Cref{eqn:reachable_prop} does not hold.
    In other words, we can reach $\xvec^{t+1}$ from $\xvec^t$ in one step as long as the states of principals in $\xvec^t$ ``coincide" with the states of pure obligees in $\xvec^{t+1}$.
    {We can now apply the above discussion to show aperiodicity and ergodicity of the Markov chain.}
x    
    We first wish to show that there exists a feasible state $\xvec^t$ with a self-loop, so that we can guarantee aperiodicity whenever we have ergodicity.
    From \Cref{eqn:reachable_prop} it follows that $\xvec^t$ has a self-loop, i.e.
    \(
        \Pr(\Xvec^{t+1}=\xvec^t \mid \Xvec^t=\xvec^t) > 0,
    \)
    if and only if
    \(
        \{ j \in\din(i) \mid x^t_j=x^t_i \}\neq\varnothing.
    \)
    Additionally, at least one such state always exists: the state in which no nodes fail, i.e. the vector of all zeros $\xvec=\mathbf{0}$.  We define $X^0_i \sim \text{Bernoulli}(r_i)$ as independent random variables with $r_i<1$, so
    \(
        \Pr(\Xvec^0=\mathbf{0}) = \prod_{i\in\mathcal{V}} (1-r_i) > 0.
    \)
    Therefore, we have that $\xvec^t=\mathbf{0}$ is a feasible state in which all pure obligees $i$ satisfy
    \[
        \{j\in\din(i) \mid x_j^t=x_i^t\}
        = \{j\in\din(i) \mid \mathbf{0}_j^t=\mathbf{0}_i^t\}
        =\din(i)\neq\varnothing,
    \]
    so we always have at least one feasible state with a self-loop.

    Next we show that the Markov chain is ergodic.  Let $\xvec^{t+2}$ and $\xvec^{t}$ denote two arbitrary feasible states in the Markov chain.  We claim that $\Pr(\Xvec^{t+2} = \xvec^{t+2} \mid \Xvec^t = \xvec^t) > 0$ because we can always construct some state $\xvec^{t+1}$ satisfying $\Pr(\Xvec^{t+2} = \xvec^{t+2} \mid \Xvec^{t+1} = \xvec^{t+1}) > 0$ and $\Pr(\Xvec^{t+1} = \xvec^{t+1} \mid \Xvec^{t} = \xvec^{t}) > 0$.
    We will use the fact that no principals $j\in\din(i)$ of any pure obligee $i$ can be contained in the set of pure obligees and \Cref{eqn:reachable_prop} to show that we can set the states of principals and pure obligees in $\xvec^{t+1}$ separately, making it possible to coordinate with both $\xvec^t$ and $\xvec^{t+2}$.  In particular, we can set the states of principals in $\xvec^{t+1}$ so that the states of pure obligees in $\xvec^{t+2}$ are reachable, then set the states of pure obligees in $\xvec^{t+1}$ to be reachable from the states of principals in $\xvec^t$.
    
    We begin by observing that for any feasible $\xvec^{t+2}$, there must be some feasible $\tilde{\xvec}^{t+1}$ such that
    \(
        \Pr(\Xvec^{t+2}=\xvec^{t+2} \mid \Xvec^{t+1}=\tilde{\xvec}^{t+1}) > 0.
    \)
    However, no pure obligees can be a principal in $\din(i)$ for any node $i$, so their realized states at time $t+1$ do not influence any default probabilities at time $t+2$:
    \begin{align*}
        \Pr(X^{t+2}_i=x^{t+2}_i \mid \Xvec^{t+1}=\tilde{\xvec}^{t+1})
        &= (1-x^{t+2}_i) + (2x^{t+2}_i-1)\Big( (1-\alpha_i)r_i + \alpha_i\sum_{j\in\din(i)}w_{ij}\tilde{x}^{t+1}_j \Big) \\
        &= \Pr\left( X^{t+2}_i=x^{t+2}_i \mid X^{t+1}_j=\tilde{x}^{t+1}_j \;\forall j\in\din(i) \right).
    \end{align*}
    That is, for any feasible state $\xvec^{t+1}$ such that $x^{t+1}_j=\tilde{x}^{t+1}_j$ for all principals $j$, its transition probabilities are the same as those of $\tilde{\xvec}^{t+1}$:
    \begin{align*}
        \Pr(\Xvec^{t+2} = \xvec^{t+2} \mid \Xvec^{t+1}=\tilde{\xvec}^{t+1})
        &= \prod_{i\in\mathcal{V}} \Pr\left( X^{t+2}_i=x^{t+2}_i \mid X^{t+1}_j=\tilde{x}^{t+1}_j=x^{t+1}_j \;\forall j\in\din(i) \right) \\
        &= \Pr(\Xvec^{t+2} = \xvec^{t+2} \mid \Xvec^{t+1}=\xvec^{t+1}) > 0.
    \end{align*}
    We can then arbitrarily set the states of pure obligees $i$ in $\xvec^{t+1}$ so that we also satisfy
    \(
        \Pr(\Xvec^{t+1} = \xvec^{t+1} \mid \Xvec^t = \xvec^t) > 0
    \)
    by letting $x^{t+1}_i=\tilde{x}^t_j$ for any $j\in\din(i)$.  Using \Cref{eqn:reachable_prop}, we can check that this indeed ensures a positive transition probability.
    Then $\xvec^{t+1}$ is a feasible state that satisfies both $\Pr(\Xvec^{t+2} = \xvec^{t+2} \mid \Xvec^{t+1} = \xvec^{t+1}) > 0$ and $\Pr(\Xvec^{t+1} = \xvec^{t+1} \mid \Xvec^{t} = \xvec^{t}) > 0$, and we have described a way to construct $\xvec^{t+1}$ such that
    \(
        \Pr(\Xvec^{t+2} = \xvec^{t+2}, \Xvec^{t+1}=\xvec^{t+1} \mid \Xvec^t = \xvec^t) > 0.
    \)
    Thus, we have that
    \(
        \Pr(\Xvec^{t+2} = \xvec^{t+2} \mid \Xvec^t = \xvec^t) > 0
    \)
    for any two feasible states, i.e. any two states can always reach each other in just two steps with positive probability.

    In summary, we have shown that most states can reach each other in one step, as long as states of principals and pure obligees are ``consistent" with each other as described by \Cref{eqn:reachable_prop}.
    Then this implies that the zero vector $\mathbf{0}$ is a feasible state with a self-loop, and that we can always construct an intermediate state $\xvec^{t+1}$ that allows us to traverse between any two feasible states in two steps.
    As a result, the Markov chain is ergodic and aperiodic, and therefore it converges to a unique stationary distribution.
    \end{proof}

\subsection{Extension of \cref{thm:mixing_time} to Time-Varying Graphs}
\label{app:time_varying_mixing}

We next extend the mixing-time result in \cref{thm:mixing_time} to contractor networks whose structure evolves over time.
Specifically, we show that our coupling-based contraction argument continues to hold under mild regularity conditions even when the contracting relationships change dynamically.

Formally, consider a \srsedit{fixed} sequence of directed graphs ${(G_t)}_{t \ge 0}$, where each $G_t = (\V, \E_t)$ represents the contracting relationships at time $t$ over a fixed set of nodes~$\V$.
The structure and interpretation of the model remain identical to \cref{sec:preliminary}, except that contractual ties (and their associated weights) may now vary across time.
We then consider the same stochastic process as in \cref{eq:sp_dynamics}, but now allowing both the failure propagation probabilities and edge weights to evolve with~$t$.
Each node~$i$ has a time-varying propagation probability $\alpha_i^t \in [0,1]$, and $w_{ij}^t$ denotes the fraction of~$i$’s projects subcontracted to principal~$j$ at time~$t$.
For each $t \in \mathbb{N}$, this gives rise to the matrices
\[
\A_t = \mathrm{diag}\{\alpha_i^t\}, \qquad \W_t = (w_{ij}^t)_{i,j \in \V},
\]
where $\W_t$ is row sub-stochastic by construction.
Accordingly, the time-varying stochastic process evolves as
\begin{equation}
\label{eq:time-varying_sp}
X_i^{t+1} \sim \Ber \left((1 - \alpha_i^t) r_i + \alpha_i^t \sum_{j \in \din^t(i)} w_{ij}^t X_j^{t}\right).
\end{equation}

In the static model of \cref{sec:preliminary}, principals satisfy $\alpha_i < 1$ while pure obligees have $\alpha_i = 1$.
In the time–varying setting, it is natural to preserve these across nodes at each step: project owners (obligees) do not suddenly begin subcontracting, and principals do not abruptly stop serving as upstream contractors.
Formally, we impose the following role–persistence condition, together with a time–uniform analogue of the static assumption on propagation probabilities.
\begin{assumption}
\label{ass:no_immediate_flip}
\srsedit{For each $t\ge 0$, define the set of pure obligees  \( O^t := \{ i\in V : \delta^{t}_{\out}(i)=\emptyset\}. \)
We assume the pure-obligee set is time-invariant: $O^{t+1}=O^{t}$ for all $t\ge 0$.
Equivalently, with $P^t:=V\setminus O^t$, we have $P^{t+1}=P^t$ for all $t\ge 0$.}
\end{assumption}

\Cref{ass:no_immediate_flip} reflects that pure obligees (e.g., municipal agencies or project owners) do not act as subcontractors from one period to the next, while principals continue to perform bonded work and may only evolve in their contractual connections.
Graphically, the assumption ensures that outgoing edges from a pure obligee do not appear between $t$ and $t+1$, and that nodes identified as principals retain at least one outgoing edge across time.
Consistent with this interpretation, pure obligees always satisfy $\alpha_i^t = 1$ for all $t \in \NN$.

In addition, we impose a uniform bound on the propagation parameters of all non-obligee nodes.

\begin{assumption}
\label{ass:alpha_uniform_bound}
There exists $\bar\alpha \in (0,1)$ such that, for every $t \in \NN$ and every node $i \in \V$ that is not a pure obligee at time $t$, we have $\alpha_i^t \le \bar\alpha$.
\end{assumption}

\Cref{ass:alpha_uniform_bound} generalizes the static assumption that all principals and intermediaries transmit failures with probability strictly less than one.
Intuitively, while the magnitude of $\alpha_i^t$ may vary over time (e.g., as contracting conditions change), it remains uniformly bounded.
Together, \Cref{ass:no_immediate_flip,ass:alpha_uniform_bound} ensure that the network retains its hierarchical structure across time and that the two-step contraction property continues to hold uniformly.

We now extend the mixing-time bound in \Cref{thm:mixing_time} to the time-varying setting.
In the time-varying regime, the process generally does not admit a stationary distribution, since the transition kernel changes with time.  Accordingly, mixing should be interpreted as the rate at which the process {\em forgets its initialization}: we bound the total variation distance between any two trajectories started from arbitrary initial conditions.  
Under \Cref{ass:no_immediate_flip,ass:alpha_uniform_bound}, the same coupling argument used for static graphs continues to ensure uniform geometric convergence, with a rate governed by the two-step contraction factor~$\bar\alpha$.
When the graph sequence ${G_t}_{t \ge 0}$ is time-homogeneous, this theorem exactly recovers \Cref{thm:mixing_time}, recovering the unique stationary law of the static process.

\begin{restatable}{theorem}{MixingGeneralTimeVarying}
\label{thm:mixing_time_tv}
Let $\{(\A_t,\W_t)\}_{t\ge0}$ be a sequence of time-varying contractor-network matrices defined as above satisfying \Cref{ass:no_immediate_flip,ass:alpha_uniform_bound}.

For any two initial states $\xvec, \yvec \in \{0,1\}^n$ let $\Pr(\Xvec^t \mid \xvec)$ and $\Pr(\Yvec^t \mid y)$ denote the distribution of two realizations of the stochastic process \cref{eq:time-varying_sp} started from $\Xvec^0 = \xvec$ and $\Yvec^0 = \yvec$ respectively.
Then for every $t\in\NN$,
\begin{equation}
\label{eq:tv_bound_tv}
    d_{TV}\!\bigl(\Pr(\Xvec^t \mid \xvec), \Pr(\Yvec^t \mid \yvec)\bigr)\ \le\ n\,\bigl\|\,\M_{t-1}\M_{t-2}\cdots \M_{0}\,\bigr\|_{\infty}.
\end{equation}
As a result, 
\begin{equation}
    \label{eq:tmix_tv_bound}
    d_{TV}\!\bigl(\Pr(\Xvec^t \mid \xvec), \Pr(\Yvec^t \mid \yvec)\bigr)\ \le n \bar\alpha^{\lfloor t/2 \rfloor}.
\end{equation}
\end{restatable}
\srsedit{
We adopt the convention that an empty product of matrices equals the identity matrix.
}
Before presenting the proof, we start with the following analog of \cref{lem:second_norm_bounded} that establishes that under \Cref{ass:no_immediate_flip,ass:alpha_uniform_bound}, the fixed graph two–step contraction property continues to hold.
\begin{lemma}
\label{lemma:two_step_contraction}
Under \Cref{ass:no_immediate_flip,ass:alpha_uniform_bound}, we have that
\begin{equation}
\label{eq:time-varying-contraction}
    \sup_{t\in\NN}\ \|\, \A_{t+1}\W_{t+1}\A_t\W_t\,\|_\infty \;\le\; \bar\alpha \;<\;1.
\end{equation}
\end{lemma}
\begin{proof}
For each $t$, the matrix $\M_t:=\A_t \W_t$ is entry-wise nonnegative and row sub–stochastic, hence $\|\M_t\|_\infty\le 1$. 
To obtain a strict contraction over two steps, fix $t\in\NN$ and a row index $i$. The $\ell_\infty$–induced norm equals the maximum row sum, so we estimate the $i$–th row sum of $\A_{t+1} \W_{t+1} \A_t \W_t$:
\begin{align*}
    \sum_{j} \bigl(\A_{t+1} \W_{t+1} \A_t \W_t\bigr)_{ij}
    &= \sum_{j}\sum_{k} \alpha_{i}^{t+1} w_{ik}^{t+1}\,\alpha_k^{t}\,w_{kj}^{t} \\
    &= \alpha_{i}^{t+1}\sum_{k} \alpha_k^{t} w_{ik}^{t+1} \underbrace{\sum_{j} w_{kj}^{t}}_{\le\,1}
    \;\le\; \sum_{k} \alpha_k^{t} w_{ik}^{t+1},
\end{align*}
where we used $\alpha_i^{t+1}\le 1$ in the last inequality. \srsedit{
If $w^{t+1}_{ik} > 0$, then $\delta^{t+1}_{\out}(k)\neq\emptyset$,
i.e., $k \in P^{t+1}$. By Assumption~\ref{ass:no_immediate_flip},
$P^{t+1}=P^t$, hence $k \in P^t$.
}
Therefore, by \Cref{ass:alpha_uniform_bound}, $\alpha_k^{t}\le \bar\alpha<1$. Hence
\[
    \sum_{k} \alpha_k^{t} w_{ik}^{t+1} \;\le\; \bar\alpha \sum_{k} w_{ik}^{t+1} \;\le\; \bar\alpha,
\]
and taking the maximum over $i$ gives
\[
    \|\A_{t+1} \W_{t+1} \A_t \W_t\|_\infty \;\le\; \bar\alpha \;<\;1.
\]
Finally, taking the supremum over $t$, on both sides, proves the claim.
\end{proof}

Using this lemma we can show \cref{thm:mixing_time_tv}.
\begin{proof}
At a high level, the derivation of \cref{eq:tv_bound_tv} mirrors the static case; we highlight the time-varying modifications and refer to the proof of \cref{thm:mixing_time} for omitted steps. 
We first construct a \emph{synchronous coupling} of the two trajectories $(\Xvec^t)_{t \in \NN}$ and $(\Yvec^t)_{t \in \NN}$ starting from $\Xvec^0 = \xvec$ and $\Yvec^0 = \yvec$ respectively, and show that for all $t\in\NN$,
\[
\Exp{\|\Xvec^t - \Yvec^t\|_1 \mid \xvec, \yvec} \;\le\; n\,\bigl\| \M_{t-1} \M_{t-2} \cdots \M_{0} \bigr\|_{\infty},
\qquad\text{where } \M_s \coloneqq \A_s \W_s.
\]
By the standard coupling inequality,
\begin{align*}
d_{TV}(\Pr(\Xvec^t \mid \xvec), \Pr(\Yvec^t \mid \yvec)) & \;\le\; \Pr(\Xvec^t \neq \Yvec^t \mid \xvec, \yvec)
\;=\; \Pr\!\bigl(\|\Xvec^t - \Yvec^t\|_1 \ge 1 \mid \xvec, \yvec \bigr) \\
& \;\le\; \Exp{\|\Xvec^t - \Yvec^t\|_1}
\;\le\; n\,\bigl\| \M_{t-1}\cdots \M_{0} \bigr\|_{\infty},
\end{align*}
as claimed.

We now formalize the coupling step. Write the dynamics as $\Xvec^{t+1} = f_{\thetavec^{t+1}}(\Xvec^t)$ where $\thetavec^{t} = (\theta_i^{t})_{i \in [n]}$ are independent uniform random variables, exactly as in the proof of \cref{thm:mixing_time}. For $i=1,\ldots,n$, define
\[
    D_i^{t} \;\triangleq\; \Exp{\abs{X_i^t - Y_i^t}},
    \qquad
    \Dvec^{t} \; \triangleq \; \bigl(D_1^{t},\ldots,D_n^{t}\bigr)^{\top}.
\]
By the same calculation as in \cref{lem:one_step_expected_difference} (with time index $t$ carried through), we have for all $t\ge0$,
\[
    \Dvec^{t+1} \;\leq\; \A_t \W_t\, \Dvec^t \;=\; \M_t\,\Dvec^t.
\]
Iterating this one-step domination yields
\begin{align*}
d_{TV}(\Pr(\Xvec^t \mid \xvec), \Pr(\Yvec^t \mid \yvec))
& \leq \Exp{\|\Xvec^t - \Yvec^t\|_1 \mid \xvec, \yvec}
= \|\Dvec^t\|_1
\underset{(\text{C.S.})}{\leq}\; n\, \|\Dvec^t\|_\infty \\
& \leq\; n\, \bigl\| \M_{t-1} \cdots \M_{0}\, \Dvec^0 \bigr\|_\infty 
\;\leq\; n\, \bigl\| \M_{t-1} \cdots \M_{0} \bigr\|_\infty \, \|\Dvec^0\|_{\infty}
\;\leq\; n\, \bigl\| \M_{t-1} \cdots \M_{0} \bigr\|_\infty,
\end{align*}
since $\|\Dvec^0\|_\infty = \norm{\xvec - \yvec}_{\infty} \le 1$. This is precisely the bound stated in \cref{eq:tv_bound_tv}.

Under \Cref{ass:no_immediate_flip,ass:alpha_uniform_bound}, \Cref{lemma:two_step_contraction} holds and gives a uniform two–step contraction. Hence, by submultiplicativity we have that for even steps, $t=2m$,
\[
\|\M_{t-1} \cdots \M_0\|_\infty
=\|( \M_{2m-1} \M_{2m-2})\cdots(\M_1 \M_0)\|_\infty
\ \le\ \bar\alpha^{\,m},
\]
and for odd $t=2m{+}1$,
\[
\|\M_{t-1} \cdots \M_0\|_\infty
=\|\M_{2m}(\M_{2m-1} \M_{2m-2})\cdots(\M_1\M_0)\|_\infty
\ \le\ \|\M_{2m}\|_\infty\,\bar\alpha^{\,m}\ \le\ \bar\alpha^{\,m}.
\]
Combining both cases gives
\[
\|\M_{t-1}\cdots \M_0\|_\infty \le \bar\alpha^{\lfloor t/2\rfloor}.
\]
Plugging this into the bound on $d_{TV}(\Pr(\Xvec^t \mid \xvec), \Pr(\Yvec^t \mid \yvec))$ yields the second result.
\end{proof}

\subsection{\cref{sec:mixing_time} Omitted Proofs}
\label{app:mixing_time}

\MixingDAG*
\begin{proof}
Let $\Xvec^t_S \triangleq (X^t_i)_{i\in S}$ be a random vector containing the entries of $\Xvec^t$ corresponding to nodes in $S\subset\mathcal{V}$.
By definition, the transition $\Pr(X_i^t=X_i^t \mid \Xvec^{t-1}=\xvec^{t-1})$ depends only on the state of in-neighbors (principals) in the previous step and is fixed for all $t>0$:
\[
    \Pr(X_i^t=1 \mid \Xvec^{t-1}=\xvec^{t-1}) = (1-\alpha_i)r_i + \alpha_i \sum_{j\in\din(i)} w_{ij} x_j^{t-1} = \Pr(X_i^t=1 \mid \Xvec^{t-1}_{\din(i)} = \xvec^{t-1}_{\din(i)}).
\]
Then from conditional independence we have that
\[
    \Pr(\Xvec^t=\xvec^t \mid \Xvec^{t-1}=\xvec^{t-1})
    = \prod_{i\in\mathcal{V}} \Pr(X_i^t=x_i^t \mid \Xvec^{t-1}_{\din(i)} = \xvec^{t-1}_{\din(i)})
    = \Pr(\Xvec^t=\xvec^t \mid \Xvec^{t-1}_{\Deltain^1}=\xvec^{t-1}_{\Deltain^1}),
\]
where $\Deltain^1$ (see \cref{def:in-layers}) denotes the subset of nodes that are in-neighbors (principals) of another node.
In other words, the joint distribution over $\Xvec^t$ can be determined without knowing the previous states of nodes not in $\Deltain^1$, i.e., pure obligees.

Similarly, conditional independence implies that the default probabilities of nodes in $\Deltain^k$ depend only on the previous states of their principals, which are contained in $\Deltain^{k+1}$:
\[
    \Pr(\Xvec^t_{\Deltain^k}=\xvec^t_{\Deltain^k} \mid \Xvec^{t-1}=\xvec^{t-1})
    = \prod_{i\in\Deltain^k} \Pr(X_i^t=x_i^t \mid \Xvec^{t-1}_{\din(i)} = \xvec^{t-1}_{\din(i)})
    = \Pr(\Xvec^t_{\Deltain^k}=\xvec^t_{\Deltain^k} \mid \Xvec^{t-1}_{\Deltain^{k+1}}=\xvec^{t-1}_{\Deltain^{k+1}}).
\]
Then the law of total probability implies that
\begin{align*}
    \Pr(\Xvec^t=\xvec^t \mid \Xvec^{t-2}=\xvec^{t-2})
    &= \sum_{\xvec^{t-1}_{\Deltain^1}} \Pr(\Xvec^t=\xvec^t \mid \Xvec^{t-1}_{\Deltain^1}=\xvec^{t-1}_{\Deltain^1}) \Pr(\Xvec^{t-1}_{\Deltain^1}=\xvec^{t-1}_{\Deltain^1} \mid \Xvec^{t-2}=\xvec^{t-2}) \\
    &= \sum_{\xvec^{t-1}_{\Deltain^1}} \Pr(\Xvec^t=\xvec^t \mid \Xvec^{t-1}_{\Deltain^1}=\xvec^{t-1}_{\Deltain^1}) \Pr(\Xvec^{t-1}_{\Deltain^1}=\xvec^{t-1}_{\Deltain^1} \mid \Xvec^{t-2}_{\Deltain^{2}}=\xvec^{t-2}_{\Deltain^{2}}) \\
    &= \Pr( \Xvec^t=\xvec^t \mid \Xvec^{t-2}_{\Deltain^2}=\xvec^{t-2}_{\Deltain^2} ).
\end{align*}
Repeating this argument, we eventually get that:
\begin{align*}
    \Pr(\Xvec^t=\xvec^t \mid \Xvec^{t-d}=\xvec^{t-d})
    &= \sum_{\xvec^{t-1}_{\Deltain^1}} \cdots \sum_{\xvec^{t-d+1}_{\Deltain^{d-1}}}
    \Pr(\Xvec^t=\xvec^t \mid \Xvec^{t-1}_{\Deltain^1}=\xvec^{t-1}_{\Deltain^1})
    \prod_{k=1}^{d-1} \Pr(\Xvec^{t-k}_{\Deltain^k}=\xvec^{t-k}_{\Deltain^k} \mid \Xvec^{t-k-1}_{\Deltain^{k+1}}=\xvec^{t-k-1}_{\Deltain^{k+1}}) \\
    &= \Pr(\Xvec^t=\xvec^t \mid \Xvec^{t-d}_{\Deltain^d} = \xvec^{t-d}_{\Deltain^d}),
\end{align*}
where $d$ is the longest directed path length. We note that this gives a special case of the Chapman--Kolmogorov equation as derived by \cite{Kolmogoroff1931} and \cite{chapman1928brownian}, which describes a general relation between joint probabilities in stochastic processes.
Note that any $d$-length path must start from a pure principal, so $\Deltain^d$ is contained in the set of pure principals who are unaffected by the network and always fail idiosyncratically according to their inherent risk score in $\rvec$.
Thus, $\Pr(\Xvec^{t-d}_{\Deltain^d}=\xvec_{\Deltain^d}) = \Pr(\Xvec^0_{\Deltain^d}=\xvec_{\Deltain^d})$ is fixed for all $t\geq d$.
Then it follows that for all $t\geq d$,
\begin{equation}
\label{eq:pi_dag_def}
    \Pr(\Xvec^t=\xvec^t)
    = \sum_{\xvec_{\Deltain^d}} \Pr(\Xvec^t=\xvec^t \mid \Xvec^{t-d}_{\Deltain^d}=\xvec_{\Deltain^d}) \Pr(\Xvec^{t-d}_{\Deltain^d}=\xvec_{\Deltain^d})
    = \Pr(\Xvec^d=\xvec^t) = \pi(\xvec^t).
\end{equation}
That is, when the contractor network is acyclic and has a maximum depth of $d$, the joint distribution converges to its stationary distribution $\pi(\cdot)$ in at most $d$ steps.
\end{proof}
\section{Computational Experiments: Additional Details}
\label{sec:app_simulations}

\subsection{Construction of Anonymous Network}
\label{app:simulations_network_construction}

In our simulations we use real contract and firm data from our partnering surety organization.  For privacy reasons, we do not use the original data directly. Instead, we construct a replica network that preserves key structural and distributional properties of the empirical network while protecting sensitive information.  \cref{sec:sims_network_description} provides a high level description of the surety network, with the network data included in the attached code details.

Our empirical dataset includes all bonded contracts that were active at any time in calendar year 2018 (i.e. with start date before December 31, 2018 and end date after January 1, 2018).  The graph $G$ contains the set of nodes as contracting organizations which had at least one active contract within that time period.  Each organization $i$ has an observed value $\revenue(i)$ for the total contract volume of obligee $i$ within that year, as observed in the surety’s dataset.  Contracts are dictated by a directed edge $e = (j, i)$ between an obligee $i$ and principal $j$, with a corresponding bond amount $\bond(i,j)$.  Note that in the event that $i$ and $j$ have {\em multiple} contracts, then $\bond(i,j)$ corresponds to the cumulative bond amount across all contracts active within that one year period.  In order to normalize the scale, we construct the edge weights $w_{ij}$ via:
\[
w_{ij} = \frac{\bond(i,j)}{\revenue(i)}.
\]
Note that in the event $\sum_{j \in \din(i)} \bond(i,j) \neq \revenue(i)$ then this node has {\em unobserved contracts}. We present a methodology to handle this case in \cref{app:simulations_unobserved_edges}.  For each organization in the network, we also have their idionyncratic risk score (default probability) $r_i$ and loss amount suffered by the surety if it defaults $\beta_i$. To ensure anonymity, node labels are discarded and replaced with indices under a fixed ordering.

Based on these primitives, we generate a noisy observation of this empirical network that conceals sensitive contractor information while retaining essential features of the edge structure, the distributions of $\beta_i$, $r_i$, and the relationship between network position and node characteristics based on the framework of differential privacy~\citep{dwork2006calibrating}. The construction proceeds as follows. Beginning with the set of pure principals (nodes at depth $\tau=0$), we rewire edges iteratively across depths: 
\begin{enumerate}[label=(\roman*)]
    \item Each node at depth $\tau$ maintains at least one outgoing edge to a node at depth $\tau+1$.
    \item Each node at depth $\tau+1$ with $k$ unmatched in-edges is assigned $k$ incoming edges from nodes at depth $\tau$.
    \item Any remaining unmatched out-edges from nodes at depth $\tau$ are randomly connected to unmatched in-edges of nodes at later depths.
\end{enumerate}
This process continues until depth $d-1$. Nodes at depth $d$ are pure obligees and therefore have no outgoing edges. Conditions (i)–(iii) guarantee that all nodes are matched, degree distributions are preserved, and nodes remain at their original depth, ensuring that principals and obligees retain their types.

To further protect privacy, we apply the Laplace mechanism with scale calibrated to the global sensitivity of each statistic (cf. \citealt{dwork2006calibrating}): each $r_i$ and $\beta_i$ is re-scaled and then perturbed with independent Laplace noise.  Finally, we also redefine the edge weights in our replica network.  The bond amounts of each contract (edge) in the original network also pass through the same re-scaling and Laplace mechanism as the node features, and each re-scaled amount is associated to the contract obligee.  After rewiring, these bond amounts are then arbitrarily reassigned to the rewired in-edges of the obligee and normalized so that the sum of incoming weights to each obligee is one.

This procedure yields a replica network that preserves several important characteristics of the empirical data, the shapes of the distributions of $r_i, \beta_i$ and $w_{ij}$; the way obligees distribute contracts among principals; and the relationship between a node’s position in the network and its features. For example, an obligee that originally allocates most projects to a single principal behaves similarly in the replica, and a high-risk contractor embedded in a long path of intermediaries retains that role.

\subsection{Accounting for Unobserved Edges}
\label{app:simulations_unobserved_edges}

In practice, an obligee's reported revenue $\revenue(i)$ may exceed the value of contracts bonded by the surety, since competing surety organizations may also underwrite a portion of their contracting principal's work.  In this case, we observe that $\revenue(i) > \sum_{j \in \din(i)} \bond(i,j)$ and the proposed methodology will result in a node whose weighted in-degree ($\sum_{j \in \din{i}} w_{ij}$) is strictly less than one.  To account for such unobserved contracts (aka unobserved edges) we present a methodology to impute the graph under reasonable assumption on each obligees contracting behavior.

For each obligee $i$ such that $\revenue(i) > \sum_{j \in \din(i)} \bond(i,j)$, we define a dummy principal with baseline risk $\rout{i}$ and a single outgoing edge weighted as:
\[w_{i,\out} = 1 - \sum_{k\in\din(i)} w_{ik}
\]
to represent all unobserved principals.  Note that by using this dummy variable we rely on the implicit assumption that the subcontractors of $i$ and $k$ do not contract with each other, so that $\rout{i}$ and $\rout{k}$ can be expressed as independent risk scores.
In order to estimate $\rout{i}$ we make the following assumption.

\begin{assumption}
    Suppose each contractor $j$ is assigned one of $T$ product segment types $\ctype(j) \in \{\tau_1,\dots,\tau_T\}$.
    Furthermore, let $\din(i)$ denote the set of all of $i$'s contractors, including unobserved contractors.
    We assume that for any type $\tau_k$,
    \begin{equation*}
        \frac{ \sum_{\{j\in\din(i)\,:\,\ctype(j)=\tau_k\}} w_{ij} }
        { \sum_{\{j\in\din(i)\}} w_{ij} }
        = 
        \frac{ \sum_{\{j\in\mathcal{V}\setminus\din(i)\,:\,\ctype(j)=\tau_k\}} w_{ij} }
        { \sum_{\{j\in\mathcal{V}\setminus\din(i)\}} w_{ij} }.
    \end{equation*}
\end{assumption}
In other words, the fraction of $i$'s obligations to subcontractors  $j$ of type $\tau_k$ in the unobserved network is the same as the fraction in our observed network.
Given this assumption, we define $\rout{i}$ to be a convex combination of the median baseline risk scores $\bar{r}_j$ associated with each product segment type, weighted by how often obligee $i$ works with contractors of the same type in our observed network.  In particular, we set:
\begin{equation*}
    \rout{i} = \sum_{j\in\din(i)} \frac{w_{ij}}{\sum_{j\in\din(i)} w_{ij}} \bar{r}_j.
\end{equation*}

\subsection{Real-World Network Satisfying \Cref{ass:larger_neighbors}}
\label{app:sims_not_satisfying_assumption}

In \Cref{sec:mean_field}, we established that under \Cref{ass:larger_neighbors}, that is, when each intermediary $i$ satisfies $\sum_{j \in \din(i)} w_{ij} r_j \geq r_i$, the limiting failure probabilities $m_i$ are guaranteed to exceed their idiosyncratic risk levels $r_i$.  While this condition is sufficient to ensure higher expected losses, it is not necessary. In practice, we observe that the condition holds for many but not all intermediaries, yet the expected aggregate loss $\Exp{\GAR(\Xvec^\infty)}$ is still larger than $\Exp{\GAR(\Xvec^0)}$.  Indeed, as showin in \Cref{fig:r_vs_Wr}, a number of intermediaries in our empirical network violate the assumption, but the net effect of those satisfying the inequality dominates, producing greeater expected losses and amplified tail risk.

\begin{figure}
    \centering
    \begin{subfigure}[t]{0.48\linewidth}
        \centering
        \includegraphics[width=\linewidth]{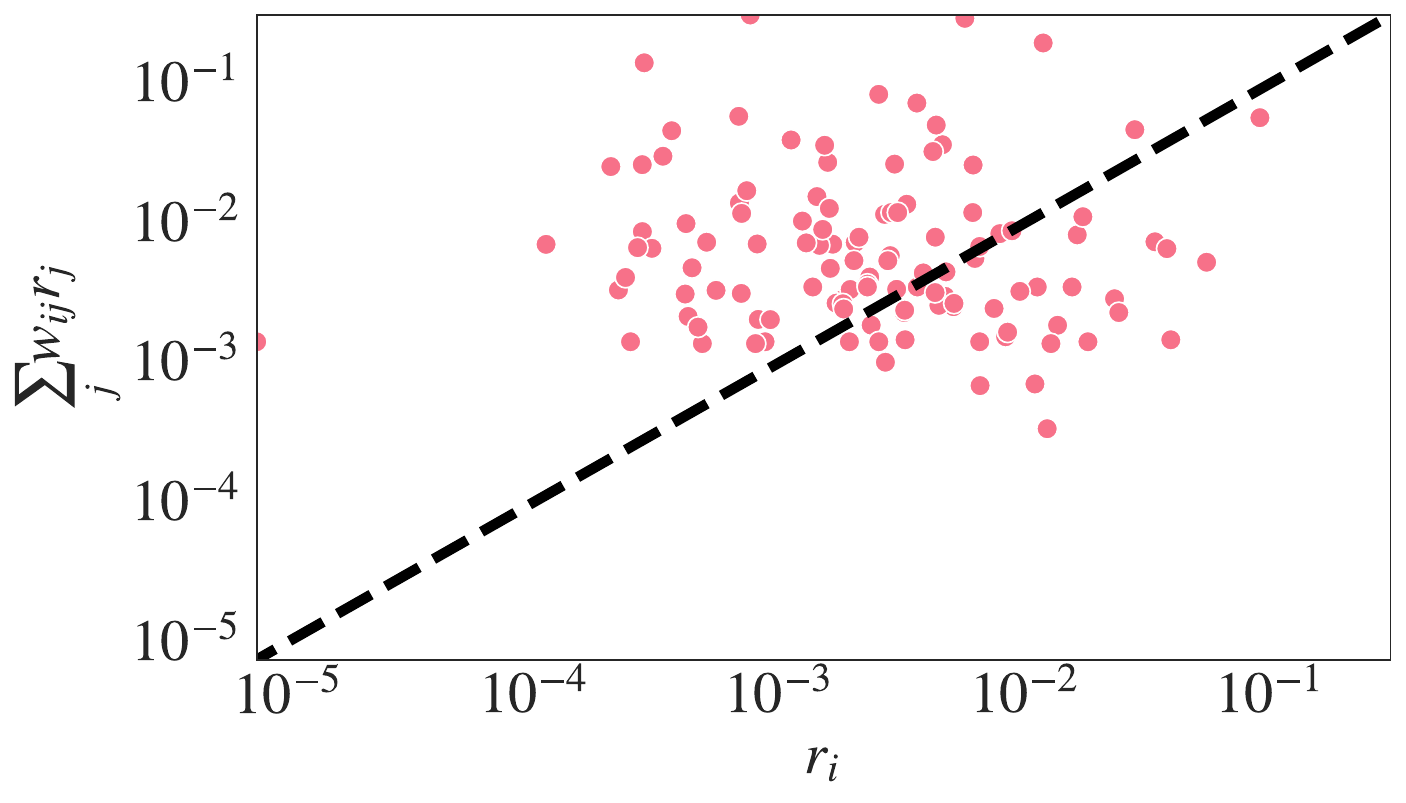}
        \caption{$r_i$ vs $\sum_j w_{ij} r_j$}
        \label{fig:r_vs_Wr}
    \end{subfigure}
    \hfill
    \begin{subfigure}[t]{0.48\linewidth}
        \centering
        \includegraphics[width=\linewidth]{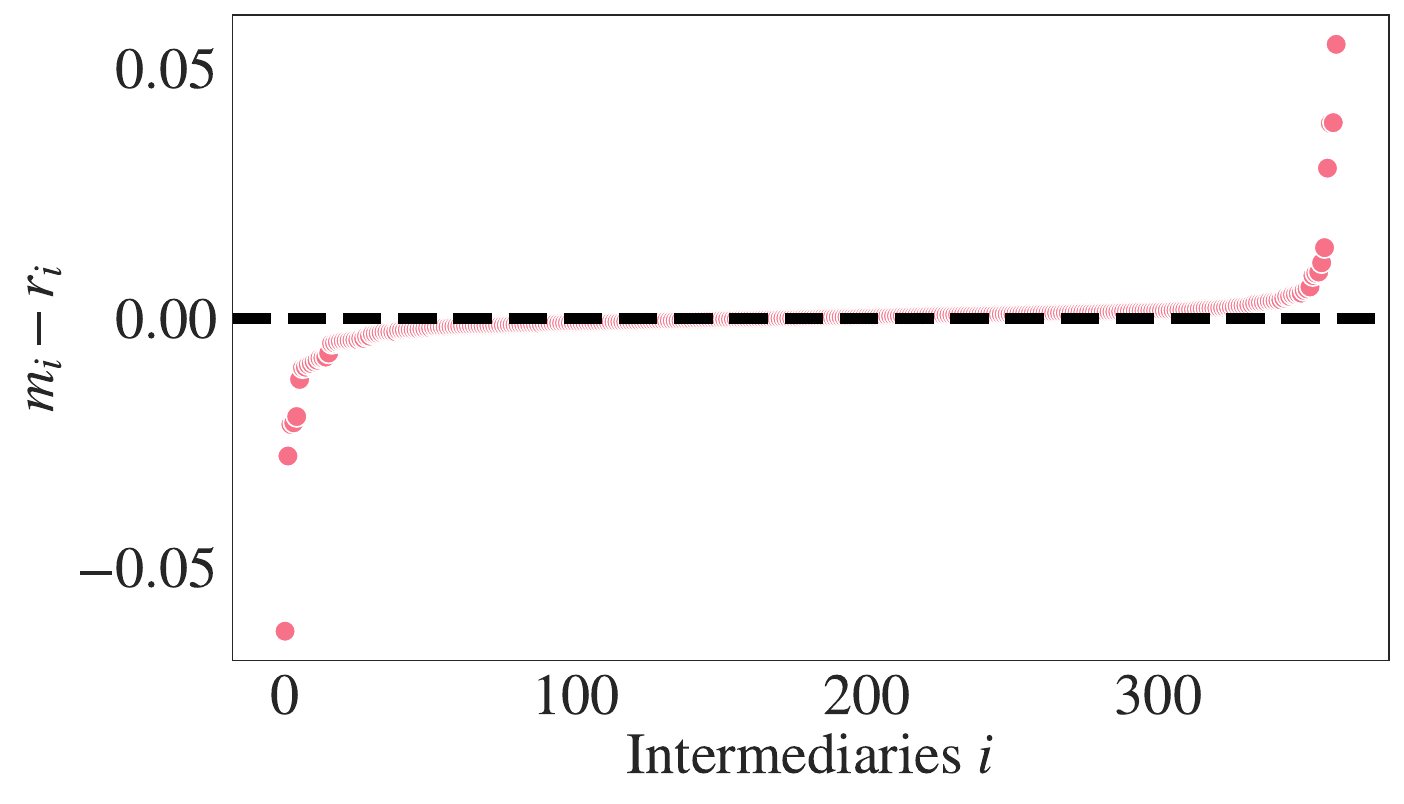}
        \caption{$m_i - r_i$}
        \label{fig:sufficient_condition}
    \end{subfigure}
    \caption{(Left) The idiosyncratic default probabilities $r_i$ for intermediaries $i$ as compared to the weighted average of their principals' idiosyncratic default probabilities. Several points fall below the dashed $y=x$ line, indicating that for these intermediaries, \Cref{ass:larger_neighbors} is not satisfied. Note that plot is shown on a logarithmic scale. (Right) Values $m_i-r_i = [(\I-\A\W)^{-1}\A(\W-\I)\rvec]_i$ for intermediaries $i$, sorted in ascending order. Most points lie above the $x$-axis, indicating that the limiting failure probabilities of these intermediaries satisfy $m_i > r_i$.}
    \label{fig:r_vs_Wr_and_m_minus_r}
\end{figure}

To further formalize this observation, we derive an alternative sufficient condition for $m_i \geq r_i$. Starting from the closed-form expression
\[
\meanvec = (\I - \A \W)^{-1} (\I - \A) \rvec,
\]
we compute
\[
\meanvec - \rvec = (\I - \A \W)^{-1}\A(\W - \I)\rvec.
\]
Thus, for each intermediary $i$, we have $m_i \geq r_i$ if and only if
\[
\bigl[(\I - \A \W)^{-1} \A (\W - \I)\rvec\bigr]_i \geq 0.
\]
Empirically, this condition is satisfied for the majority of intermediaries in our dataset (see \Cref{fig:sufficient_condition}), meaning that their network-adjusted failure probabilities exceed their idiosyncratic risks.

Taken together, these results show that although \Cref{ass:larger_neighbors} does not universally hold, a sufficient mass of intermediaries nonetheless experience risk amplification ($m_i \geq r_i$). Consequently, as summarized in \Cref{fig:loss_box_plot}, the expected aggregate loss $\Exp{\GAR(\Xvec^\infty)}$ strictly exceeds its independent-failure counterpart, consistent with the heavier right tails observed in the empirical loss distribution.
\end{APPENDICES}








\end{document}